\documentclass{article}
\usepackage{arxiv}

\usepackage[dvipsnames]{xcolor}
\usepackage[utf8]{inputenc}
\usepackage[T1]{fontenc}
\usepackage{mathtools,booktabs,longtable,tabularx}
\usepackage{hyperref}
\usepackage{tikz}
\usepackage{enumerate}
\usepackage[most]{tcolorbox}
\usepackage{enumitem}
\usepackage{amsmath}
\usepackage{amsthm}
\usepackage{amssymb}

\definecolor{AlgLightBlue}{rgb}{0.87843,0.92941,0.96863}
\definecolor{AlgMidBlue}{rgb}{0.36864,0.58430,0.75687}
\definecolor{AlgDarkBlue}{rgb}{0.14119,0.29020,0.40000}
\definecolor{AlgTeal}{rgb}{0,0.5,0.5}
\newtcolorbox{algobox}[1]{enhanced,breakable,colback=AlgLightBlue,colframe=AlgMidBlue,boxrule=.7pt,arc=1.2pt,left=2mm,right=2mm,top=1.4mm,bottom=1.4mm,title={#1},coltitle=white,colbacktitle=AlgDarkBlue,fonttitle=\bfseries}
\newcommand{\algio}[2]{{\textbf{Input:}} #1\par\smallskip{\textbf{Output:}} #2}
\newif\ifshowchanges
\showchangesfalse
\DeclareRobustCommand{\rev}[1]{#1}

\usetikzlibrary{arrows.meta,positioning,calc}
\allowdisplaybreaks[2]
\newtheorem{external}{External result}

\newcommand{\Kill}{\operatorname{Kill}}
\newcommand{\IRV}{\operatorname{IRV}}
\newcommand{\best}{\operatorname{best}}
\newcommand{\id}{\operatorname{id}}
\newcommand{\score}{\operatorname{score}}
\newcommand{\cell}{\mathcal M}
\newcommand{\cl}{\operatorname{cl}}
\newcommand{\im}{\operatorname{im}}
\newcommand{\tw}{\operatorname{tw}}
\newcommand{\bw}{\operatorname{bw}}
\newcommand{\apx}{\operatorname{apx}}
\newcommand{\ecc}{\operatorname{ecc}}
\newcommand{\diam}{\operatorname{diam}}
\newcommand{\rad}{\operatorname{rad}}
\newcommand{\Int}{\operatorname{Int}}
\newcommand{\Rep}{\operatorname{Rep}}
\newcommand{\lex}{\mathrm{lex}}
\newcommand{\Pclass}{\mathsf P}
\newcommand{\NP}{\mathsf{NP}}
\newcommand{\coNP}{\mathsf{coNP}}
\newcommand{\FP}{\mathsf{FP}}
\newcommand{\ZoneProblem}{\textnormal{\textsc{IRV-Exclusion}}}
\newcommand{\MinZone}{\textnormal{\textsc{Min-IRV-Exclusion}}}
\newcommand{\yes}{\mathsf{true}}
\newcommand{\no}{\mathsf{false}}
\newcommand{\ind}[1]{\mathbf1[#1]}
\newcommand{\minlex}{\mathop{\min}_{\lex}}
\newcommand{\Out}{\mathrm{OUT}}
\newcommand{\Exposed}{\mathrm{EXPOSED}}
\newcommand{\Hidden}{\mathrm{HIDDEN}}
\newcommand{\Areal}{A_{\mathrm{real}}}

\newcommand{\ed}[1]{{ #1}}

\newtheorem{definition}{Definition}
\newtheorem{lemma}{Lemma}

\newtheorem{theorem}{Theorem}
\newtheorem{proposition}{Proposition}
\newtheorem{corollary}{Corollary}

\newtheorem{example}{Example}

\title{\rev{A Tractability Boundary for Exclusion Zones in Instant-Runoff Voting: Planar Versus One-Apex Graphs}}

\author{
 Georgios Birmpas \\
  University of Liverpool\\
  United Kingdom \\
   \And
 Georgios Chionas \\
  University of Liverpool\\
  United Kingdom \\
  \And
 Efthyvoulos Drousiotis \\
  University of Liverpool\\
  United Kingdom \\
\AND
 Soodeh Habibi \\
  University of Liverpool \\
  United Kingdom \\
\And
 Marios Mavronicolas \\
University of Cyprus \\
    Cyprus  \\ 
\And
 Paul Spirakis \\
  University of Liverpool \\
  United Kingdom \\
}

\begin{document}
\maketitle
\begin{abstract}
We study exclusion zones for instant-runoff voting (IRV) on connected, unweighted graphs. Each vertex hosts one voter and may host a candidate. Voters rank candidates by shortest-path distance and cast one whole vote in each round. Distance ties favour the smaller identifier in a common fixed vertex ordering, and equal-score elimination ties remove the larger identifier in that ordering. A nonempty vertex set $S$ is an \emph{exclusion zone} if every election with at least one participating candidate in $S$ elects a winner in $S$.

In this setting, we consider three computational problems: deciding whether a target candidate can lose against some nonempty subset of a given set of allowed opponents, deciding whether a given vertex subset is an exclusion zone, and computing the unique minimum exclusion zone. We establish that all three are polynomial-time solvable on trees, outerplanar graphs and, more generally, planar graphs. For an $n$-vertex tree, the first problem (candidate loss) can be decided in $O(n)$ time. The planar algorithm uses the structure of an inclusion-minimal candidate set that eliminates the target in the first round: contracting the candidates' voter cells makes the target cell adjacent to every other cell, leading to bounded-width structure and enumerable separators. We give explicit implementations and conservative polynomial running-time and space bounds; the large exponents for outerplanar and planar graphs establish tractability rather than practical efficiency.

In contrast, even on bipartite exact one-apex graphs of bounded diameter with every nontarget vertex eligible to stand, candidate loss is $\NP$-complete, singleton-zone verification is $\coNP$-complete, and the minimum exclusion zone cannot be computed in polynomial time unless $\Pclass=\coNP$. Thus planarity marks a tractability boundary for these computational problems under the voting and tie-breaking conventions studied here.
\end{abstract}

\medskip

\bigskip

\bigskip

\noindent
{\bf Use of AI tools.} 
During manuscript preparation, the authors used Perplexity, a generative-AI assistant, for drafting and revising explanatory passages, manuscript organization and compression, discussion of mathematical formulations and structural arguments, literature searches and comparisons with prior work, small computational checks of voting examples, and LaTeX editing. The assistance included suggestions for references and a dependency roadmap. AI-generated suggestions were treated as provisional and were reviewed and revised by the authors. The authors retain responsibility for the final manuscript, including its mathematical claims, algorithms, proofs, references, and presentation.

\medskip

\bigskip

\bigskip

\clearpage
\section{Introduction}\label{m:introduction}

\subsection{Framework and Computational Questions}

We study instant-runoff voting (IRV) on a fixed, connected, unweighted graph. Every vertex hosts one voter and is a possible candidate location. Voters rank candidates by shortest-path distance. IRV repeatedly eliminates a candidate with the fewest first-place votes and elects the last survivor. A common identifier order resolves equal distances, and a fixed rule resolves elimination ties; Section~\ref{sec:model} specifies both.

Computing the winner for a fixed candidate set is polynomial-time (Lemma~\ref{lem:evaluate}), but does not answer a participation-robustness question: can candidates outside a specified region ever win when at least one candidate from that region stands? An \emph{exclusion zone} is a nonempty vertex set $S$ such that every election with a participating candidate in $S$ elects a winner in $S$. The guarantee concerns the winner's membership in $S$, not its identity, and does not restrict who may participate. The minimum zone is the smallest set with this guarantee, not a claim about fairness or welfare. In the sequential-facility interpretation, customers use their nearest remaining facility and the facility with the fewest customers closes~\cite{ExclusionZones}; a zone guarantees that the last facility lies in it whenever a facility there participates.

The guarantee depends on the voting convention.
Tomlinson et al.~\cite{ExclusionZones} split votes equally
among nearest candidates and require the zone guarantee
over every resolution of elimination ties.
We adopt a different convention: we assign each voter
one whole vote and fix both tie rules.
\emph{Changing the vote-assignment and elimination-tie
conventions need not preserve exclusion-zone guarantees
on an unchanged graph; transferring results between the
models therefore requires an explicit compatibility argument.}
The path example in Section~\ref{sec:model}, verified in
Section~\ref{app:common}, demonstrates this distinction
between their convention and ours.
Although some structural arguments are shared, our
computational results are stated and proved under the
exact convention specified in Section~\ref{sec:model}.

A direct verification procedure would examine exponentially many candidate configurations. We ask whether graph geometry makes this universal guarantee tractable, and how little departure from planarity suffices for hardness. We also seek the unique minimum exclusion zone; nesting gives its existence (Lemma~\ref{lem:nesting}). Our common algorithmic subproblem is \emph{candidate loss}: can a mandatory target lose against some nonempty subset of allowed opponents? Whenever the target can lose, some nonempty subset of allowed opponents eliminates it in the first round; we call such a subset a \emph{first-round witness}. Candidate-loss tests then support zone verification and minimum-zone computation. Section~\ref{sec:model} proves these reductions.

Our graph-specific results cover trees, outerplanar graphs and planar graphs, and contrast them with exact one-apex graphs: nonplanar graphs made planar by deleting one vertex~\cite{RobertsonSeymourGMX}. On a tree rooted at the target, the voters preferring an opponent to the target form a rooted subtree; these regions are nested or disjoint, enabling the linear-time algorithm. For a fixed participating set, a candidate's \emph{voter cell} consists of its supporters. The \emph{cell quotient} replaces each cell by one vertex, with adjacency when an original edge joins two cells. For planar graphs, the method exploits the bounded width of the quotient of an unknown inclusion-minimal first-round witness, not bounded width of the input graph. Section~\ref{m:planar} develops this distinction; the contributions and techniques are summarized next.

\subsection{Contributions and Techniques}

Under the whole-vote, deterministic-tie model, we give polynomial-time algorithms on trees, outerplanar graphs and planar graphs for three tasks: deciding whether a target can lose against some nonempty subset of allowed opponents; verifying whether a supplied vertex set is an exclusion zone; and computing the unique minimum exclusion zone.

On trees, candidate loss admits a linear-time algorithm; Table~\ref{tab:intro-complexity} gives the resulting zone bounds.
For the planar algorithm's completeness proof, fix a planar input graph and choose a first-round witness none of whose nonempty proper subsets is also a first-round witness. Deleting any chosen opponent must increase the target's score. Using shortest paths, we show that every opponent's cell shares an edge with the target's cell. Contracting the connected cells gives a planar quotient in which the target is adjacent to every other vertex.

Removing the target vertex from this quotient leaves an outerplanar graph, giving bounded-width structure even when the input graph itself has unbounded width. The participating candidates and their quotient are unknown, so the algorithm instead enumerates polynomially many boundary descriptions; each has a constant number of candidate identities although the boundary may be long. A dynamic program (DP) combines the regions, counts each voter once, tracks the target separately, and finalizes an opponent's score only after proving that no voter outside the processed region can support it in the represented election. Its final comparison certifies a genuine first-round defeat; Section~\ref{m:planar} gives the construction and correctness argument.

By contrast, we show that candidate loss is $\NP$-complete on connected, bipartite exact one-apex graphs even when deleting the target makes the graph planar and every other vertex is an allowed opponent. Singleton-zone verification is $\coNP$-complete on the same restricted class, and the minimum zone cannot be computed in polynomial time unless $\Pclass=\coNP$. The allowed opponents need not all participate.

Table~\ref{tab:intro-complexity} reports conservative bounds for the analysed implementations, with $n=|V(G)|$. Each time bound includes the preprocessing required by that task; a single tree loss query does not construct an all-pairs distance matrix. The $\Kill$-space column concerns one loss query, not the complete transfer procedure; ``Minimum/list'' means computing the minimum zone or enumerating all nonempty zones (at most $n$).
These are upper bounds, not claims of optimality or practical performance. The large exponents arise from checking many combinations of candidate identities and vote-count records, together with separator descriptions in the planar case. The tree algorithm uses a different method; so its linear running time does not directly improve those bounds.
\begin{table}[htbp]
\centering\small
\caption{Bounds in the whole-vote, deterministic-tie model. Exact-one-apex hardness holds already on the restricted subclass of Theorem~\ref{m:apextheorem}.}
\label{tab:intro-complexity}
\begin{tabular}{@{}lcccc@{}}
\toprule
Graph class & $\Kill$ time & $\Kill$ space & Verify time & Minimum/list time\\
\midrule
Trees & $O(n)$ & $O(n)$ & $O(n^2)$ & $O(n^3)$\\
Outerplanar & $O(n^{32})$ & $O(n^{15})$ & $O(n^{33})$ & $O(n^{34})$\\
Planar & $O(n^{76})$ & $O(n^{21})$ & $O(n^{77})$ & $O(n^{78})$\\
Exact one-apex & $\NP$-complete & --- & $\coNP$-complete & not in $\FP^{\dagger}$\\
\bottomrule
\end{tabular}
\par\smallskip
\raggedright\footnotesize $^{\dagger}$ $\FP$ denotes functions computable in deterministic polynomial time. The lower bound holds unless $\Pclass=\coNP$; it is a conditional function lower bound.
Verification hardness already holds for singleton zones.
\end{table}

\subsection{Related Work}

Tomlinson et al.~\cite{Moderation} introduced exclusion zones for IRV with distance-based preferences on a line. In one setting with a uniform voter distribution on $[0,1]$, the central interval $[1/6,5/6]$ is an exclusion zone. Their later paper proves nesting of exclusion zones, shows 
{\sf coNP}-completeness of verifying a proposed zone in their graph model, and studies randomized approximations and minimum zones for particular graph families~\cite{ExclusionZones}.
Tomlinson et al.~\cite{ExclusionZones} use equal splitting among nearest candidates and require the exclusion-zone guarantee for every resolution of elimination ties. We instead assign whole votes using a common candidate order and fix the elimination rule. Both differences matter: agreement of vote assignments alone does not imply agreement of election outcomes. Their structural descriptions for selected graph families and our algorithmic results therefore concern different voting conventions. Our algorithms cover all trees and, more generally, planar graphs under the convention specified here; they do not automatically subsume those structural descriptions.
Tomlinson et al.~\cite{ExclusionZones} also use a pairwise-loss graph and directed reachability to identify positions forced into the minimum zone. Under our deterministic convention, the pairwise-loss graph is a tournament. We use its singleton-closure property to reduce the exact minimum-zone search to at most $n$ sets, each verified against larger elections.

Elkind et al.~\cite{PreferenceRestrictions} survey how restrictions on preference rankings affect voting problems. Our restrictions arise from shortest-path distances, and we ask how the underlying graph affects exclusion-zone computation. Marx and Pilipczuk~\cite{MarxPilipczuk} use boundaries between unknown nearest-facility regions in algorithms for planar facility location. Our planar method faces an analogous challenge: neither the participating candidates nor their voter cells are known in advance. For the geometric part of our proof, we use the bridge-aware sphere-cut formulation discussed by Bl\"asius et al.~\cite{Blasius}.

\section{Voting Model and Structural Foundations}\label{sec:model}\label{app:common}

We use the graph-metric exclusion-zone setting of Tomlinson et al.~\cite{ExclusionZones}, with the deterministic whole-vote convention specified below. We state and prove the structural foundations needed under these conventions. Some arguments adapt established ideas, particularly nesting and pairwise-loss closure; the candidate-loss formulation and reduction bounds below specify the interfaces used by our algorithms.
 Let $G=(V,E)$ be a finite, simple, connected, undirected, unweighted graph, with $n=|V|\ge1$. Each vertex hosts one voter and is a possible candidate location. The input specifies a bijection $\id:V\to\{0,\ldots,n-1\}$ shared by all voters; distances $d_G$ are computed in $G$. The graph, its voters and the identifier order remain fixed throughout every election and across candidate sets. Eliminating a candidate removes only its candidacy, not its vertex or voter. We suppress the identifier order in the notation for election outcomes and computational problems.

\begin{definition}[Voter preference and IRV]\label{def:scores}\label{def:irv}
Voter $v$ ranks candidates by the lexicographic key $\kappa_v(c)=(d_G(v,c),\id(c))$. For a nonempty active candidate set $K$, let
$\best_v(K):=\arg\min_{c\in K,\lex}\kappa_v(c)$.
Every voter gives one whole vote to this unique candidate. For each active candidate $c\in K$, define its cell $\cell_K(c)=\{v\in V:\best_v(K)=c\}$, score $\score_K(c)=|\cell_K(c)|$, and elimination key
$\lambda_K(c)=(\score_K(c),-\id(c))$.
While at least two candidates remain, IRV eliminates the one with the smallest elimination key and recomputes every vote for the remaining set. Thus equal-distance ties favour smaller identifiers, while equal-score elimination ties remove larger identifiers. Write $\IRV(G,K)$ for the survivor when this process starts from candidate set $K$. A singleton wins without an elimination round.
\end{definition}

\begin{lemma}\label{lem:evaluate}
For a supplied active set, its first-round scores, complete elimination order and winner can be computed in polynomial time.
\end{lemma}
\begin{proof}
Run breadth-first search from each of the $n$ vertices. This computes all distances in $O(n(n+|E|))$ elementary graph operations. In a round containing $k$ candidates, scan all $k$ keys for each of the $n$ voters and increment the winning candidate's counter. This uses $O(nk)$ comparisons. Scan the $k$ elimination keys to remove the smallest. There are at most $n-1$ rounds, so the additional work is bounded by
\[
 O\!\left(n\sum_{k=2}^{n}k\right)=O(n^3).
\]
All distances, identifiers and scores are integers of $O(\log n)$ bits. Ties are settled by the specified total order, so the computation has one outcome rather than a branching search.
\end{proof}

\begin{lemma}\label{lem:deletion}
If a candidate $q$ is removed from an active set $K$, the score of every surviving candidate weakly increases. If a voter's old favorite survives, its favorite is unchanged.
\end{lemma}
\begin{proof}
Let $p\in K\setminus\{q\}$ have been the favorite of voter $v$. It was preferred to every other member of $K$, and hence to every member of the smaller active set other than itself. It remains the favorite. Thus only voters formerly assigned to $q$ change their assignment, and none of the old supporters of a survivor is lost.
\end{proof}

Every active candidate gets its own vertex's vote. For a fixed candidate set, the complete election is polynomial-time evaluable (Lemma~\ref{lem:evaluate}).
We isolate the candidate-loss question used by the graph-specific algorithms.

\begin{definition}[$\Kill$]\label{def:kill}
Given $G$, a mandatory target $u\in V$, and allowed opponents $A\subseteq V\setminus\{u\}$, ask whether there is a nonempty $F\subseteq A$ such that $\IRV(G,\{u\}\cup F)\ne u$. Write $\Kill(G,u,A)=\yes$ in this case, and call $F$ a witness. Set $D=A\cup\{u\}$; vertices outside $D$ remain voters even though they cannot stand in this instance.
\end{definition}

\noindent
$A$ specifies eligibility, not compulsory participation. In particular, $\Kill(G,u,\varnothing)=\no$.

\begin{definition}[Exclusion zone and minimum-zone computation]\label{def:zone}\label{def:minzone}
A nonempty $S\subseteq V$ is an exclusion zone if, for every nonempty $K\subseteq V$,
$K\cap S\ne\varnothing\quad\Longrightarrow\quad\IRV(G,K)\in S$.
The problem $\ZoneProblem$ asks whether a given nonempty vertex set $S$ is an exclusion zone; the empty set is rejected. The problem $\MinZone$ outputs the unique zone $S^\star$ contained in every other zone, whose existence follows from Lemma~\ref{lem:nesting}.
\end{definition}

Unlike the candidate-loss problem $\Kill(G,u,A)$, where $A$ is supplied as part of the input, this definition has no fixed allowed-opponent set: it quantifies over all nonempty candidate sets $K\subseteq V$. The full vertex set $V$ is always an exclusion zone.

\begin{example}
An unchanged voter graph can have different zone answers
under our voting convention and that of
Tomlinson et al.~\cite{ExclusionZones}.
Consider the path $a-u-b$, with one voter per vertex,
and give $u$ the smallest identifier and $b$ the largest.
Under our convention, $u$ wins all four candidate
configurations containing it, so $\{u\}$ is a zone.
Under split votes with unrestricted elimination-tie
resolution, activating all three permits eliminating $u$
first and electing an endpoint. Thus the same singleton
is not a tie-robust zone.
Section~\ref{app:common} gives the complete calculation.
This example changes both vote assignment and the
elimination-tie quantifier; it does not isolate the
effect of splitting alone.
We continue with the details.

Consider $G$ with $V(G)=\{a,u,b\}$ and $E(G)=\{au,ub\}$. Every vertex hosts one unit voter and is eligible to participate. Fix $\id(u)=0$, $\id(a)=1$, and $\id(b)=2$. The candidate distances are
\[
\begin{array}{c|ccc}
\text{voter}&a&u&b\\\hline
 a&0&1&2\\u&1&0&1\\b&2&1&0
\end{array}
\]
A zone claim for $S=\{u\}$ concerns every nonempty candidate configuration containing $u$. There are exactly four.

Under our whole-vote rule, with smaller-ID distance ties and larger-ID elimination ties:
\begin{enumerate}[label=(\roman*)]
\item For $K=\{u\}$, all three voters support $u$, and the singleton wins without an elimination round.
\item For $K=\{u,a\}$, the voter at $a$ supports $a$ and the voters at $u,b$ support $u$. The scores are two for $u$ and one for $a$; $a$ is eliminated and $u$ wins.
\item For $K=\{u,b\}$, the voters at $a,u$ support $u$ and the voter at $b$ supports $b$. The scores are two for $u$ and one for $b$; $b$ is eliminated and $u$ wins.
\item For $K=\{u,a,b\}$, each candidate initially receives its own vertex's vote, so all scores are one. The fixed rule eliminates $b$, the largest-ID candidate. Its voter then strictly prefers $u$ to $a$, at distances one and two. The remaining scores are two for $u$ and one for $a$, so $u$ wins.
\end{enumerate}
These cases exhaust the configurations meeting $S$, proving that $\{u\}$ is a zone under our convention.

For comparison, use the split-vote graph convention of~\cite{ExclusionZones}: a voter divides its unit vote equally among its nearest active candidates, and a zone must hold for every permitted resolution of elimination ties. Activate $K=\{u,a,b\}$. Each voter is uniquely closest to its own candidate, so the initial scores are again one each. A permitted resolution removes $u$ first. Its vertex and voter remain. With $a,b$ active, the middle voter is equidistant from them and contributes one half to each; the endpoint voters support themselves. Hence
\[
 s(a)=1+\tfrac12=\tfrac32,
 \qquad s(b)=1+\tfrac12=\tfrac32.
\]
Either endpoint can be the last survivor after the next elimination tie. This permitted execution elects outside $S$, disproving the tie-robust zone guarantee. An existential failing execution suffices; the endpoints need not win under every resolution.
\end{example}

The example disproves unchanged-instance preservation under the displayed conventions. It changes both vote assignment and the elimination-tie quantifier, and does not isolate splitting alone. It does not rule out shared structural arguments, a separately proved adaptation, or a reduction between the two models. The shared objects are the voter graph and candidate locations; the voting conventions differ.

\begin{lemma}[Round-one reduction]\label{lem:restart}
$\Kill(G,u,A)=\yes$ if and only if some nonempty $F\subseteq A$ eliminates $u$ in the first round on $\{u\}\cup F$.
\end{lemma}
\begin{proof}
If $u$ loses, let $R$ be the active set immediately before it is eliminated. Restart the election with $R$: the graph, voters, identifier order, votes and elimination keys in its first round are exactly those in the decisive original round. Hence $u$ is eliminated first, with nonempty opponent set $R\setminus\{u\}\subseteq A$. The reverse implication is immediate.
\end{proof}

\begin{lemma}\label{lem:accept}
Let $K=\{u\}\cup F$ with $F\ne\varnothing$. Define
\[
 \mu_{\rm opp}:=\minlex\{(\score_K(q),-\id(q)):q\in F\}.
\]
Then $u$ is eliminated first exactly when
\begin{equation}
 (\score_K(u),-\id(u))<_{\lex}\mu_{\rm opp}.\label{eq:accept}
\end{equation}
If $u$ has the smallest identifier in $K$, this is equivalent to $\score_K(u)<\score_K(q)$ for every $q\in F$.
\end{lemma}
\begin{proof}
All elimination keys are distinct, because their identifier coordinates are distinct. The target is the unique minimum exactly when its key is smaller than every opponent's, equivalently smaller than the minimum opponent key. If $u$ has the smallest identifier and an opponent ties its score, that opponent has a smaller negative-identifier coordinate and precedes $u$. Thus score ties never eliminate such a target. Strictly smaller target score than every opponent is both necessary and sufficient in this special case.
\end{proof}

For $K=\{u\}\cup F$ with $F\ne\varnothing$, the exact first-round test is $\lambda_K(u)<_{\lex}\min_{q\in F,\lex}\lambda_K(q)$ (Lemma~\ref{lem:accept}). If $u$ has the smallest identifier in $K$, this requires $\score_K(u)<\score_K(q)$ for every $q\in F$; equality eliminates the tied opponent first.

\begin{lemma}\label{lem:zonekill}
A nonempty $S\subseteq V$ is an exclusion zone if and only if $\Kill(G,u,V\setminus S)= {\sf false}$ for every $u\in S$.
\end{lemma}
\begin{proof}
\ed{Suppose $\Kill(G,u,V\setminus S)=\yes$ for some $u\in S$. Then some nonempty $F\subseteq V\setminus S$ makes $u$ lose the election on $\{u\}\cup F$. Candidate $u$ is its only candidate from $S$, so the winner lies outside $S$. Hence $S$ is not an exclusion zone.}

Conversely, suppose $S$ is not an exclusion zone. There is a candidate set $K$ intersecting $S$ whose winner $w$ lies outside $S$. Let $u$ be the last candidate from $S$ to be eliminated in this election. Immediately before its elimination, the active set is $\{u\}\cup F$ for some $F\subseteq V\setminus S$. Since $w\in F$, this set is nonempty.

Restarting with $\{u\}\cup F$ preserves the round eliminating $u$, by the argument of Lemma~\ref{lem:restart}. Therefore $\Kill(G,u,V\setminus S)=\yes$.
\end{proof}

\begin{corollary}\label{cor:singleton}
For every target $u$,
\[
 \{u\}\text{ is not an exclusion zone}
 \quad\Longleftrightarrow\quad
 \Kill(G,u,V\setminus\{u\})=\yes.
\]
\end{corollary}
\begin{proof}
Set $S=\{u\}$ in Lemma~\ref{lem:zonekill} and negate the equivalence.
\end{proof}

Tomlinson et al.~\cite{ExclusionZones} prove nesting under their convention. The same two-candidate argument below establishes it under ours and ensures that the minimum zone is well-defined.

\begin{lemma}[Nesting]\label{lem:nesting}
For any two exclusion zones $S,T$, either $S\subseteq T$ or $T\subseteq S$. Consequently, there is a unique zone $S^\star$ contained in every other zone, and it is the only zone having no nonempty proper subset that is also a zone.
\end{lemma}
\begin{proof}
\ed{Suppose neither set contains the other. Choose $s\in S\setminus T$ and $t\in T\setminus S$, and consider the election with candidates $\{s,t\}$. If $s$ wins, the election contains a candidate from $T$ but has its winner outside $T$. If $t$ wins, the same contradiction holds for $S$. Thus one set contains the other.}

There are finitely many exclusion zones, and at least one exists because $V$ is an exclusion zone. Choose a zone $S^\star$ with the fewest vertices. For every zone $T$, the first part gives $S^\star\subseteq T$ or $T\subseteq S^\star$. A proper inclusion $T\subsetneq S^\star$ would contradict the choice of $S^\star$, so $S^\star$ is contained in every zone.

This makes $S^\star$ unique. It has no proper subset that is a zone, whereas every different zone contains the proper subset $S^\star$, which is a zone. Finally, $S^\star$ is nonempty by the definition of an exclusion zone.
\end{proof}

\begin{definition}[Pairwise-loss graph and closure]\label{def:lossgraph}
The \emph{pairwise-loss graph} $L(G)$ has vertex set $V$ and, for distinct $x,y\in V$, contains the directed edge $x\to y$ exactly when $\IRV(G,\{x,y\})=y$. This is a tournament: between every pair of distinct vertices there is exactly one directed edge. For $X\subseteq V$, let $\cl(X)$ be the vertices reachable by directed paths from $X$, including paths of length zero. Write $\cl(x)=\cl(\{x\})$. A set $X\subseteq V$ is \emph{closed} if $\cl(X)=X$.
\end{definition}

\begin{lemma}\label{lem:closed}
If $S$ is an exclusion zone, then $\cl(S)=S$.
\end{lemma}
\begin{proof}
If $x\in S$ and $x\to y$, the two-candidate election $\{x,y\}$ intersects $S$ and is won by $y$. Hence $y\in S$. Apply this implication to each arc of a directed path beginning in $S$ to obtain $\cl(S)\subseteq S$. Length-zero paths give the reverse inclusion.
\end{proof}

Every exclusion zone $S$ is closed. If $x\in S$ loses to $y$ in the two-candidate election, that election meets $S$, so the zone guarantee requires $y\in S$. Applying this implication along directed paths gives the closure property (Lemma~\ref{lem:closed}).
\begin{lemma}\label{lem:singleclosure-main}
Every nonempty set $S$ closed in $L(G)$ equals $\cl(s)$ for some $s\in S$.
\end{lemma}
\begin{proof}
Choose $s\in S$ for which $|\cl(s)|$ is largest. Since $S$ is closed, $\cl(s)\subseteq S$. Suppose there is a vertex $t\in S\setminus\cl(s)$.

For any $x\in\cl(s)$, the arc $x\to t$ cannot exist, since a path from $s$ to $x$, followed by that arc, would make $t$ reachable from $s$. Since $L(G)$ is a tournament, the opposite arc $t\to x$ must therefore exist.

Thus $t$ reaches every vertex of $\cl(s)$ and also itself. This gives $|\cl(t)|>|\cl(s)|$, contradicting our choice
of $s$. Hence $S=\cl(s)$.
\end{proof}

Closedness alone is insufficient. Consider the path $1-2-3-4-5$, with identifiers in increasing numerical order. Vertex $3$ wins every two-candidate election in which it participates, so $\cl(3)=\{3\}$. With candidates $\{2,3,4\}$, however, the scores are $2,1,2$, respectively, and $3$ is eliminated first. Thus $\{3\}$ is closed but is not an exclusion zone. We show:

\begin{theorem}\label{thm:transfer}
Suppose that, once all-pairs distances are available, $\Kill(G,u,A)$ can be decided in time $T(n)$ for every target $u$ and allowed set $A$. The distance matrix takes $O(n(n+|E|))$ time to compute. Thereafter zone verification takes at most $n$ calls to $\Kill$ and $O(n)$ additional operations; minimum-zone computation or listing all zones takes at most $n^2$ calls and $O(n^3)$ additional operations.
\end{theorem}

\begin{proof}
Verification is precisely Lemma~\ref{lem:zonekill}. To generate the minimum-zone search family, compute the result of every two-candidate election, build $L(G)$, and compute the $n$ sets $\cl(v)$. Each pairwise election takes $O(n)$ key comparisons after distances are available; all pairwise elections and all directed reachability searches therefore take $O(n^3)$ operations.

By Lemmas~\ref{lem:closed} and~\ref{lem:singleclosure-main}, every zone is among these closures. Remove duplicate sets and verify each remaining set using Lemma~\ref{lem:zonekill}, that is, by the corresponding calls to $\Kill$. There are at most $n$ sets and at most $n$ calls for each. At least one passes: $V$ is a zone and, by Lemma~\ref{lem:singleclosure-main} applied to the whole tournament, it is one of the singleton closures. The accepted sets are exactly the zones, and they are nested by Lemma~\ref{lem:nesting}. Return the smallest accepted set for \MinZone, or list all accepted sets for enumeration.

Pairwise closure is only a necessary condition for being an exclusion zone. Every resulting set is still checked using $\Kill$ as required by Lemma~\ref{lem:zonekill}. The output contains at most $n$ sets of at most $n$ vertices, so its encoding length is polynomial as well.
\end{proof}

\begin{corollary}[Membership in $\NP$ and $\coNP$]\label{cor:membership}
The problem $\Kill$ belongs to $\NP$, and exclusion-zone verification belongs to $\coNP$. The same statements hold after restricting the inputs to any graph class that can be recognized in polynomial time.
\end{corollary}
\begin{proof}
For $\Kill$, guess a nonempty set $F\subseteq A$ and use Lemmas~\ref{lem:restart} and~\ref{lem:accept} to check in polynomial time whether it eliminates $u$ in the first round.
For exclusion-zone verification, if a proposed nonempty set $S$ is not an exclusion zone, a candidate set $K$ with $K\cap S\ne\varnothing$ and $\IRV(G,K)\notin S$ demonstrates this. Lemma~\ref{lem:evaluate} computes its winner in polynomial time. Invalid encodings and empty proposed zones can be rejected directly. If the input is restricted to a polynomially recognizable graph class, membership in that class can also be checked in polynomial time.
\end{proof}

\begin{algobox}{{I. Reduction from candidate loss to exclusion zones}}
\algio{Graph $G$, a routine deciding $\Kill(G,u,A)$, and a proposed set $S$ when verification is requested.}{The verification answer, the minimum exclusion zone, or the list of all nonempty exclusion zones, according to the requested task.}
For verification, execute Step~1. For minimum-zone computation or enumeration, execute Steps~2--3, using Step~1 as a verification subroutine.
\begin{enumerate}[leftmargin=5mm,itemsep=.35ex,topsep=.45ex]
\item {Verify a proposed zone.} Return \mbox{{\sf false}} if $S=\varnothing$. For every $u\in S$, test $\Kill(G,u,V\setminus S)$. If any test returns \mbox{{\sf true}}, return \mbox{{\sf false}}. Otherwise return \mbox{{\sf true}}.
\item {Generate candidate zones.} Compute all two-candidate outcomes, orient each loser--winner pair, and form the pairwise-loss tournament $L(G)$. For every $v\in V$, compute its reachability closure $\cl(v)$.
\item {Verify and select.} Remove duplicate closures. Verify each remaining closure by Step~1. Return the smallest verified closure as the minimum exclusion zone, or list all verified closures.
\end{enumerate}
\end{algobox}

The preceding algorithm box summarizes the verification, minimum-zone, and enumeration procedures.
The symbol $\bot$ means ``no candidate''. Discard it from candidate lists before applying $\best_v$, and set $\best_v(\varnothing)=\bot$. This is a notation convention, not an empty election. An opponent is finalized once its complete first-round score has been determined. Its elimination key can then be included in the minimum $\mu$; set $\mu=\infty$ when no opponent has been finalized. Finalization does not remove the candidate from the election. Here $\infty$ is a special value ordered above every finite elimination key. A missing vote-count entry means zero only when the algorithm proves that no counted voter can support that candidate.


We use standard plane and surface-embedding terminology;
see~\cite{MoharThomassen}.
A graph is planar if it admits a drawing in the plane
in which edges meet only at common endpoints.
We also use the equivalent sphere representation:
removing a point from a face gives a plane drawing
with that face unbounded.
A graph is outerplanar if it has such a plane drawing
with every vertex incident with the unbounded face.
An embedding need not be supplied with the input. A tree decomposition has vertex bags on a tree: every graph vertex appears in a bag, every edge has both endpoints in one bag, and the bags containing any fixed vertex form a connected subtree. Its width is one less than the largest bag size; $\tw(H)$ is the minimum width of a tree decomposition of $H$. Definitions~\ref{def:outerplanar} and~\ref{def:treedecomp} 
in this section give the details.

\section{IRV on Trees}\label{m:trees}\label{sec:trees}

We give a linear-time algorithm for $\Kill$ on trees. It uses the preference regions defined below and maintains their union while scanning score thresholds. By Lemma~\ref{lem:restart}, it suffices to eliminate the target first. Root $T$ at $u$. Let $T_x$ be the subtree consisting of $x$ and all its descendants, and let the depth of vertex $x$ be $d_T(u,x)$. For $q\in A$, the preference region
$R_q := \{v\in V:\kappa_v(q)<_{\lex}\kappa_v(u)\}$
contains exactly the voters choosing $q$ against $u$ in their two-candidate election. Other opponents may take some of these votes, so $|R_q|$ is only an upper bound on $q$'s first-round score.

This section proves the linear-time tree algorithm. All vote counts refer to first-round scores in the fixed whole-vote, common-ID model.

\subsection{Preference Regions and the Linear Threshold algorithm}\label{subsec:tree-quadratic}
Root the tree at the compulsory target $u$ and use the voting model and $\Kill$ definition from Section~\ref{sec:model}. Lemma~\ref{lem:restart} reduces the problem to first-round elimination.

The algorithm uses a simple object for every allowed candidate. For $q\in A$, define
\begin{equation}\label{eq:tree-quadratic-region}
 R_q:=\cell_{\{u,q\}}(q)
     =\{v\in V:\kappa_v(q)<_{\lex}\kappa_v(u)\}.
\end{equation}
We call $R_q$ the \emph{preference region of $q$ against $u$}. Thus $R_q$ is the set of voters who would choose $q$ in the two-candidate election $\{u,q\}$. It is defined only from the pair $u,q$; when other
candidates participate, some voters in $R_q$ may vote for somebody else.

\begin{lemma}\label{lem:tree-quadratic-regions}\label{m:tree-regions}
For $q\in A$, put $\ell_q=d_T(u,q)$ and let $b_q$ be the ancestor of $q$ at depth
\begin{equation}\label{eq:tree-quadratic-depth}
 k_q=\begin{cases}
 \ell_q/2,&\ell_q\text{ is even and }\id(q)<\id(u),\\
 \lfloor\ell_q/2\rfloor+1,&\text{otherwise}.
 \end{cases}
\end{equation}
Then  $R_q=V(T_{b_q})$. We call $b_q$ the \emph{root of the preference region $R_q$}. Consequently, any two preference regions are either disjoint or one contains the other. Moreover, $R_p=R_q$ exactly when $b_p=b_q$.
A family with this property is called \emph{laminar}.
\end{lemma}

\begin{proof}
Since $q\ne u$, we have $\ell_q\ge1$ and hence $1\le k_q\le\ell_q$, so $b_q$ exists. Fix a voter $v$. Let $p$ be the last common vertex of the $u$--$v$ and $u$--$q$ paths, and put $j=d_T(u,p)$. By the unique-path
property of a tree, $d_T(v,q)-d_T(v,u)=\ell_q-2j$.
By Definition~\ref{def:irv}, voter $v$ prefers $q$ to $u$ exactly when $2j>\ell_q$, or when $2j=\ell_q$ and $\id(q)<\id(u)$. These conditions are equivalent to $j\ge k_q$, which holds exactly when $v$ is a descendant of $b_q$. Hence $R_q=V(T_{b_q})$.

If two rooted subtrees intersect, their roots are both ancestors of any vertex in the intersection, so one root is an ancestor of the other. Their vertex sets are therefore nested. Distinct incomparable roots give disjoint subtrees, and two such subtrees are equal only when their roots are equal.
\end{proof}

For later use, notice what these regions say about an arbitrary first-round election. Let $F\subseteq A$ and $K=\{u\}\cup F$. A voter votes for $u$ exactly when it belongs to none of the regions $R_q$,
$q\in F$. Also, any voter who votes for an opponent $q$ must belong to $R_q$. Therefore
\begin{equation}\label{eq:tree-quadratic-scores}
 \score_K(u)=n-\left|\bigcup_{q\in F}R_q\right|,
 \qquad
 \score_K(q)\le |R_q|\quad(q\in F).
\end{equation}
The second relation is only an upper bound in a general election, because two opponents may compete for voters in overlapping regions.

We now define the test used by the algorithm. Fix an integer $t\in\{1,\ldots,n-1\}$. Here $t$ is a possible upper bound on the target's first-round score. We say that an allowed candidate
$q$ \emph{qualifies at threshold $t$} if
\begin{equation}\label{eq:tree-quadratic-qualifies}
 (|R_q|,-\id(q))>_{\lex}(t,-\id(u)).
\end{equation}
Equivalently, $q$ qualifies when $|R_q|>t$, or when $|R_q|=t$ and $\id(q)<\id(u)$. The second case is exactly the identifier condition needed for $u$ to lose an equal-score elimination tie against $q$.

Let
\begin{equation}\label{eq:tree-quadratic-threshold}
 Q_t:=\{q\in A:q\text{ qualifies at threshold }t\},
 \qquad
 U_t:=\bigcup_{q\in Q_t}R_q.
\end{equation}
Thus $Q_t$ is the set of \emph{qualifying candidates}, and $U_t$ is the set of voters covered by at least one qualifying preference region. We also need one simple notion for these regions. A \emph{maximal
qualifying region} is a distinct region $R_q$ with $q\in Q_t$ that is not strictly contained in another qualifying region. Equal regions are treated as one region. For every maximal qualifying region choose exactly one qualifying candidate having that region; this candidate is its \emph{representative}. Because preference regions are nested or disjoint, distinct maximal qualifying regions are pairwise disjoint. Every qualifying region is contained in one of them, so their union is still exactly $U_t$.

\begin{lemma}\label{lem:tree-quadratic-test}\label{m:tree-threshold}
The following equivalence holds:
\begin{equation}\label{eq:tree-quadratic-test}
 \Kill(T,u,A)=\yes
 \quad\Longleftrightarrow\quad
 \exists\,t\in\{1,\ldots,n-1\}:\quad
 Q_t\ne\varnothing\ \text{ and }\ n-|U_t|\le t.
\end{equation}
For every successful threshold $t$, the representatives of the maximal qualifying regions form a first-round witness.
\end{lemma} 
\begin{proof}
\emph{Soundness.}
Suppose a threshold $t$ satisfies the right-hand side of \eqref{eq:tree-quadratic-test}, and let $F_t$ contain one representative of every maximal qualifying region. Since $Q_t\ne\varnothing$, at least
one maximal qualifying region exists, so $F_t\ne\varnothing$. As explained above, these regions are pairwise disjoint and their union is $U_t$.

Fix $q\in F_t$ and a voter $v\in R_q$. The voter prefers $q$ to $u$. For every other selected candidate $p\in F_t$, disjointness gives $v\notin R_p$, so $\kappa_v(u)<_{\lex}\kappa_v(p)$, since these two keys are distinct. Hence $v$ prefers $q$ to every
other participating candidate and votes for $q$. Conversely, a voter outside $R_q$ cannot vote for $q$ while $u$ participates. Thus
\[
 \score_{\{u\}\cup F_t}(q)=|R_q|\qquad(q\in F_t).
\]
Equation~\eqref{eq:tree-quadratic-scores} also gives
\[
 \score_{\{u\}\cup F_t}(u)=n-|U_t|\le t.
\]
Every $q\in F_t$ qualifies at threshold $t$, so
\[
 \lambda_{\{u\}\cup F_t}(u)
 \le_{\lex}(t,-\id(u))
 <_{\lex}(|R_q|,-\id(q))
 =\lambda_{\{u\}\cup F_t}(q).
\]
By Lemma~\ref{lem:accept}, $u$ is eliminated first.

\emph{Completeness.}
Suppose $\Kill(T,u,A)=\yes$. By Lemma~\ref{lem:restart}, let $F\subseteq A$ be a first-round witness, put $K=\{u\}\cup F$, and set $t=\score_K(u)$. Every participating candidate receives at least its own
vote, so $1\le t\le n-1$.

For every $q\in F$, Equation~\eqref{eq:tree-quadratic-scores} gives $\score_K(q)\le |R_q|$. Since $u$ is eliminated before $q$, Lemma~\ref{lem:accept} gives
\[
 (|R_q|,-\id(q))
 \ge_{\lex}\lambda_K(q)
 >_{\lex}\lambda_K(u)
 =(t,-\id(u)).
\]
Thus every opponent in $F$ belongs to $Q_t$. Hence
\[
 \bigcup_{q\in F}R_q\subseteq U_t,
\]
and Equation~\eqref{eq:tree-quadratic-scores} yields
\[
 n-|U_t|
 \le n-\left|\bigcup_{q\in F}R_q\right|
 =t.
\]
Therefore this threshold succeeds. If $A=\varnothing$, including every valid input with $n=1$, both sides of~\eqref{eq:tree-quadratic-test} are
false.
\end{proof}

\subsection{Implementation and Complexity}

A first depth-first traversal computes depths, children and subtree sizes. A second depth-first traversal maintains a root-to-vertex path array, which locates each region root $b_q$ in constant time. Put $r_q=|T_{b_q}|$ and set $h_q=r_q$ if $\id(q)<\id(u)$, and $h_q=r_q-1$ otherwise. Candidate $q$ qualifies exactly when $t\le h_q$. Bucket $h$ contains the candidates with $h_q=h$. Ignore bucket zero and scan thresholds in decreasing order, processing each bucket before testing its threshold.
Maintain the union of processed rooted subtrees and its size. When a region root is not yet covered, traverse its subtree, stopping at already covered vertices. Between insertions, every covered vertex has only covered descendants. After bucket $t$ has been processed, the covered set is exactly $U_t$. Each vertex is expanded at most once, and each child edge is inspected at most once. There are at most $|A|$ insertion attempts; the total coverage work is therefore linear. Keep separate region-root marks for witness recovery. At a successful threshold, traverse from $u$. At the first marked region root on each branch, output its stored candidate and skip its descendants. This selects one representative per maximal qualifying region. Box~II summarizes the procedure.

\begin{algobox}{{II. TreeCandidateLoss}}
\algio{A tree $T$, target $u$, and allowed opponents $A$.}{Whether $\Kill(T,u,A)=\yes$, 
and a first-round witness when one exists.}
\begin{enumerate}[leftmargin=5mm,itemsep=.35ex,topsep=.45ex]
\item Return \mbox{{\sf false}} if $A=\varnothing$. Root $T$ at $u$ and compute depths, child lists and subtree sizes. During a second depth-first traversal, use the path array described above to compute every region root $b_q$.
\item For each $q\in A$, set $r_q=|T_{b_q}|$ and $h_q=r_q$ if $\id(q)<\id(u)$, and $h_q=r_q-1$ otherwise. Put each candidate with $h_q\ge1$ in bucket $h_q$. Initialize separate region-root and coverage flags to false, representative entries to empty, and the covered-vertex counter $C=0$.
\item For each $t=n-1,\ldots,1$, process every candidate in bucket $t$: mark its region root, store a representative if needed, and update coverage by the downward traversal described above. After processing the entire bucket, test whether some candidate has been processed and $n-C\le t$. If so, traverse from $u$, outputting the stored candidate at each first marked region root and skipping its descendants. Return \mbox{{\sf true}} together with this witness.
\item Return \mbox{{\sf false}} if every threshold fails.
\end{enumerate}
\end{algobox}

\begin{theorem}[Linear-time tree loss testing]\label{m:treetheorem}
For every tree $T$, target $u$ and allowed set $A$, $\Kill(T,u,A)$ and a first-round witness when one exists are computable in $O(n)$ time and $O(n)$ working space. On trees, zone verification takes $O(n^2)$ time; minimum-zone computation and enumeration of all nonempty zones take $O(n^3)$ time.
\end{theorem}
\begin{proof}[Implementation and complexity]
Return $\no$ if $A=\varnothing$. Store membership in $A$ in a Boolean array, constructed in $O(n)$ time, so that each membership test takes constant time. This query does not construct an all-pairs distance matrix. Root the adjacency-list tree at $u$. A depth-first traversal computes parent pointers, depths, child lists, a preorder, and all subtree sizes in $O(n)$ time. During a second depth-first traversal maintain an array $P$ such that, at a current vertex $x$, $P[d]$ is the ancestor of $x$ at depth $d$. Push $x$ on entry and pop it on exit. For every $q\in A$, compute $k_q$ from \eqref{eq:tree-quadratic-depth} and read $b_q=P[k_q]$ in $O(1)$ time. Thus all region roots and sizes $r_q=|T_{b_q}|$ are obtained in $O(n)$ total time.

For each $q\in A$, define
\[
 h_q=\begin{cases}
 r_q,&\id(q)<\id(u),\\
 r_q-1,&\id(q)>\id(u).
 \end{cases}
\]
By \eqref{eq:tree-quadratic-qualifies}, $q\in Q_t$ exactly when $t\le h_q$. Since $1\le r_q\le n-1$, we have $0\le h_q\le n-1$. Bucket $h$ is the list of candidates with $h_q=h$. Insert candidates with $h_q\ge1$ into their unique buckets; those with $h_q=0$ never qualify at a tested threshold and are ignored. Initializing the lists and distributing the candidates takes $O(n)$ time.

Initialize separate Boolean arrays $\mathrm{marked}$ and $\mathrm{covered}$ to false, all representative entries to empty, and $C=0$. The flag $\mathrm{marked}(x)$ records that $x$ is the root of a processed preference region; $\mathrm{covered}(x)$ records that $x$ belongs to the union of processed regions. Also record whether any candidate has been processed. Scan $t=n-1,n-2,\ldots,1$. Before processing bucket $t$, coverage is the union of the regions inserted at larger thresholds.

When a candidate $q$ in bucket $t$ is inserted, mark $b_q$ and store a representative there if none was stored previously. If $b_q$ is already covered, do not traverse it. Otherwise start a downward traversal at $b_q$: mark every newly reached vertex covered, increase $C$, and inspect its children; when a covered vertex is reached, stop below it.

The invariant remains true. Every inserted region is the whole rooted subtree $T_{b_q}$, and a covered set is a union of rooted subtrees. Thus if $b_q$ is covered, every vertex of $T_{b_q}$ is already covered. Otherwise the traversal adds exactly the previously uncovered vertices of $T_{b_q}$; stopping at a covered vertex is sound because all its descendants are covered as well. Since the inserted candidates after bucket $t$ are exactly $Q_t$, we have $C=|U_t|$ after that bucket.

The total time of all coverage traversals is $O(n)$. Each vertex becomes newly covered at most once. When it becomes covered, only its child edges are inspected, so at most $n-1$ child-edge inspections can create a later stopped visit. The only other traversal starts are the at most $|A|\le n-1$ bucket insertions. Thus successful visits, stopped visits, and starts are all $O(n)$. The bucket scan itself is linear.

After each bucket, test whether at least one candidate has been inserted and $n-C\le t$. By Lemma~\ref{lem:tree-quadratic-test}, at least one threshold succeeds if and only if $\Kill(T,u,A)=\yes$. Return $\no$ if all thresholds fail. To recover a witness, traverse from $u$: at a marked vertex output its stored representative and do not enter its children. The reached marked vertices have no marked ancestor, hence their subtrees are precisely the maximal qualifying regions; the output is a first-round witness by the same lemma. This traversal is linear. All arrays and stacks have $O(n)$ entries, proving the claimed time and space bounds.

\emph{Zone bounds.}
Apply Theorem~\ref{thm:transfer} using the $O(n)$ query bound just proved. Its one-off distance computation costs $O(n^2)$ on a tree. The at most $n$ calls used for verification therefore cost $O(n^2)$, and the at most $n^2$ calls used for minimum-zone computation or enumeration cost $O(n^3)$. The $O(n)$ working-space bound concerns one $\Kill$ query, not storage for the complete transfer procedure.
\end{proof}

\section{IRV on Outerplanar Graphs}\label{m:outer}\label{sec:outer}
This section gives the complete outerplanar dynamic program. Bags of at most three vertices determine which candidates can receive votes across a bag. The algorithm does not rely on the laminar preference regions used on trees.

\subsection{Outerplanarity and the Decomposition Used by the Algorithm}

\begin{definition}\label{def:outerplanar}
A plane embedding is an embedding of a graph on the sphere without edge crossings except at common endpoints; choosing one face as the unbounded face realizes it as a drawing in the plane. A graph is planar if such an embedding exists. It is outerplanar if it has a plane drawing in which every vertex is incident with the unbounded face.
\end{definition}

The distinction is between a graph property and a particular drawing: a planar or outerplanar graph need not be supplied with an embedding. A bridge is an edge whose deletion increases the number of connected components. A cutvertex is a vertex whose deletion does so.

\begin{external}\label{ext:planarity}
A finite graph can be tested for planarity in polynomial time, and a plane embedding can be obtained when it is planar. We use this standard algorithmic result~\cite{HopcroftTarjan} only for preprocessing and checking the graph restrictions.
\end{external}

\begin{lemma}\label{lem:getouter}
An outerplane embedding of a connected outerplanar graph can be found in polynomial time.
\end{lemma}
\begin{proof}
Add a new vertex $z$ adjacent to every vertex of $G$. If $G$ has an outerplane embedding, choose one outer-face occurrence of each vertex. Trace that facial walk just outside a thin vertex-and-edge neighborhood, separating repeated occurrences of a cutvertex or bridge side. This produces the boundary of a disk in the outer face, with a distinct marked point for each chosen occurrence; the disk-neighborhood justification is given independently in Lemma~\ref{lem:facedisks}. Place $z$ in this disk, join it to the marked points by a noncrossing fan in their cyclic order, and extend each fan edge through its chosen small vertex sector to the corresponding original vertex. Those extensions avoid the old graph except at their endpoints. This gives a plane embedding of the enlarged graph, even when the original outer facial walk is not a simple cycle. Conversely, in any plane embedding of the enlarged graph, deleting $z$ and its incident edges merges their incident faces. Every old vertex has an edge to $z$, so it is incident with this merged face. Taking it as the outer face gives an outerplane embedding of $G$. Run External result~\ref{ext:planarity} on the enlarged graph and perform this deletion. The new vertex is only a device for finding an embedding; it is not inserted into the voter instance.
\end{proof}

\begin{definition}\label{def:treedecomp}
A tree decomposition of a graph $G$ is a tree $\mathcal T$ with a vertex set $X_t\subseteq V(G)$, called a bag, at every node $t$ of $\mathcal T$, satisfying:
\begin{enumerate}
\item[{\sf (1)}]
Every graph vertex occurs in at least one bag.
\item[{\sf (2)}]
The endpoints of every graph edge occur together in some bag.
\item[{\sf (3)}]
For each graph vertex $v$, the nodes whose bags contain $v$ form a connected subtree of $\mathcal T$.
\end{enumerate}
Its width is $\max_t |X_t|-1$. The treewidth $\tw(G)$ is the minimum such width.
\end{definition}

The width-two bound for outerplanar graphs is standard; we give the triangulated-polygon construction needed to obtain explicit bags~\cite{BondyMurty}.
\begin{lemma}\label{lem:outertw}
Every outerplanar graph has treewidth at most two. For a connected outerplanar graph, a decomposition with bags of size at most three and $O(n)$ nodes can be constructed in polynomial time.
\end{lemma}
\begin{proof}
The cases $n=1,2$ have one bag containing every vertex. Suppose $n\ge3$ and fix an outerplane embedding from Lemma~\ref{lem:getouter}. Add edges while preserving this embedding and the property that every vertex is on the outer face, until no further edge can be added.

The resulting graph has no cutvertex. If $v$ were a cutvertex, two consecutive edges $va,vb$ around an appropriate outer-face occurrence of $v$ would lead to different components after deleting $v$. Their endpoints $a,b$ are not adjacent, or those components would be connected. An edge $ab$ can be drawn in the outer-face wedge between them, preserving outerplanarity and contradicting maximality. Thus the outer boundary is a simple cycle through all vertices. Every bounded face is a polygon. If such a face had more than three sides, a missing diagonal could be drawn inside it, again contradicting maximality. Therefore the supergraph is a triangulated polygon.

Use one bag for the three vertices of each bounded triangle. Join two bags when their triangles share an internal diagonal. This weak dual is a tree: starting with a polygon, every noncrossing diagonal splits one region into two; the adjacency graph of the regions gains a node and one link at each split, preserving the tree property. There are $n-2$ triangles.

Every vertex is in a triangle and every edge is in at least one triangle. For a fixed polygon vertex $v$, the triangles incident with $v$ occur in order around its interior angle; successive triangles in that order share a diagonal incident with $v$. Hence their bags form a connected path in the weak dual. These observations prove all three decomposition axioms. Each bag has three vertices. Since only edges, not vertices, were added, the same decomposition is valid for the original graph. The added diagonals are never used to compute election distances.

If the graph in the structural statement is disconnected, perform the construction in each component and join the resulting decomposition trees by arbitrary tree edges. No vertex occurs in different components, so the connectedness axiom is preserved.
\end{proof}

\begin{definition}\label{def:nice}
A rooted decomposition is nice if each node is one of four types: an empty-bag leaf; an introduce node $t$ with one child $c$ and $X_t=X_c\cup\{x\}$; a forget node with one child and $X_c=X_t\cup\{x\}$; or a join node with two children and three equal bags. For a node $t$, let $V_t$ be the union of all bags in the subtree rooted at $t$, and define its private set by
\[
 I_t:=V_t\setminus X_t.
\]
Private voters will be counted at $t$. Bag voters are deferred until the vertex becomes private or is handled at the root.
\end{definition}

We use the standard nice-decomposition conversion, recording explicitly the root and private-target conditions required by this DP.
\begin{lemma}\label{lem:nicenormal}
The decomposition of Lemma~\ref{lem:outertw} can be put in nice form, with $O(n)$ nodes, bags of size at most three, and root bag $\{u\}$. In this form $u\notin I_t$ for every node $t$.
\end{lemma}
\begin{proof}
Choose a bag containing $u$ and attach a new root bag $\{u\}$ to it. Root the resulting tree there. On a parent--child edge, replace the change of bag by a path which first goes through their intersection and then inserts the remaining vertices one at a time. Read bottom-up, each step is an introduce or forget step. No intermediate bag is larger than one of the original endpoint bags. Replace a node with several children by a binary tree of copies of its bag, giving join nodes. Beneath each remaining nonempty leaf, add a path down to the empty bag. All changes preserve the three decomposition axioms. Bag size is at most three, so each original edge or leaf introduces only constantly many new nodes; binary branching likewise adds only linearly many nodes.

If $u\in V_t$, some descendant bag contains $u$. The root bag also contains $u$. The connectedness axiom places $u$ in every bag along the path between them, in particular in $X_t$. Therefore $u\in V_t$ implies $u\in X_t$, which is exactly $u\notin I_t$.
Suppress any unary node whose bag equals its child's bag, including one created by adjoining a new root $\{u\}$ to an existing bag $\{u\}$. Contracting such an edge preserves all decomposition axioms, width and the root bag. Thus every remaining node is one of the four types in Definition~\ref{def:nice}.
\end{proof}

For later use, the private sets satisfy the following literal identities. At an introduce node, $I_t=I_c$ and the introduced $x$ does not occur in $V_c$. At a forget node, $I_t=I_c\mathbin{\dot\cup}\{x\}$. At a join,
\[
 I_t=I_{c_1}\mathbin{\dot\cup}I_{c_2}.
\]
For the last identity, any vertex occurring in both child subtrees must also occur in the join bag by connectedness of its occurrence subtree; hence it is private in neither child.

\subsection{Why Three Boundary Representatives Suffice}
\begin{lemma}\label{lem:bagsep}
Every path from a vertex of $I_t$ to a vertex of $V\setminus I_t$ contains a vertex of $X_t$.
\end{lemma}
\begin{proof}
A vertex in $I_t$ appears in a descendant bag but not in $X_t$. All its occurrences must remain strictly below $t$, because an occurrence outside the subtree would force it into $X_t$ along the connecting path of bags. Consider the first edge $ab$ of a purported path leaving $I_t$. Some bag contains both endpoints. Since $a\in I_t$, this bag lies below $t$, so $b\in V_t$. But $b\notin I_t$, and therefore $b\in X_t$. Every such path consequently meets the bag.
\end{proof}

\begin{lemma}\label{lem:bagdom}\label{m:bagdom}
Fix a node $t$.
\begin{enumerate}[label=(\roman*)]
\item 
Let $E_t\subseteq V\setminus I_t$ be nonempty and set $\alpha_t(b)=\best_b(E_t)$ for each $b\in X_t$. Every private voter has its best member of $E_t$ in $\im\alpha_t$.
\item Let $F_t\subseteq I_t$, and set $\pi_t(b)=\best_b(F_t)$, allowing $\bot$ if $F_t=\varnothing$. Every voter outside $I_t$ has its best member of $F_t$, when there is one, in $\Pi_t:=\im\pi_t\setminus\{\bot\}$.
\end{enumerate}
\end{lemma}

\begin{proof}
For (i), take $v\in I_t$ and $q\in E_t$. A shortest $v$--$q$ path meets $X_t$ at some $b$, and its subpaths are shortest, so
\[
 d_G(v,q)=d_G(v,b)+d_G(b,q).
\]
Let $p=\alpha_t(b)$. By definition, $\kappa_b(p)\le_{\lex}\kappa_b(q)$. If $p$ is strictly closer at $b$, then
\[
 d_G(v,p)\le d_G(v,b)+d_G(b,p)<d_G(v,q).
\]
If the distances at $b$ are equal, then $\id(p)\le\id(q)$ and the same triangle inequality gives $d_G(v,p)\le d_G(v,q)$. In either case $\kappa_v(p)\le_{\lex}\kappa_v(q)$. Apply this to the unique best exterior candidate $q$ from $v$. Since $p$ also belongs to $E_t$, uniqueness forces $p=q$. Thus the exterior favorite is retained among the representatives.

For (ii), there is nothing to prove if $F_t$ is empty. Otherwise take an outside voter $v$ and $q\in F_t$. A shortest path enters the private part through some $b\in X_t$, and the same distance equality holds. Replace $q$ at that boundary by $p=\pi_t(b)$ and repeat the two cases above. Applying the result to the unique best member of $F_t$ proves the claim. Notice that all distances are global distances in $G$; paths are not restricted to decomposition pieces.
\end{proof}

\begin{definition}\label{def:exposed}
For $F_t\subseteq A\cap I_t$, the exposed internal candidates are $\Pi_t=\im\pi_t\setminus\{\bot\}$, where $\pi_t(b)=\best_b(F_t)$. The hidden internal candidates are $F_t\setminus\Pi_t$.
\end{definition}

\begin{corollary}\label{cor:hiddenfinal}
In any full election having internal candidates $F_t$, no voter outside $I_t$ supports a hidden member of $F_t$. Its complete score is therefore the number of its supporters in $I_t$.
\end{corollary}
\begin{proof}
Lemma~\ref{lem:bagdom}(ii) gives an exposed internal candidate that strictly beats each hidden one at every outside voter. The exposed candidate remains active in the full election, regardless of the exterior candidates. A hidden candidate cannot be that voter's favorite.
\end{proof}

\subsection{Realizable Exterior Contexts and Exact State Semantics}
\begin{definition}\label{def:context}
A map $\alpha_t:X_t\to D\cup\{\bot\}$ is called an
\emph{exterior context}. It is \emph{realizable} when a set
    $E_t\subseteq D\setminus I_t,\qquad u\in E_t$,
satisfies
    $\alpha_t(b)=\best_b(E_t)
    \qquad\text{for every }b\in X_t.$
We say that such a set $E_t$ realizes the context $\alpha_t$. The compulsory target is always outside $I_t$ by Lemma~\ref{lem:nicenormal}.
\end{definition}

\begin{lemma}\label{lem:realizable}
Let $E_\alpha=\{u\}\cup(\im\alpha_t\setminus\{\bot\})$. A proposed map is realizable if and only if all its nonempty values lie in $D\setminus I_t$ and
\[
 \alpha_t(b)=\best_b(E_\alpha)\qquad(b\in X_t).
\]
\end{lemma}
\begin{proof}
If $E_t$ realizes the map, every listed representative belongs to $E_t$ and so does $u$. Hence $E_\alpha\subseteq E_t$. At each bag vertex, the unique best candidate in $E_t$ was retained in $E_\alpha$, so it remains best after the restriction. This proves necessity. Conversely the displayed conditions say that $E_\alpha$ is itself a valid exterior set containing $u$ and inducing the map, proving sufficiency. On a nonempty bag a realizable map has no $\bot$ value, since any set realizing it contains $u$ and is therefore nonempty. An empty bag has the unique empty map.
\end{proof}

Fix a realizable context and $F_t\subseteq A\cap I_t$. Put
\[
 C_t=(\im\alpha_t\setminus\{\bot\})\cup F_t,
 \qquad s_t(q)=|\{v\in I_t:\best_v(C_t)=q\}|.
\]
The state induced by these choices is
\begin{equation}
 \Sigma_t(\alpha_t,F_t)=(\pi_t,\nu_t,\mu_t,\rho_t),\label{eq:bagstate}
\end{equation}
with the following fields:
\begin{align*}
 \pi_t(b)&=\best_b(F_t), & \Pi_t&=\im\pi_t\setminus\{\bot\},\\
 R_t&=(\im\alpha_t\cup\im\pi_t)\setminus\{\bot\},
 &\nu_t(q)&=s_t(q)\quad(q\in R_t),\\
 \mu_t&=\minlex\{(s_t(q),-\id(q)):q\in F_t\setminus\Pi_t\},
 &\rho_t&=\ind{F_t\ne\varnothing}.
\end{align*}

We call the candidates in $R_t$ \emph{visible at node $t$}, meaning that their identities and their current vote counts are stored explicitly in the state.
The notation $R_t$ refers to candidate identities, not to a region of voters. It contains at most six candidates.

\begin{lemma}\label{lem:repexact}
For every set $E_t$ that realizes $\alpha_t$, every $F_t\subseteq A\cap I_t$, and every $v\in I_t$,
\[
 \best_v(E_t\cup F_t)=\best_v(C_t).
\]
Thus the state scores depend on the exterior only through $\alpha_t$.
\end{lemma}
\begin{proof}
If the unique full favorite is internal, it is retained. If it is exterior, Lemma~\ref{lem:bagdom}(i) says that it is one of the representatives, so it is also retained. Removing other candidates cannot change a favorite that remains. Therefore every private assignment is identical, and all private score totals agree. For hidden internal candidates those private scores are also complete scores by Corollary~\ref{cor:hiddenfinal}.
\end{proof}

The algorithm must store \emph{every} achievable tuple under each realizable context. Storing only one tuple per context could discard a necessary score distribution. On the other hand, equal tuples can be deduplicated: the following transition proofs show that hidden identities need not be retained beyond their key.

\subsection{Leaf and Introduce Transitions}

\begin{lemma}\label{lem:outerleaf}
At an empty-bag leaf, the only state is $(\varnothing,\varnothing,\infty,0)$ under the empty context.
\end{lemma}
\begin{proof}
The union of descendant bags is empty, so $I_t=\varnothing$. Its only internal opponent set is empty, it has no private voters, and both maps and the hidden family are empty. Each state field follows directly from~\eqref{eq:bagstate}.
\end{proof}

\begin{lemma}\label{lem:outerintroduce}
Let $X_t=X_c\cup\{x\}$. For every realizable parent context use $\alpha_c=\alpha_t|_{X_c}$. Each child state produces the parent state
\begin{align*}
 \pi_t(b)&=\pi_c(b)\quad(b\in X_c), &
 \pi_t(x)&=\best_x(\Pi_c),\\
 \nu_t(q)&=\nu_c(q)\quad(q\in R_t), &
 \mu_t&=\mu_c,\qquad \rho_t=\rho_c,
\end{align*}
where an absent child score entry is zero. This is exactly the state for the same internal set $F_t=F_c$.
\end{lemma}
\begin{proof}
Here $I_t=I_c$, and any set $E_t$ that realizes the parent context also realizes the restricted child context. Lemma~\ref{lem:repexact} shows that the two lists of representatives give exactly the same private votes in the complete election $E_t\cup F_c$.

At old bag vertices, the best internal candidate is unchanged. The new bag vertex $x$ lies outside $I_c$. If $F_c=\varnothing$, set $\pi_t(x)=\bot$. Otherwise, Lemma~\ref{lem:bagdom}(ii) implies that $\best_x(F_c)$ belongs to $\Pi_c=\im\pi_c\setminus\{\bot\}$. Hence $\pi_t(x)=\best_x(\Pi_c)$, as displayed above. The child state therefore determines the new entry without storing $F_c$. Every old exposed identity stays exposed because its witnessing old bag vertex remains. No previously hidden identity becomes exposed, by the same lemma. Hence $\Pi_t=\Pi_c$ and the hidden summary is unchanged.

Every candidate visible at the parent was already visible at the child, except possibly an exterior candidate newly listed because of the additional bag vertex $x$. If such an exterior candidate was not visible at the child, then it received no vote from $I_c=I_t$: by Lemma~\ref{lem:bagdom}(i), every exterior candidate that receives a private vote is represented among the child's exterior choices. Its missing child contribution is therefore correctly interpreted as zero.

No voter becomes private at an introduce node, and no candidate is added to or removed from the internal set. The displayed formulas therefore give exactly the parent state.
\end{proof}

\subsection{Forget Transitions: One voter and At Most One Opponent}

\begin{lemma}\label{lem:outerforget}
Let $X_c=X_t\cup\{x\}$, so $I_t=I_c\mathbin{\dot\cup}\{x\}$. For a parent context branch on $\eta\in\{0,1\}$, allowing $\eta=1$ only if $x\in A$. Put
\begin{align}
 \alpha_c(b)&=\best_b(\{\alpha_t(b)\}\cup\{x:\eta=1\})&& (b\in X_t),\label{eq:forgetcontext1}\\
 \alpha_c(x)&=\begin{cases}x,&\eta=1,\\
 \best_x(\im\alpha_t\setminus\{\bot\}),&\eta=0.
 \end{cases}\label{eq:forgetcontext2}
\end{align}
For each child state under this context, calculate
\begin{align}
 \pi_t(b)&=\best_b(\{\pi_c(b)\}\cup\{x:\eta=1\})&&(b\in X_t),\label{eq:forgetpi}\\
 w_x&=\best_x((\im\alpha_c\cup\Pi_c\cup\{x:\eta=1\})\setminus\{\bot\}),\label{eq:forgetvote}\\
 \nu_t(q)&=\nu_c(q)+\ind{w_x=q}&&(q\in R_t),\label{eq:forgetnu}\\
 \rho_t&=\rho_c\vee(\eta=1).\label{eq:forgetrho}
\end{align}
With $H_{\rm new}=(\Pi_c\setminus\Pi_t)\cup\{x:\eta=1,\ x\notin\Pi_t\}$, put
\begin{equation}
 \mu_t=\minlex\left(\{\mu_c\}\cup
 \{(\nu_c(q)+\ind{w_x=q},-\id(q)):q\in H_{\rm new}\}\right).\label{eq:forgetmu}
\end{equation}
The result is exactly the parent state induced by $F_t=F_c\cup\{x:\eta=1\}$.
\end{lemma}
\begin{proof}
The target is never private, so $x\ne u$. In any parent internal set, $x$ is either absent or selected as an allowed opponent. These are exactly the two branches.

Let $E_t$ realize the parent context. The corresponding child exterior is $E_c=E_t\cup\{x:\eta=1\}$. It contains $u$ and lies outside $I_c$. At a parent bag vertex, the best member is the better of $\alpha_t(b)$ and $x$ when selected, proving~\eqref{eq:forgetcontext1}. At $x$, a selected self-candidate is uniquely closest. If $x$ is not selected, then $x\in I_t$ and Lemma~\ref{lem:bagdom}(i), applied at the parent, gives its best exterior member among the parent representatives. This proves~\eqref{eq:forgetcontext2} and also proves that the child context is realizable.

The complete active sets agree literally:
\[
 E_c\cup F_c=E_t\cup F_t.
\]
Thus every old private voter is evaluated in the same election by Lemma~\ref{lem:repexact}. The branch does not add a previously unrepresented candidate to an old election; it transfers the status of $x$, if active, from exterior to internal. This distinction is why old private score contributions stay valid.

The formula for $\pi_t$ compares the best old internal candidate with the only possibly new internal identity. The newly private voter is $x$. Because $x\in X_c$, $\alpha_c(x)$ is its best exterior member, and Lemma~\ref{lem:bagdom}(ii) gives its best internal member among $\Pi_c$. Hence~\eqref{eq:forgetvote} gives its actual vote. Every parent-visible candidate has its old private score plus this possible vote, proving~\eqref{eq:forgetnu}. Missing child entries are zero because an omitted exterior candidate cannot receive child-private votes and a parent-exposed internal candidate cannot previously have been hidden.

It remains to verify the hidden family, rather than only its minimum. A child-hidden candidate stays hidden: the bag becomes smaller and the internal set can only gain $x$, so an identity beaten at every old bag vertex cannot become a best candidate at any remaining bag vertex. Its score stays fixed because the new voter $x$ was outside $I_c$ and cannot support it by Corollary~\ref{cor:hiddenfinal}. The other parent-hidden candidates are exactly the previously exposed identities that lose exposure, and $x$ itself when it is active but not exposed. These are precisely $H_{\rm new}$. They have the score displayed in~\eqref{eq:forgetmu}, including their possible newly private vote. Taking the minimum of those keys and the old hidden minimum is therefore exact. The presence flag is the stated disjunction.
\end{proof}

\subsection{Join Transitions and Identity Consistency}
\begin{lemma}\label{lem:outerjoin}
At a join let the common bag be $X$ and let $I_t=I_1\mathbin{\dot\cup}I_2$. Fix a realizable parent context $\alpha_t$. A pair of child states, under contexts $\alpha_1,\alpha_2$, is compatible when for $i\ne j$ and every $b\in X$,
\begin{equation}
 \alpha_i(b)=\best_b(\{\alpha_t(b),\pi_j(b)\}).\label{eq:joincontext}
\end{equation}
For compatible states set
\begin{align}
 \pi_t(b)&=\best_b(\{\pi_1(b),\pi_2(b)\}),\label{eq:joinpi}\\
 \nu_t(q)&=\nu_1(q)+\nu_2(q)\qquad(q\in R_t),\label{eq:joinnu}\\
 H_{\rm new}&=(\Pi_1\cup\Pi_2)\setminus\Pi_t,\notag\\
 \mu_t&=\minlex\left(\{\mu_1,\mu_2\}\cup
 \{(\nu_1(q)+\nu_2(q),-\id(q)):q\in H_{\rm new}\}\right),\label{eq:joinmu}\\
 \rho_t&=\rho_1\vee\rho_2.\label{eq:joinrho}
\end{align}
These are exactly the parent fields for $F_t=F_1\cup F_2$.
\end{lemma}
\begin{proof}
Let $E_t$ realize the parent context, and choose sets $F_1\subseteq A\cap I_1$ and $F_2\subseteq A\cap I_2$ that induce the two child states. They are disjoint because they lie in disjoint private sets. The exterior for child $i$ is $E_i=E_t\cup F_j$. Its best member at $b$ is the better of the best member of $E_t$, which is $\alpha_t(b)$, and the best member of $F_j$, which is $\pi_j(b)$. Thus~\eqref{eq:joincontext} is exactly its true context. This argument also proves realizability of the contexts used by every true parent set.

The election on either child has complete active set
\[
 E_i\cup F_i=E_t\cup F_1\cup F_2.
\]
Therefore the two summaries evaluate the same election on disjoint groups of voters. The best internal candidate at a bag vertex is the better of the two child-internal best candidates, proving~\eqref{eq:joinpi}. Visible score contributions add by candidate identity, proving~\eqref{eq:joinnu}.

A child-hidden candidate has no votes outside that child by Corollary~\ref{cor:hiddenfinal}. It stays hidden at the parent because it remains dominated at each unchanged bag vertex. A parent-hidden candidate not already hidden in a child must have been exposed in its own child and then lost exposure; it is exactly an element of $H_{\rm new}$. It can have private voters in both children. If it does, the other child sees it as an exterior representative; its contribution is therefore in that child's visible map. Add both contributions before finalizing it, as in~\eqref{eq:joinmu}.

A missing visible entry contributes zero: a hidden internal candidate cannot become parent-exposed, and an omitted exterior candidate cannot receive that child's private votes. There is no collision between two hidden internal identities, since their locations would have to lie in both private sets. No bag voter is counted by either child, so no subtraction for bag overlap is needed. Finally the parent has an opponent exactly when at least one child does, proving~\eqref{eq:joinrho}.
\end{proof}

\subsection{Operational Overview and Algorithm}

The tree algorithm relies on nested or disjoint
preference regions. This property need not hold when cycles are present, so we use a different algorithm for outerplanar graphs. We use a bottom-up dynamic program on a rooted tree decomposition with bags of at most three vertices. At each bag we record which participating candidates are preferred at its vertices and the vote counts needed to combine the parts. After deciding $\Kill(G,u,A)$, Theorem~\ref{thm:transfer} gives the exclusion-zone algorithms.
Outerplanar graphs have treewidth at most two~\cite{BondyMurty}. Lemma~\ref{lem:outertw} gives a constructive proof giving a tree decomposition with bags of at most three vertices. Convert it to the rooted form constructed in this section, with root bag $\{u\}$. For a decomposition node $t$, let $X_t$ be its bag and $V_t$ contain the graph vertices appearing in bags at $t$ or its descendants; put $I_t=V_t\setminus X_t$. We call $I_t$ the {\it private set} of $t$. Every path from $I_t$ to $V\setminus I_t$ meets $X_t$, and $u\notin I_t$. Lemma~\ref{lem:nicenormal} constructs the required rooted decomposition, and Lemma~\ref{lem:bagsep} proves the separation property.
Here internal means belonging to $I_t$, and exterior means belonging to $V\setminus I_t$, including the bag vertices $X_t$. Every path between these sets meets $X_t$. Among participating exterior candidates, a voter in $I_t$ has its favourite among the candidates preferred at the bag vertices; 
among participating internal candidates, an exterior voter has its favourite among those preferred at the bag vertices. 
Formally, we prove: 

If $q\in F_t$ is not the preferred member of $F_t$ at any bag vertex, then no voter outside $I_t$ can support $q$. Once votes from $I_t$ are counted, its score is complete; the DP needs only the smallest elimination key of such finalized candidates.
Not every possible assignment of candidates to the bag vertices can arise from a genuine set of participating candidates outside $I_t$. The dynamic program therefore keeps only assignments that can actually occur.
For a proposed map $\alpha_t$, let
   $E_\alpha=\{u\}\cup
   \bigl(\operatorname{im}(\alpha_t)\setminus\{\bot\}\bigr),$
where $\operatorname{im}(\alpha_t)$ denotes the candidates appearing as values of the map. We keep the map only if every listed candidate belongs to $D\setminus I_t$ and
   $\alpha_t(b)=\best_b(E_\alpha)
   \qquad\text{for every }b\in X_t.$
These conditions are sufficient because $E_\alpha$ itself is then a set of outside participating candidates producing exactly the proposed choices at the bag vertices.
For each assignment $\alpha_t$ that passes this test, the table stores every distinct feasible record $(\pi_t,\nu_t,\mu_t,\rho_t)$. A record is feasible when some $F_t\subseteq A\cap I_t$ induces it; the transitions construct these records without enumerating all such sets. The map $\pi_t$ gives the preferred member of $F_t$ at each bag vertex (or $\bot$ when $F_t$ is empty). Put $R_t=(\im\alpha_t\cup\im\pi_t)\setminus\{\bot\}$, which contains at most six candidates. For each $q\in R_t$, the map $\nu_t$ stores its votes from $I_t$ in the election with candidate set $F_t\cup(\im\alpha_t\setminus\{\bot\})$. The hidden internal opponents are the candidates in $F_t\setminus\operatorname{im}(\pi_t)$. The value $\mu_t$ is the smallest complete elimination key among them, or $\infty$ if this set is empty. The flag $\rho_t$ records whether $F_t\ne\varnothing$. For a fixed node $t$, write $\alpha$ for an exterior assignment that passes the test above, and let $\mathcal D_t(\alpha)$ be the set of feasible records stored for that assignment. An exterior assignment passing the preceding test is also called a realizable exterior context. A node with one child $c$ introduces $x$ when $X_t=X_c\cup\{x\}$ and forgets $x$ when $X_c=X_t\cup\{x\}$, with $x$ the sole vertex in the bag difference. A join has two children $c_1,c_2$ with $X_t=X_{c_1}=X_{c_2}$. Leaves have empty bags. At a join, child records with internal maps $\pi_1,\pi_2$ and exterior contexts $\alpha_1,\alpha_2$ are \emph{compatible} with parent context $\alpha$ when $\alpha_i(b)=\best_b(\{\alpha(b),\pi_j(b)\})$ for both $i\ne j$ and every common bag vertex $b$. Thus each child's exterior includes the other child's internal candidates; Lemma~\ref{lem:outerjoin} and~\eqref{eq:joincontext} give the exact join rule. The recurrence is:

\[
\mathcal D_t(\alpha)=
\begin{cases}
\{(\varnothing,\varnothing,\infty,0)\},
&\text{if $t$ is an empty-bag leaf},\\[1mm]
\operatorname{Introduce}_{x}\!\left(\mathcal D_c(\alpha|_{X_c})\right),
&\text{if $t$ introduces $x$},\\[1mm]
\displaystyle\bigcup_{\substack{\varepsilon\in\{0,1\}\\\varepsilon=0\ \text{or}\ x\in A}}
\operatorname{Forget}_{x,\varepsilon}\!\left(\mathcal D_c(\alpha_\varepsilon)\right),
&\text{if $t$ forgets $x$},\\[3mm]
\displaystyle\bigcup_{\substack{
\alpha_1,\alpha_2\ \text{realizable child contexts}\\ R_1\in\mathcal D_{c_1}(\alpha_1),\ R_2\in\mathcal D_{c_2}(\alpha_2)\\ (\alpha_1,R_1,\alpha_2,R_2)\text{ compatible with }\alpha
}}
\operatorname{Join}(R_1,R_2),
&\text{if $t$ joins $c_1,c_2$}.
\end{cases}
\]

Here $\varepsilon$ selects the forgotten vertex only when it is allowed, and $\alpha_\varepsilon$ is the induced child context. The operators return sets of parent records; dependence on $t,\alpha$ is suppressed. Introduce updates boundary information, forget counts the newly private voter and finalizes newly hidden opponents, and join adds compatible contributions on disjoint private sets. All counts concern the represented election's first round; bag votes remain deferred. Section~\ref{sec:outer} defines the operators and proves their exactness.

\begin{algobox}{{III. OuterplanarCandidateLoss}}
\algio{A connected outerplanar graph $G$, target $u$, and allowed opponents $A$.}{Whether $\Kill(G,u,A)=\yes$.}
\begin{enumerate}[leftmargin=5mm,itemsep=.35ex,topsep=.45ex]
\item Construct the rooted decomposition described above, with bags of at most three vertices and root bag $\{u\}$.
\item For every bag, enumerate realizable exterior maps $\alpha_t$ and initialize feasible records $(\pi_t,\nu_t,\mu_t,\rho_t)$ at leaves.
\item Process the decomposition from leaves to root using the three operations described above. Retain compatible records, count each voter other than $u$ when its vertex leaves a bag, and finalize an internal opponent only when no bag vertex prefers it among the internal opponents.
\item For each root record, add $u$'s deferred vote to its stored target score. Return true if some record contains an opponent and gives the target an elimination key smaller than every opponent key. Otherwise return false. If a witness is required, recover it from an accepting record before returning true.
\end{enumerate}
\end{algobox}

\subsection{Exactness, Root Acceptance and Time Complexity}
\begin{lemma}\label{thm:bwexact}
For every node $t$, every realizable context $\alpha_t$, and every tuple $\Sigma$, the algorithm stores $\Sigma$ at $(t,\alpha_t)$ if and only if $\Sigma=\Sigma_t(\alpha_t,F_t)$ for some $F_t\subseteq A\cap I_t$. No nonrealizable context is used.
\end{lemma}
\begin{proof}
Contexts are enumerated and tested by Lemma~\ref{lem:realizable}, proving the last assertion. For the state invariant, induct bottom-up through the nice decomposition.

At a leaf Lemma~\ref{lem:outerleaf} gives both inclusions. At an introduce node the private set does not change. The possible internal candidate sets are identical because $I_t=I_c$. Fix a set $E_t$ that realizes the parent context. The same set also realizes the restricted child context, so Lemma~\ref{lem:outerintroduce} transports the exact state for each internal set. This does not assert that every set realizing the child context also realizes the parent context. At a forget node a parent internal set has a unique choice of whether $x$ belongs to it and a unique restriction to $I_c$. Lemma~\ref{lem:outerforget} covers both choices, produces the correct context, and gives all exact fields. Conversely every transition is realized by the union prescribed in that lemma.

At a join, every parent set decomposes uniquely into its intersections with $I_1$ and $I_2$. Its child contexts satisfy~\eqref{eq:joincontext}, and the induction hypothesis provides the child states. Lemma~\ref{lem:outerjoin} combines them into the required exact parent state. Conversely, by induction each compatible child state corresponds to some set of participating candidates in that child's private
set. These sets can be chosen independently because the two private sets are disjoint. Their union gives the participating candidate set for the parent, and the same lemma proves the resulting fields. No hidden identity is required for testing compatibility beyond its location in its own private set and its stored finalized minimum. This completes the induction.
\end{proof}

At the root $r$, use $X_r=\{u\}$ and $\alpha_r(u)=u$. Its private set is $V\setminus\{u\}$. The missing voter is the target's own voter and it supports the target. Thus, for every root state,
\[
 s_u=\nu_r(u)+1,
 \qquad
 \mu_{\rm opp}=\minlex\left(\{\mu_r\}\cup
        \{(\nu_r(q),-\id(q)):q\in\Pi_r\}\right).
\]
Accept exactly when $\rho_r=1$ and $(s_u,-\id(u))<_{\lex}\mu_{\rm opp}$.

\begin{lemma}\label{lem:outerroot}
A root state is accepting if and only if its realizing internal set eliminates $u$ in the first round.
\end{lemma}
\begin{proof}
Every voter except $u$ is private at the root, and $u$'s vote goes to itself. The displayed target score is therefore exact. Opponents partition into exposed identities and hidden identities; the displayed minimum uses their exact scores. The flag excludes the empty-opponent case. Lemma~\ref{lem:accept} now gives the stated equivalence.
\end{proof}

\begin{theorem}[Outerplanar tractability]
\label{m:outertheorem}\label{thm:outer}
For every connected outerplanar graph $G$, target $u$ and allowed set $A$, a direct implementation decides $\Kill(G,u,A)$ 
in $O(n^{32})$ word operations and $O(n^{15})$ words of space. Zone verification takes $O(n^{33})$ time; 
minimum-zone computation and enumeration of all nonempty exclusion zones take $O(n^{34})$ time.
\end{theorem}
\begin{proof}
Obtain the nice decomposition from Lemmas~\ref{lem:getouter}--\ref{lem:nicenormal}, run the transitions, and test the root. Lemma~\ref{thm:bwexact}, Lemma~\ref{lem:outerroot}, and the round-one reduction prove soundness and completeness.

For the explicit count, each bag has at most three vertices. There are at most $(n+1)^3$ context maps and $(n+1)^3$ internal exposed maps. At most six identities are visible, giving at most $(n+1)^6$ visible score vectors. A hidden key has at most $1+n(n+1)\le(n+1)^2$ possibilities, including the sentinel, and there are two presence flags. Hence a full node table, including all contexts, has at most
\[
 T_0=2(n+1)^{14}
\]
entries. Each context is checked against at most four candidates at each of at most three bag vertices, so its realizability test costs a constant number of precomputed-distance comparisons.

Introduce and forget transitions enumerate only a parent context, a child entry and at most two branches. A deliberately unoptimized join enumerates a parent context and a pair of child entries. There are at most $(n+1)^3T_0^2=O((n+1)^{31})$ such combinations per node; every combination is checked with constant-size maps and at most six score entries. There are $O(n)$ nice nodes. This proves the stated time bound, which includes the lower preprocessing cost. Storing all node tables uses $O(nT_0)=O((n+1)^{15})$ words. These are finite polynomial bounds, not claims of an efficient practical implementation.
\end{proof}

\begin{corollary}\label{cor:outer}
On connected outerplanar graphs, exclusion-zone verification is in $\Pclass$, minimum-zone computation is in $\FP$, and all nonempty exclusion zones can be enumerated in polynomial time.
\end{corollary}
\begin{proof}
Theorem~\ref{m:outertheorem} applies to every allowed set, in particular the regions $V\setminus S$ used by Theorem~\ref{thm:transfer}. That theorem gives each asserted algorithm and bounds its number of calls to $\Kill$.
\end{proof}

The outerplanar proof does not assume a bound on graph distances or degrees, does not count only bag voters, and does not evaluate shortest paths in an augmented polygon. Only a constant number of candidate identities need to be stored
at a bag, even though the set of voters below that bag may be arbitrarily large.

\section{IRV on Planar Graphs}\label{m:planar}\label{sec:planar}

Planar graphs need not have treewidth bounded by a fixed constant. Instead, we exploit a structure created by an unknown inclusion-minimal first-round witness, and enumerate separator descriptions without enumerating all candidate sets. A \emph{separator descriptor} is a constant-size encoding of an oriented separating curve, listing candidate identities (guards), triangular-face contacts and the information needed to reconstruct paths in fixed shortest-path trees (Definition~\ref{def:descriptor}); its expanded boundary may be long. This is a dynamic program implemented by finite-state saturation over a polynomially enumerable family of these descriptors, rather than a bottom-up computation on a tree decomposition constructed in advance: starting from leaf records, it repeatedly applies local rules combining one or two child records until no new descriptor--state pair is produced. The completeness proof shows that the hierarchy associated with an inclusion-minimal witness is recovered, while the soundness proof constructs a genuine witness from an accepting pair of stored root records. The construction and correctness proofs follow in this section.

This section gives the complete proof of the planar algorithm. The argument has four stages.

First, starting from an inclusion-minimal first-round witness, we show that every other participating candidate has a cell adjacent to the target's cell. Contracting the cells therefore gives a
planar graph of bounded treewidth.

Second, because the algorithm does not know the witness or its cells, we construct a bounded-width auxiliary plane graph only for the completeness proof and use its branch decomposition to obtain
a hierarchy of small separating curves.

Third, we show that all separators needed by this hierarchy belong to a polynomial-size family that can be generated directly from  the input graph, without knowing the witness.

Finally, we give local compatibility tests and a dynamic program over these generated separators. Soundness constructs an actual candidate set from any accepting computation, while completeness
shows that the separators arising from every minimal witness are among those generated by the algorithm.

All election distances throughout this section are the original distances in $G$.

\subsection{Shortest-Path Structure and Inclusion-Minimal Witnesses}
\begin{lemma}\label{lem:geodesic}\label{m:persistence}
Let $a,b$ be distinct potential candidates.
\begin{enumerate}[label=(\roman*)]
\item If $\kappa_v(a)<_{\lex}\kappa_v(b)$ and $x$ lies on a shortest $v$--$a$ path, then $\kappa_x(a)<_{\lex}\kappa_x(b)$.
\item If $\kappa_z(a)<_{\lex}\kappa_z(b)$ and $z$ lies on a shortest $b$--$x$ path, then $\kappa_x(a)<_{\lex}\kappa_x(b)$.
\end{enumerate}
Consequently, a deterministic cell contains every vertex of every shortest path from any of its vertices to its candidate.
\end{lemma}

\begin{proof}
For (i), the shortest-path assumption gives
\[
 d_G(v,a)=d_G(v,x)+d_G(x,a),\qquad
 d_G(v,b)\le d_G(v,x)+d_G(x,b).
\]
If $d_G(v,a)<d_G(v,b)$, subtracting $d_G(v,x)$ gives
\[
 d_G(x,a)<d_G(v,b)-d_G(v,x)\le d_G(x,b).
\]
If the distances at $v$ are equal, their preference comparison instead gives $\id(a)<\id(b)$ and the same calculation with weak inequalities gives $d_G(x,a)\le d_G(x,b)$. A strict inequality at $x$ settles the comparison by distance; equality at $x$ retains the same identifier advantage. This proves (i).

For (ii), use
\[
 d_G(x,b)=d_G(x,z)+d_G(z,b),\qquad
 d_G(x,a)\le d_G(x,z)+d_G(z,a).
\]
A strict distance advantage for $a$ at $z$ makes the second expression strictly smaller than the first. If the distances at $z$ agree, $\id(a)<\id(b)$, and at $x$ either the distance inequality is strict or this unchanged identifier advantage resolves the tie. This proves (ii).

If $v\in\cell_K(a)$, it prefers $a$ to every other candidate in $K$. Apply (i) against each rival along the chosen shortest path. Every vertex of that path has the same unique owner $a$. The argument applies to any shortest path, not only one selected path.
\end{proof}

\begin{definition}\label{def:minwitness}
A nonempty $F\subseteq A$ is an inclusion-minimal first-round witness if it eliminates $u$ in the first round and no nonempty proper subset of $F$ does so. The empty subset is never a first-round witness. By finiteness and Lemma~\ref{lem:restart}, such a set exists whenever $\Kill(G,u,A)=\yes$.
\end{definition}

Fix such a witness for the next structural arguments and write $K=\{u\}\cup F$. This is an existence choice for the completeness proof; the algorithm is not given $F$.

\begin{lemma}\label{lem:indispensable}
For every $q\in F$,
\[
 \score_{K\setminus\{q\}}(u)>\score_K(u).
\]
\end{lemma}
\begin{proof}
When $F=\{q\}$, the candidate $q$ originally owns at least its own voter, so $\score_K(u)\le n-1$. After its deletion $u$ owns all $n$ voters. This handles the sole-opponent case without treating the one-candidate election as an elimination round.

Suppose $|F|\ge2$. By Lemma~\ref{lem:deletion}, the target score cannot decrease. If it did not increase, it would remain equal to its old score, and its identifier is unchanged, so its elimination key would be unchanged. Every surviving opponent's score weakly increases by the same lemma, so its key cannot become smaller. The target key was strictly smaller than every opponent key before the deletion. It would still be strictly smaller than every remaining opponent key afterwards. At least one opponent remains, so $F\setminus\{q\}$ would be a proper first-round witness, contradicting inclusion-minimality.
\end{proof}

\begin{lemma}\label{thm:adjacency}
For each $q\in F$, some original edge of $G$ has one endpoint in $\cell_K(q)$ and the other in $\cell_K(u)$.
\end{lemma}
\begin{proof}
The strict score increase in Lemma~\ref{lem:indispensable} provides a voter $v$ newly assigned to $u$ after deleting $q$. By Lemma~\ref{lem:deletion}, a voter whose old owner survives cannot change owner. Thus $v$ was originally assigned to $q$.

Take a shortest path
\[
 v=v_0,v_1,\ldots,v_t=u.
\]
In the election $K\setminus\{q\}$, every vertex on this path belongs to the target cell by Lemma~\ref{lem:geodesic}(i). In the original election, let $i$ be the last index whose vertex is owned by $q$. This index exists because $v_0$ is so owned, and it is smaller than $t$ because $u$ owns itself. The successor $v_{i+1}$ is not owned by $q$ originally. Its old owner therefore survives the deletion and must remain its owner. But after deletion its owner is $u$, so its old owner was already $u$. The edge $v_iv_{i+1}$ is the required edge. No assumption is made that the original path has only one transition between the two cells; choosing the last $q$-owned vertex avoids that assumption.
\end{proof}

\subsection{The Cell Quotient (not the Input Graph) Has Bounded Width}
\begin{definition}\label{def:cellquotient}

Contract the connected cells $\cell_K(c)$, one for each $c\in K$, to vertices labelled by their candidates. Delete resulting loops and identify parallel edges. The simple graph obtained is $Q_K$.
\end{definition}

Cell connectedness follows from Lemma~\ref{lem:geodesic}, and distinct cells are disjoint, so these contractions are well-defined.

\begin{lemma}[Planar universal-vertex quotient]\label{lem:quotientwidth}
The graph $Q_K$ is planar and connected, $u$ is adjacent to every other vertex of $Q_K$, and
\[
 \tw(Q_K-u)\le2,\qquad \tw(Q_K)\le3.
\]
\end{lemma}
\begin{proof}
Contract each cell along a spanning tree in a plane embedding. Contracting an embedded edge can be performed inside a narrow neighborhood of the edge and its endpoints; it introduces no crossing. Deleting loops or surplus parallel edges also preserves planarity. Connectivity is preserved by contraction. Lemma~\ref{thm:adjacency} gives an edge from $u$ to every other quotient vertex.

Delete $u$ in this embedding. Its incident faces merge into a face with an occurrence of every remaining vertex on its boundary, because every remaining vertex was adjacent to $u$. Thus $Q_K-u$ is outerplanar. Lemma~\ref{lem:outertw} supplies a decomposition with bags of size at most three. Insert $u$ into each bag. Every old edge remains covered, each edge incident with $u$ is covered wherever its other endpoint occurs, and all occurrences of $u$ are connected because it is now in every bag. The resulting bags have size at most four. This proves the two width bounds.
\end{proof}

\begin{proposition}\label{m:adjacency}
For an inclusion-minimal first-round witness, every opponent cell is adjacent to the target cell along an edge of the original graph. Consequently the cell quotient $Q_K$, obtained by contracting the cells and deleting loops and replacing parallel edges by a single edge, is planar and has $u$ adjacent to every other quotient vertex. Moreover, $\tw(Q_K-u)\le2$ and $\tw(Q_K)\le3$.
\end{proposition}
\begin{proof}

Fix $K=\{u\}\cup F$ and $q\in F$. Removing $q$ changes only votes previously assigned to $q$, so every surviving candidate's score weakly increases. In fact $u$'s score strictly increases. If $F=\{q\}$, the opponent owns at least its own voter, while after deletion $u$ receives all votes. If $|F|\ge2$ and the target's score did not increase, its elimination key would be unchanged and every remaining opponent's key would weakly increase. The target would still have the smallest key, making $F\setminus\{q\}$ a nonempty proper first-round witness, contrary to minimality.

Thus some voter $v$ originally assigned to $q$ switches to $u$ when $q$ is removed. Take a shortest path $v=v_0,\ldots,v_t=u$. By Lemma~\ref{m:persistence}, every vertex of this path belongs to $u$'s cell in the reduced election. Let $i$ be the last index for which voter $v_i$ originally supported $q$. Such an index exists because $v_0$ supported $q$, and $i<t$ because voter $u$ supports itself. The candidate originally supported by $v_{i+1}$ survives the deletion, so that voter's choice does not change. That voter supports $u$ after the deletion and therefore already supported $u$ before it. Therefore $v_iv_{i+1}$ is an original edge between $\cell_K(q)$ and $\cell_K(u)$. This argument does not assume a single transition between the cells along the path.
The cells are connected by Lemma~\ref{m:persistence}. Contracting them and deleting loops and surplus parallel edges preserves planarity and connectedness. Every opponent is adjacent to $u$ in the quotient. Deleting $u$ in this embedding merges its incident faces into one incident with every remaining vertex, so $Q_K-u$ is outerplanar and has treewidth at most two. Add $u$ to every bag of a decomposition of $Q_K-u$ of width at most two: old edges remain covered, each edge incident with $u$ is covered where its other endpoint occurs, and $u$'s occurrences are connected. The bags have at most four vertices, proving $\tw(Q_K)\le3$. 
For the details, 
see Lemmas~\ref{thm:adjacency} and~\ref{lem:quotientwidth}.
\end{proof}

Opponent cells can be adjacent to one another, so $Q_K$ need not be a star. Also, contracting large connected cells can reduce width substantially; the lemma does not bound $\tw(G)$. It does not by itself provide a polynomial algorithm because $K$ and its cells are unknown. The next stages address that missing algorithmic step explicitly.

\subsection{A Triangulation Used Only for Topology}

The algorithm first tests every opponent set of size one or two. There are at most $|A|+\binom{|A|}{2}$ such sets and their scores are polynomially computable. This also handles all graphs with at most three vertices. In the remaining completeness argument assume $n\ge4$ and $|F|\ge3$, so $|K|\ge4$.

We record the standard same-vertex plane triangulation construction and the size counts needed below. Added edges serve only the topology; voting distances remain those of $G$.
\begin{lemma}\label{lem:triangulate}
A connected simple plane graph on $n\ge4$ vertices has a simple triangulated plane supergraph $G^+$ on the same vertex set, obtainable in polynomial time by adding edges. Every face of $G^+$ is a triangle, and it has $2n-4$ faces and $3n-6$ edges.
\end{lemma}
\begin{proof}
Keep adding an edge that can be drawn inside a face between two currently nonadjacent vertices. There are only $\binom n2$ possible edges, and face incidences and the existence of such a drawing are decidable from the finite rotation system, so a direct implementation is polynomial.

A maximal graph obtained this way has no cutvertex: at a cutvertex, two consecutive incident edges entering different components after its deletion have nonadjacent endpoints and permit an edge in their common facial wedge. Thus every face boundary is a simple cycle. A face of length at least four has a missing diagonal. Indeed, for four cyclically ordered boundary vertices, both alternating diagonals cannot already be embedded outside that face, since in the complementary disk they would cross. At least one is absent and can be inserted inside the face. Maximality therefore implies that all face boundaries are triangles.

For a triangulation, count edge sides to get $3f=2m$. Euler's formula $n-m+f=2$ gives $m=3n-6$ and $f=2n-4$. The operations added edges but no vertices, as claimed.
\end{proof}

Obtain an embedding of $G$ by External result~\ref{ext:planarity} and form $G^+$. All election distances continue to be $d_G$. For each eligible candidate $c\in D$, choose one fixed shortest-path tree $T_c$ in the original graph. A deterministic construction runs breadth-first search from $c$ and, for each vertex other than $c$, chooses the smallest-identifier neighbor whose distance from $c$ is one less. Distances strictly decrease along parent edges, so this is a spanning tree, and each root path is an original shortest path.

\begin{lemma}\label{lem:celltrees}
For every active $c\in K$, the restriction
\[
 T_K(c):=T_c[\cell_K(c)]
\]
is a spanning tree of $\cell_K(c)$. For any two vertices $a,b$ of that cell, the unique $a$--$b$ path in $T_c$ is contained in the cell.
\end{lemma}
\begin{proof}
The root $c$ belongs to its cell. Every root path from a cell vertex is a shortest path to $c$ and hence stays in the cell by Lemma~\ref{lem:geodesic}. Those root paths make the restriction connected, and it is acyclic as a subgraph of a tree. Thus it spans the cell. The unique tree path between $a$ and $b$ is contained in the union of their two root paths, so it too stays in the cell. It need not be a shortest $a$--$b$ path in $G$; only its tree uniqueness and cell containment are used.
\end{proof}

Contracting the same cells in $G^+$ yields another simple planar quotient $Q_K^+$. Every original target adjacency remains, so the proof of Lemma~\ref{lem:quotientwidth} also gives $\tw(Q_K^+)\le3$. Added edges cannot alter owners because owners are still defined with $d_G$, not with distances in $G^+$.

\subsection{The Direct Prediagram and its Exact Face Correspondence}

The geometric dual of a connected plane graph has one vertex in each face and one edge crossing each primal edge, joining its incident face vertices; this standard plane-duality construction is recalled, for example, by Bondy and Murty~\cite[Section~10.2]{BondyMurty}. If an edge is incident with the same face on both sides, its dual is a loop; repeated face adjacencies give parallel edges. The dual embedding records occurrences, not just a simplified adjacency relation.

\begin{definition}\label{def:prediagram}
Start with the geometric dual $(G^+)^*$. For every edge in the disjoint forest $\bigcup_{c\in K}T_K(c)$, delete its dual edge. Then repeatedly delete any degree-one dual vertex together with its incident edge. The remaining embedded graph is $H_K$, called the prediagram.
\end{definition}

The prediagram is only a structural object for the unknown witness. It is not constructed by the algorithm deciding $\Kill$. Later, its vertices are represented by known triangular faces of $G^+$.

\begin{lemma}\label{lem:directprediagram}
The graph $H_K$ is a simple connected plane graph, all of whose vertices have degree two or three. It has exactly $|K|$ faces. These faces are in bijection with the active candidates: the face $\phi_c$ contains exactly the original graph vertices of $\cell_K(c)$ and contains the embedded restricted tree $T_K(c)$.
\end{lemma}
\begin{proof}
First consider the full dual. Every primal face is triangular, so every dual vertex has degree three. There are no dual loops because a triangulation has no bridge. There are no parallel dual edges: otherwise two distinct triangular faces would share at least two edges, forcing both to have the same three edges. Both sides of that triangle would then be faces, leaving no other vertex in a connected embedding, contrary to $n\ge4$. Thus the full dual is simple, cubic and connected.

Its faces correspond to original primal vertices. Delete the dual edges of the cell forest one at a time. At each step a forest edge joins two different components of the partially processed forest; otherwise the forest would contain a cycle. Its two endpoint vertex-faces therefore belong to different currently merged dual faces. Deleting that dual edge merges those faces. An edge with different faces on its two sides is not a bridge, so its deletion preserves dual connectedness. Induction proves that all deletions preserve connectedness and merge exactly the vertex-faces belonging to each cell tree, with no identification between different cells. The dual now has one face per candidate cell.

Removing a degree-one vertex removes a bridge. A bridge has the same face on both sides, so removing it does not merge different faces; it only removes a boundary excursion of one face. Accordingly the face correspondence and the containment of original vertices are unchanged. Restricted cell-tree edges do not cross any retained dual edge: the only dual edge that crossed each was deleted. Thus each restricted tree stays in the corresponding face throughout.

There are $|K|\ge4$ faces after the forest deletions. A connected acyclic graph has only one face, so a cycle remains. Leaf pruning cannot remove a cycle and cannot disconnect the retained graph. Its final graph is therefore nonempty, connected and has minimum degree at least two. Maximum degree three and simplicity are inherited from the full dual. The face correspondence survives each pruning step, proving all assertions.
\end{proof}

\begin{lemma}\label{lem:prediagramdual}
After loops and parallel edges are simplified, the graph $H_K^*$ is a subgraph of $Q_K^+$. In particular $\tw(H_K^*)\le3$.
\end{lemma}
\begin{proof}
A surviving dual edge of $H_K$ crosses an edge $ab$ of $G^+$. If it separates two different candidate faces $\phi_c,\phi_d$, then $a$ belongs to one associated cell and $b$ to the other. Throughout the forest-edge deletions, each original vertex-face is merged only with vertex-faces in the same candidate cell; pruning a degree-one dual edge merges no faces. Hence the two local sides of a surviving dual edge retain the cell labels of the endpoints of the primal edge it crosses. Thus this face adjacency is an edge $cd$ of $Q_K^+$. If both sides have the same label, its dual in $H_K^*$ is a loop; multiple incidences between two labels give parallel edges. Removing these does not increase or change the needed bag-size bound: a loop requires only its one endpoint in a bag, and a parallel edge requires the same pair as an existing edge. A decomposition of $Q_K^+$ therefore also covers the simplified dual and all its nonsimplified copies. Its width is at most three.
\end{proof}

An embedding is \emph{2-cell} if every face is homeomorphic to an open disk.

\begin{external}\label{ext:duality}
Let $\Lambda$ be a connected 2-cell embedding of an ordinary graph on the sphere, and let $\Lambda^*$ be its geometric dual, retaining loops and parallel edges. Then
\[
\tw(\Lambda^*)\le\tw(\Lambda)+1.
\]
This is the ordinary-graph sphere case of Mazoit's surface-duality theorem~\cite[Theorem~1]{Mazoit}: the sphere has Euler genus zero and every dual edge has at most two distinct endpoints.
\end{external}

The connected embedding of $H_K$ on the sphere is 2-cell, as is its geometric dual $H_K^*$. With multiplicities retained, the geometric dual of $H_K^*$ is $H_K$. Lemma~\ref{lem:prediagramdual} gives $\tw(H_K^*)\le3$: loops require only their endpoint in a bag, while parallel edges require the same pair. Applying External result~\ref{ext:duality} to $H_K^*$ therefore yields $\tw(H_K)\le4$.

\begin{definition}\label{def:branch}
A branch decomposition of a graph with at least two edges is an unrooted tree whose internal nodes have degree three and whose leaves are bijectively labelled by graph edges. Deleting a decomposition-tree edge partitions the graph-edge labels into nonempty sets $E_1,E_2$. Its middle set is the graph vertices incident with an edge in each part. The width is the largest middle-set size, and the branchwidth $\bw$ is the minimum such width.
\end{definition}

We recall the standard treewidth--branchwidth inequality of Robertson and Seymour~\cite[Section~5, (5.1)]{RobertsonSeymourGMX}, and retain a direct proof for our branch-decomposition convention.
\begin{lemma}\label{lem:bwtw}
For every graph with at least two edges, $\bw(G)\le\tw(G)+1$.
\end{lemma}
\begin{proof}
Take a tree decomposition of width $k$. For every graph edge, select one bag containing its endpoints and attach a new leaf with that edge's label to the selected bag. Replace a high-degree bag node by a tree of copies of its bag so that the resulting internal degree is at most three. Remove parts with no labelled leaves and suppress degree-two nodes. This yields a branch-decomposition tree.

Consider any cut before suppressing degree-two nodes. If it separates along an original decomposition-tree edge, a graph vertex incident with edge labels on both sides appears in bags on both sides. Connectedness of its bag occurrences forces it into the intersection of the two bags at that cut. If the cut is inside a replacement tree of copies of one bag, the same connectedness argument forces it into that copied bag. If it isolates a labelled edge, its middle vertices are a subset of that edge's endpoints, also contained in its selected bag. Every middle set therefore has size at most $k+1$. Suppressing degree-two nodes or pruning empty parts only retains cuts already represented by an earlier separation, and does not increase this bound. Minimize over the initial tree decomposition.
\end{proof}

\begin{corollary}\label{lem:prediagramwidth}
For every witness currently under consideration,
\[
 \tw(H_K^*)\le3,\qquad \tw(H_K)\le4,\qquad \bw(H_K)\le5.
\]
We use the fixed bound $b_0=6$ for separator descriptions.
\end{corollary}
\begin{proof}
Combine Lemma~\ref{lem:prediagramdual}, External result~\ref{ext:duality}, and Lemma~\ref{lem:bwtw}. The graph $H_K$ contains a cycle, so it satisfies the nontriviality hypothesis of the branchwidth statement. The larger constant six remains a valid upper bound.
\end{proof}

\subsection{Ordinary Nooses and the Retained Hierarchy}

A \emph{noose} in a plane graph is a simple closed curve that meets the graph only at vertices. Its complement on the sphere has two open disks. An orientation specifies which side is designated as its interior; no distinguished unbounded face is required. For a branch-decomposition cut, a representing noose has the middle set as its graph intersection and puts the relative interiors of the two edge sets on opposite sides.

\begin{external}\label{ext:spherecuts}
A connected plane graph of branchwidth $b$ admits a width-$b$ branch decomposition represented by nooses. Each face contains at most one segment of a representing noose, except when the cut isolates a single bridge edge; that exceptional noose may have two segments in its one incident face. The face-segment convention, including the exceptional singleton-bridge cut, is explained immediately before Theorem~17 of Bl\"asius et al.; that theorem supplies the sphere-cut decomposition~\cite[Section~7.1, Theorem~17]{Blasius}.
\end{external}

We call a noose \emph{ordinary} when every face contains at most one of its segments. An ordinary noose through $r$ graph vertices traverses $r$ distinct faces. Its $r$ visited vertices are distinct too. When $r=1$, the segment starts and ends at the same graph vertex through two different incidences of that vertex with the face. These are two boundary occurrences, not one incidence traversed twice.

Apply External result~\ref{ext:spherecuts} to $H_K$. A non-bridge edge exists because $H_K$ has a cycle. Choose a decomposition leaf labelled by a non-bridge edge and the decomposition-tree edge incident with that leaf as a common root cut. Cutting that tree edge gives two rooted sides. Represent the common cut by the same noose with opposite orientations on the two sides. Each other cut is oriented towards its descendant edge labels.

Delete only exceptional cuts that isolate a singleton bridge on their descendant side. They are leaf cuts: a non-leaf descendant subtree of a branch decomposition contains at least two labelled leaves. Keep their parents, even when one or both children have been deleted. The resulting \emph{retained hierarchy} on either root side therefore has zero, one or two children at each node. Indeed, after rooting at the common cut, each non-root internal decomposition node has at most two children. We remove only exceptional cuts isolating bridge-labelled leaves and retain their parents. A node that becomes a retained leaf represents at most its two deleted leaf labels; a retained original leaf represents its single edge label. All its nooses are ordinary. No extra width claim is needed for a unary node.

The root edge was chosen non-bridge, so an exceptional singleton bridge cannot occur on the complementary side of a non-leaf cut: that complementary side contains the chosen non-bridge root label. Thus every exception is removed by the stated rule. This observation prevents an unhandled exceptional noose from remaining at an internal node.

\subsection{Face Disks, Repeated Occurrences and Fixed Contacts}\label{subsec:radial}

The next construction makes simultaneous lifting a statement about closed disks and fixed trees. It does not assume that independently chosen shortest-path replacements preserve containment.

We recall the standard regular-neighborhood (ribbon-graph) construction~\cite{RibbonNeighborhood}. The next lemma also records the particular facial-occurrence and finite-tree avoidance properties needed for our simultaneous lifting; these additional clauses are checked below.
\begin{lemma}\label{lem:facedisks}
A connected finite graph embedded on the sphere has an arbitrarily thin closed regular neighborhood whose complementary components are disks, one for each face. Their boundaries distinguish every occurrence of a graph vertex or edge side on a facial walk. The neighborhood may be chosen disjoint from any prescribed finite family of compact embedded trees lying in the open faces.
\end{lemma}
\begin{proof}
Replace each graph vertex by a small disk and each edge by a thin rectangle meeting the endpoint disks in disjoint intervals, respecting the cyclic order of edges. First do this for a spanning tree. A single vertex disk is a disk; adding one tree leaf and its incident rectangle attaches another disk along one interval and again gives a disk. Induction gives a disk neighborhood of the entire spanning tree. Its complementary closure on the sphere is also a disk.

Now add the remaining edge rectangles in their given embedding. An unused edge interior lies in one current complementary disk. Its two ends meet distinct boundary intervals, even when their graph vertices coincide as occurrences. Attaching its narrow rectangle splits that complementary disk into two disks. Induction proves the disk assertion after all edges have been inserted. This is also the face count: each inserted non-tree edge increases the number of complementary disks by one.

A walk along a boundary of the neighborhood records each side of every rectangle and each intervening vertex-disk sector separately. In particular, the two sides of a bridge remain different occurrences on the same boundary, and repeated cutvertices remain distinct boundary positions. Finally, the graph and the prescribed face trees are disjoint compact subsets of a finite embedding. The vertex disks and rectangles can be chosen sufficiently small to avoid all those trees. Finiteness permits one choice that works for the entire family.
\end{proof}

Fix such a neighborhood of $H_K$ avoiding the restricted trees $T_K(c)$ and all original graph vertices. Write $D_c$ for the closed complementary disk belonging to $\phi_c$. All these disks have pairwise disjoint interiors. They contain their respective restricted trees in their interiors.

We now specify contacts using only objects whose possible values are later enumerable in the input. Place an auxiliary point $h_f$ inside every triangular face $f$ of $G^+$ and join it inside that face to its three corners. Let $J$ be the plane graph consisting of $G^+$ and these spokes. The points $h_f$ are \emph{not} voters or candidates. The graph $J$ is a topological subdivision device; no election distance is computed in it. It has $3n-4$ vertices and $O(n)$ edges and faces.

Draw the full dual of $G^+$ so that its vertex in $f$ is $h_f$, and its three outgoing half-edges meet the three edge sides of $f$. They alternate with the three corner spokes in the rotation around $h_f$. A surviving vertex of $H_K$ is one of these points $h_f$.

\begin{lemma}\label{lem:contacts}
For each occurrence of a face $\phi_c$ at a vertex $h_f$ of $H_K$, one can fix a corner $a$ of $f$ in that occurrence's sector such that $a\in\cell_K(c)$. The spoke $h_fa$, apart from $h_f$, lies in $\phi_c$. Different face-occurrence sectors at $h_f$ use different spokes. This choice can be made once and used for every cut in the hierarchy.
\end{lemma}
\begin{proof}
Before dual edges are removed, a corner sector at $h_f$ lies in the full-dual face of its primal corner $a$. Forest-edge deletions merge only vertex-faces belonging to the same deterministic cell, by Lemma~\ref{lem:directprediagram}. Later leaf pruning removes facial excursions and never merges different faces. It follows that every original corner sector included in a current occurrence of $\phi_c$ has its corner in $\cell_K(c)$.

Each current occurrence contains at least one of the original corner sectors. Choose one, for example the one whose corner has smallest identifier. The corner spoke is wholly inside that sector except at its endpoint $h_f$, so it lies in the required face. Distinct sectors have disjoint sets of original corners and therefore give different spokes. The choices depend on the embedded prediagram and its cell labels, not on a particular noose, so the same choice serves the entire hierarchy.
\end{proof}

Refine the thin regular neighborhood, if necessary, so that every chosen spoke crosses $\partial D_c$ exactly once. To justify this simultaneous choice, remove a sufficiently short initial segment at each endpoint $h_f$. The remaining finitely many compact spoke subarcs are disjoint from $H_K$; hence the vertex disks and edge strips can be chosen to avoid them. Within a sufficiently small vertex disk, each initial spoke segment leaves through its prescribed sector once. The refinement still avoids all original vertices and restricted cell trees.

Truncate each chosen spoke at this unique intersection with $\partial D_c$ when travelling from $h_f$ towards its corner. Its endpoint on $\partial D_c$ is a \emph{terminal}. Distinct occurrences yield distinct terminals. The portion of each spoke within $D_c$ attaches that terminal to its corner in $T_K(c)$. The union of $T_K(c)$ and all these truncated spokes is an embedded tree $\widehat T_c$ meeting $\partial D_c$ precisely in terminal leaves. Multiple terminals may attach to the same corner, but each attaches by its own edge and creates no cycle. All other points of $\widehat T_c$ lie in the interior of $D_c$.

\subsection{Canonical Radial Cycles and Simultaneous Disk Relations}

The \emph{radial graph} of $H_K$ has a vertex for every prediagram vertex and a vertex $w_c$ for every face $\phi_c$. It has an edge for every occurrence of a prediagram vertex on a face boundary. Repeated occurrences are different incidence edges, not a single simplified edge.

Use the fixed contact spokes outside the disks $D_c$. Inside each $D_c$, join every terminal to one interior point $w_c$ by a noncrossing fan in the cyclic order of the terminals. These paths construct a fixed embedded radial graph. Each open radial face contains the relative interior of exactly one edge of $H_K$; call it that edge's \emph{radial tile}. One sees this locally by following the two sides of an $H_K$ edge between the incidence paths at its endpoints. This remains valid for a bridge: its two face-side occurrences are kept distinct. The closure of a tile need not have a simple boundary, and we do not assume that it does.

\begin{lemma}\label{lem:radialdisks}
Replace each face segment of an ordinary noose by the two appropriate radial incidence paths. The result is a simple radial cycle representing the same edge partition. Its designated closed disk is the union of the closures of the radial tiles for the designated edge set.

For the retained rooted hierarchy, these canonical disks satisfy
\[
 D_X^{\rm rad}\subseteq D_P^{\rm rad},\qquad
 \Int(D_X^{\rm rad})\subseteq\Int(D_P^{\rm rad})
\]
for a child $X$ of $P$, and
\[
 \Int(D_L^{\rm rad})\cap D_R^{\rm rad}=\varnothing,
 \qquad
 \Int(D_R^{\rm rad})\cap D_L^{\rm rad}=\varnothing
\]
for siblings $L,R$.
\end{lemma}
\begin{proof}
An ordinary noose visits each graph vertex and each face at most once. Replacing a face segment by the fan path through its face vertex therefore introduces no repeated radial vertex. In the length-one case, the two distinct occurrence edges give a simple two-edge cycle in the radial multigraph. At graph vertices the chosen incidence sectors remain the same. Within a face, the replacement connects the same two boundary occurrences, so it separates the same facial edge-side intervals. Hence the radial cycle represents exactly the original graph-edge partition.

The interior of any radial tile is disjoint from the radial cycle. It is therefore wholly on one side, and its $H_K$ edge interior determines that side. The open radial faces partition the complement of the radial graph, and each contains the relative interior of exactly one edge of $H_K$; hence that edge label determines the side of the radial cycle containing the entire face. The finitely many tile interiors are dense in the sphere away from the radial graph. Closing their union on the designated side gives precisely the closed Jordan disk bounded by the cycle. This proves the asserted tile description.

Descendant edge-label sets are contained in their parent edge-label sets, even after a bridge leaf is suppressed. The tile description gives inclusion of their closed disks. An interior point of the child disk has an open neighborhood inside that disk, hence inside the parent disk; it cannot be a boundary point of the latter. This gives the interior inclusion.

Sibling edge-label sets are disjoint. If their disk interiors had a nonempty intersection, that intersection would be open and would meet an open radial tile, assigning its edge to both siblings, a contradiction. To obtain the stronger relation with the other \emph{closed} disk, suppose a point belonged to the interior of one disk and the boundary of the other. Every neighborhood of a Jordan-boundary point meets the latter disk's interior. A sufficiently small neighborhood lies in the first disk's interior, contradicting the disjointness just established. The symmetric relation follows by exchanging the siblings.
\end{proof}

This normalization supplies simultaneous relations from edge-set inclusion. No assumption of simultaneous disjointness of the original, individually supplied nooses is needed.

\subsection{The Embedded-Tree Cap Lemma}

Here an interval of a circle means one of the two closed boundary arcs between two distinct endpoints. Its relative interior omits its endpoints.

\begin{lemma}\label{lem:treecaps}
Let $D_0$ be a closed disk and $T_0$ a finite embedded tree whose intersection with $\partial D_0$ consists of distinct terminal leaves, with all other tree points in $\Int(D_0)$. For two distinct terminals $a,b$, the unique tree path $P_{ab}$ is a proper simple arc of the disk. For a boundary interval $J$ from $a$ to $b$, let $U(J)$ be the closed disk bounded by $J\cup P_{ab}$ and incident with the relative interior of $J$.

These caps have the following properties for all terminal intervals in the same tree:
\begin{enumerate}[label=(\roman*)]
\item If $J\subseteq J'$, then $U(J)\subseteq U(J')$;
\item if $J_1,J_2$ have disjoint relative interiors, then $\Int(U(J_1))$ and $\Int(U(J_2))$ are disjoint.
\item The conclusions allow common endpoints, common tree subpaths and equality of two intervals or caps.
\end{enumerate}
\end{lemma}
\begin{proof}
A path in a tree is simple. A terminal leaf cannot be an internal vertex of a tree path, so $P_{ab}$ meets $\partial D_0$ only at $a,b$. Thus it is a proper arc and, together with either boundary interval, bounds a disk. This defines $U(J)$ without any choice of a lane or perturbation.

For (i), let $P'$ be the tree path for the endpoints of $J'$. It separates $D_0$ into two sides. Every component of the tree away from $P'$ attaches to $P'$ at at most one vertex. Two attachments would give a cycle by combining a path through the component with the subpath of $P'$ between the attachments. A component off $P'$ lies entirely on one side of $P'$ because it is connected and does not meet $P'$.

The endpoints of $J$ belong to $J'$. A simple tree path between them cannot visit a component on the open opposite side of $P'$: it would have to enter and leave through that component's same attachment vertex, repeating that vertex. Hence the tree path defining $U(J)$ is contained in $U(J')$. The boundary interval $J$ is contained in $J'$ too. The open complement of $U(J')$ in $D_0$ is connected and incident with the boundary interval complementary to $J'$, which is outside the relative interior of $J$. It is consequently on the non-$J$ side of $P_{ab}$, and cannot lie in $U(J)$. This proves $U(J)\subseteq U(J')$, including the case of a shared endpoint. Equality follows directly when $J=J'$.

For (ii), let $J_2^c$ denote the closed interval complementary to $J_2$. Disjoint relative interiors imply $J_1\subseteq J_2^c$. By (i), $U(J_1)\subseteq U(J_2^c)$. The two caps $U(J_2)$ and $U(J_2^c)$ share their defining tree path and have disjoint interiors. Therefore $U(J_1)$ has no interior overlap with $U(J_2)$.

Nothing in these arguments requires distinct paths to be disjoint. Their intersections can be vertices or whole subpaths. The unique path and the selected boundary interval still define the cap, proving (iii). If the original attachments of two distinct terminal leaves coincide at an interior tree vertex, their terminal-to-terminal path consists of the two terminal stubs through that vertex; it remains a proper arc. This also covers a zero-length portion in the original, unextended tree.
\end{proof}

A useful consequence is that all replacements below can be made \emph{literally in the fixed trees}. Shared subpaths are allowed rather than resolved by an unproved simultaneous perturbation.

\subsection{Simultaneous Lifting to Original-Tree Paths}

For every retained radial cycle, replace its fan path inside $D_c$, when present, by the path between the same two terminals in $\widehat T_c$. Keep its portions outside all face disks unchanged. The resulting curve uses, in each traversed candidate face, two contact spokes and one path in the fixed original shortest-path tree $T_c$.

For a lifted, oriented curve $K_Z$, define
\[
 \Gamma(Z)=\{c:\phi_c\text{ is traversed by its original noose}\},
 \qquad \Lambda(Z)=K_Z\cap V(G).
\]
The vertices of $\Gamma(Z)$ are called its \emph{guards}. They are actual candidate identities, not separator vertices. Let $I(Z)$ be the original vertices strictly on the designated side, and put $R(Z)=I(Z)\cup\Lambda(Z)$. We reserve $S$ for an exclusion zone and use $Z$ for a separator.

\begin{lemma}\label{thm:lift}
For each retained rooted hierarchy, the replacements just defined give oriented Jordan curves with at most $b_0$ guards, and the following properties hold simultaneously.
\begin{enumerate} 
\item[{\sf (i)}] Every boundary vertex belongs, in the true election $K$, to a guard's cell:
\[
 \Lambda(Z)\subseteq\bigcup_{c\in\Gamma(Z)}\cell_K(c).
\]
\item[{\sf (ii)}] Every original graph path joining opposite strict sides contains a vertex of $\Lambda(Z)$.
\item[{\sf (iii)}] If $q\notin\Gamma(Z)$, its cell and its physical candidate vertex $q$ lie on the same strict side of the lifted curve as its face $\phi_q$ lay relative to the radial cycle.
\item[{\sf (iv)}] For a retained child $X$ of $P$,
\[
 I(X)\subseteq I(P),\quad \Lambda(X)\subseteq R(P),\quad R(X)\subseteq R(P).
\]
\item[{\sf (v)}] For retained siblings $L,R$,
\[
 I(L)\cap R(R)=\varnothing,\qquad I(R)\cap R(L)=\varnothing.
\]
\end{enumerate}
The opposite root orientations lift to the same geometric curve with opposite designated sides.
\end{lemma}
\begin{proof}
\emph{Individual simplicity.} Each ordinary radial cycle uses a face at most once, so inside its disk it uses only one proper tree arc. Paths belonging to different faces lie in different cells and face disks. Outside those disks the radial portions are unchanged. A visited prediagram vertex occurs once and its two incident contact spokes are distinct. Thus replacement introduces no self-intersection or repeated point beyond the common initial/final point of the closed curve. It is Jordan. A zero-length original-tree path still includes two distinct contact stubs and is covered by Lemma~\ref{lem:treecaps}; the length-one noose case uses distinct spokes at its one prediagram vertex. The guard count equals the number of traversed faces, at most $b_0$.

\emph{Simultaneous containment.} Intersect a canonical radial disk with a face disk $D_c$. If the cycle does not traverse $\phi_c$, this intersection is either the whole face disk or empty. If it does traverse the face, the intersection is the cap selected by one terminal boundary interval, with the fan path as its other boundary. The nested and disjoint radial-disk relations of Lemma~\ref{lem:radialdisks} give, in each $D_c$, nested intervals for a parent and child, and intervals with disjoint relative interiors for siblings, whenever both have a proper cap there. The whole-disk and empty cases already have the corresponding relations. A proper cap bounded by a fan path and a selected interval $J\subseteq\partial D_c$ meets $\partial D_c$ exactly in $J$. Consequently, containment of two radial caps implies containment of their selected intervals. If two radial-cap interiors are disjoint, their selected intervals have disjoint relative interiors: an open overlap of boundary intervals would have interior points of both caps immediately inside the disk. Shared endpoints cause no problem.

Replace fan caps by caps in the one fixed tree $\widehat T_c$. Lemma~\ref{lem:treecaps} preserves inclusions and disjoint open interiors in every face disk at once. Outside the face disks the regions and curves are unchanged. The replacements have identical boundary intervals on $\partial D_c$, so they glue to the designated regions of the new Jordan curves. The closed child region is contained in the closed parent region. Sibling open regions are disjoint. The same boundary-point argument as in Lemma~\ref{lem:radialdisks} upgrades disjoint sibling interiors to disjointness between one open region and the other's closed region. A child interior point cannot lie on the parent boundary. Restricting these topological relations to original vertices proves (iv) and (v). This argument allows equality of collapsed boundaries, shared tree edges and shared vertices.

\emph{Boundary ownership.} Each original boundary vertex is on a chosen contact or on a guard-tree path. Its contact corner is in that guard's cell by Lemma~\ref{lem:contacts}; the complete tree path is in the cell by Lemma~\ref{lem:celltrees}. This proves (i).

\emph{Graph separation.} The lifted curve is a simple cycle in the plane graph $J$. Since $G$ is a subgraph of $J$, an original path between opposite strict sides must meet that cycle. An original edge cannot cross a spoke or another edge in the fixed embedding. If the path meets a cycle edge by sharing that original edge, it also contains one of its endpoints on the cycle. An intersection therefore supplies an original boundary vertex, proving (ii).

\emph{Untraversed faces.} For a non-guard $q$, the radial cycle is disjoint from the interior of $D_q$, and that whole disk lies on one strict side. The lifting makes no change in $D_q$ and changes regions only inside other face disks. The restricted cell tree and all its original vertices remain in the same strict side, proving (iii).

Finally, the two root orientations use the same ports and the same unique tree paths. Reversing orientation only exchanges the selected boundary intervals with their complements. The lifted geometric curve is identical and its two sides are complementary.
\end{proof}

\begin{lemma}\label{lem:coverage}
Every candidate in $K$ is a guard of at least one retained cut. No selected candidate that is not a guard of a retained leaf lies in that leaf's closed region.
\end{lemma}
\begin{proof}
Each face $\phi_c$ is incident with a non-bridge edge of $H_K$. Otherwise its dual vertex in the connected graph $H_K^*$ would have only loops and no edge to another dual vertex, impossible because there is more than one face. The cut isolating a non-bridge edge is ordinary, is retained, and traverses both of its incident faces. It therefore exposes $c$ for at least one incident non-bridge edge. This proves coverage, including a cut that is the common root.

A retained leaf contains either a single non-bridge edge label, or only the one or two singleton bridge labels whose exceptional child cuts were deleted. If a face is not traversed by that leaf cut but lies wholly on its interior side, every edge-side occurrence on its boundary must have its edge interior on that side. In the bridge-only case this contradicts the existence of an incident non-bridge edge. In the single non-bridge case it would say that the boundary of a face consists of just that edge. A non-loop edge cannot alone form a facial closed walk unless both its sides are traversed, in which case it is a bridge; this contradicts its choice. Here $H_K$ is simple, so a one-edge loop face is not possible either.

Thus every untraversed candidate face lies strictly outside the leaf. Lemma~\ref{thm:lift}(iii) places its cell, and in particular its site, strictly outside the lifted leaf curve. This is stronger than merely excluding a hidden candidate from its strict interior: it excludes such a site from its boundary as well.
\end{proof}

\subsection{Enumerating Separators Without Knowing the Witness}\label{subsec:planarenum}

The algorithm knows $G^+$, $J$, $D=A\cup\{u\}$ and all trees $T_c$. It knows none of $F$, $K$, their cells, $H_K$ or its branch decomposition. The following enumeration uses only the known objects.

\begin{definition}\label{def:descriptor}
For $1\le r\le b_0$, a descriptor consists of $r$ distinct guards $c_1,\ldots,c_r\in D$, $r$ distinct triangular faces $f_1,\ldots,f_r$ of $G^+$, a corner $a_i$ of $f_{i-1}$ and a corner $b_i$ of $f_i$ for each $i$ (indices cyclically), and an orientation. Its expanded closed route concatenates
\[
 h_{f_{i-1}}\;a_i\; T_{c_i}[a_i,b_i]\; b_i\;h_{f_i},\qquad i=1,\ldots,r.
\]
Here $T_c[a,b]$ means the unique tree path including endpoints. A repeated displayed endpoint is written only once in the expanded route. When $a_i=b_i$, its tree portion has no edges.
\end{definition}

Reject a descriptor unless its expanded route is a simple cycle of $J$. In particular, no edge is retraced, no vertex repeats except the closing vertex, and the two spokes at a visited face center differ. For $r=1$ this last requirement makes the two contacts distinct. Also impose the following explicit guard-only check: each original vertex on the $c_i$ tree portion must have $c_i$ as its favorite among the descriptor's entire guard set. Call a descriptor passing both tests \emph{accepted}.

This preference check uses only the candidates explicitly listed by the descriptor; it does not require knowledge of the true candidate cells.

Lemma~\ref{lem:enumeration} shows that every separator required by
the completeness proof satisfies this additional condition.

\begin{lemma}\label{lem:enumeration}
All accepted oriented descriptors, their guard sets $\Gamma(Z)$, boundaries $\Lambda(Z)$, and vertex sides $I(Z)$ and $I(\bar Z)$ can be enumerated in polynomial time. Their number is $O(n^{2b_0})$. Every true lifted retained noose from Lemma~\ref{thm:lift} occurs in the enumeration, together with its inverse orientation $\bar Z$.
\end{lemma}
\begin{proof}
There are at most $n$ choices for each guard, $2n-4$ choices for each triangle, and $3^2$ choices for the two corners at each guard segment. Distinctness only reduces the count. Including two orientations gives the upper bound
\[
 2\sum_{r=1}^{b_0}(18n^2)^r=O(n^{2b_0}).
\]
Each of the at most $b_0$ tree paths has at most $n-1$ edges. Expand the route in $O(b_0n)$ time using the precomputed trees. A repeated-vertex and repeated-edge scan tests simplicity. Since the route follows edges of the known plane graph $J$, simplicity of the closed cycle suffices for its Jordan property; no numerical curve-intersection oracle is needed. The guard-only preference checks scan its $O(b_0n)$ original vertices against at most $b_0$ guards, using precomputed original distances.

To compute sides, construct the geometric dual of $J$ and delete its edges that cross the cycle. A simple cycle on the sphere separates the dual faces into its two sides. Each side is connected through face adjacencies not crossing the cycle: within one open disk, any two face-interior points can be joined while avoiding the finite vertices and crossing intervening edges transversally. Flood fill therefore gives exactly two dual components. An original vertex off the cycle has all its incident face sectors on the same side; assign it that side. Original vertices on the cycle form $\Lambda(Z)$. The specified orientation chooses $I(Z)$, and reversing orientation gives
\[
 I(\bar Z)=V(G)\setminus\bigl(I(Z)\cup\Lambda(Z)\bigr).
\]
These are operations on $O(n)$ embedded vertices and edges.

For a true lifted noose, choose its guard list, visited triangles and the fixed corners from Lemma~\ref{lem:contacts}. Its expanded route is exactly the Jordan curve of Lemma~\ref{thm:lift}. Every original vertex on a guard segment is truly owned by that guard under $K$. Restricting $K$ to the listed guards preserves that owner, so the guard-only check passes. Both orientations are enumerated. Thus none of the cuts required for completeness is lost.
\end{proof}

\begin{lemma}\label{m:lifting}
For an inclusion-minimal witness with at least four active candidates, both retained hierarchies can be lifted so that:
\begin{enumerate}
\item[{\sf (1)}] 
Every $Z$ is a simple closed separator with $|\Gamma(Z)|\le b_0$ and $\Lambda(Z)\subseteq\bigcup_{c\in\Gamma(Z)}\cell_K(c)$;
\item[{\sf (2)}] 
Every original path from $I(Z)$ to $V\setminus R(Z)$ meets $\Lambda(Z)$, and each participating candidate not in $\Gamma(Z)$ lies on the same side after lifting as its corresponding prediagram face did before lifting;
\item[{\sf (3)}] 
A child $X$ of $P$ satisfies $I(X)\subseteq I(P)$ and $\Lambda(X)\subseteq R(P)$; siblings $L,R$ satisfy $I(L)\cap R(R)=I(R)\cap R(L)=\varnothing$;
\item[{\sf (4)}] 
Every candidate in $K$ is a guard of at least one retained separator, and no participating candidate that is not a guard of a retained leaf lies inside that leaf's closed region.
\end{enumerate}
All resulting oriented separators 
belong to a polynomial-size family enumerable from $G,u,A$. 
\end{lemma}
\begin{proof}
Lemma~\ref{thm:lift} gives the boundary, separation, side-preservation and hierarchy properties. Lemma~\ref{lem:coverage} gives guard coverage and the retained-leaf property. Lemma~\ref{lem:enumeration} supplies the enumerable family containing both orientations of the required separators.
\end{proof}

For every accepted descriptor, regardless of whether it comes from a true witness, the three sets $I(Z),\Lambda(Z),I(\bar Z)$ partition $V(G)$, and every original path between the two strict sides meets $\Lambda(Z)$. These are geometric facts of an accepted simple cycle in $J$, not assumptions about a solution.

\begin{lemma}\label{lem:signature}
Descriptors with the same triple $(\Gamma(Z),\Lambda(Z),I(Z))$ are interchangeable in all the dynamic-programming rules below. They may be stored as one separator index, retaining one representative description and the complementary index.
\end{lemma}
\begin{proof}
Every rule below uses a separator only through its guard identities, its boundary and its designated vertex side. The region $R(Z)$ is their set union, and its opposite side is their complement. All geometric tests, candidate-site tests, voter counts and state domains are therefore identical for equal triples. Each retained representative is still an accepted cycle, so its individual graph-separation property remains valid. Completeness can replace a true descriptor by its equal-signature representative because the tests are on vertex sets; soundness uses precisely those tests. This proves interchangeability without identifying different guard sets or different sides.
\end{proof}

\subsection{The Complete Local Compatibility Tests}
Let $P$ be an enumerated parent separator and let $\mathcal X$ consist of one or two proposed child separators. Put
\[
 E_P:=\Gamma(P)\cup\bigcup_{X\in\mathcal X}\Gamma(X).
\]
This set contains at most $3b_0$ candidate identities. The transition is geometrically and electorally permitted only if all the following tests pass:
\begin{align}
 I(X)&\subseteq I(P), &\Lambda(X)&\subseteq R(P) &&(X\in\mathcal X),\label{eq:geomparent}
 \\
 I(L)\cap R(R)&=\varnothing, &I(R)\cap R(L)&=\varnothing &&(\mathcal X=\{L,R\}),\label{eq:geomsib}
\\
 q\in\Gamma(X)\setminus\Gamma(P)&\Longrightarrow q\in I(P),&&&&(X\in\mathcal X),\label{eq:disappear}
 \\
 q\in E_P\setminus\Gamma(X)&\Longrightarrow q\notin R(X),&&&&(X\in\mathcal X),\label{eq:siteoutside}
 \\
 \best_z(E_P)&\in\Gamma(P)&&&&(z\in\Lambda(P)),\label{eq:guardparent}
 \\
 \best_z(E_P)&\in\Gamma(X)&&&&(z\in\Lambda(X),\ X\in\mathcal X).\label{eq:guardchild}
\end{align}
The symbols $q,z$ in these tests are original graph vertices. Thus a candidate-site condition is a check of its actual location, not of a label assigned to an abstract face. All preferences use $\kappa_z$ in the original graph. When a non-guard is beaten in one of the boundary tests, the preference is strict in the complete distance-and-identifier order, because identifiers are distinct.

\begin{lemma}\label{lem:truecompatible}
The lifted retained hierarchy of a true minimal witness satisfies all conditions \eqref{eq:geomparent}--\eqref{eq:guardchild}.
\end{lemma}

\begin{proof}
The geometric conditions are Lemma~\ref{thm:lift}(iv)--(v). To prove the site conditions without presupposing a guard's position on its lifted path, first use the canonical radial disks. A guard $q$ of a child but not of its parent has its radial face vertex $w_q$ on the child boundary and in the closed parent disk. It is not on the parent boundary, because the parent does not traverse $\phi_q$. Hence it is strictly inside the parent. Its entire untraversed face is on that side, so Lemma~\ref{thm:lift}(iii) puts its actual site strictly inside the lifted parent. This proves \eqref{eq:disappear}.

Now let $q$ be a parent or sibling guard but not a guard of child $X$. Its face vertex $w_q$ is on the parent or sibling radial boundary. It cannot lie in the strict child disk: for a parent boundary this follows from child--parent inclusion, and for a sibling boundary from the strong sibling relation. It is not on the child boundary either, since $q$ is not a child guard. It is therefore strictly outside the child. Non-guard side preservation again puts the actual site outside the child's closed region, proving \eqref{eq:siteoutside}.

Finally, at each lifted boundary vertex the true owner is one of that separator's guards by Lemma~\ref{thm:lift}(i). This owner belongs to $E_P\subseteq K$. Restricting the full active set to $E_P$ retains the unique best candidate at the vertex. This proves both boundary tests.
\end{proof}

\subsection{Finite Derivations and Candidate Locations}

The algorithm combines enumerated separators by the checked rules, not necessarily by a known decomposition. We first prove the geometric invariants for any such finite combination, without using vote counts.

A \emph{finite derivation} is a finite rooted tree of separator occurrences. A leaf has no children, and every internal occurrence has one or two children satisfying \eqref{eq:geomparent}-\eqref{eq:guardchild}. Repeated occurrences of the same separator index are allowed. For an occurrence $Z$, let $\Rep(Z)$ be the set of all guard identities occurring at $Z$ or below it, and call $\Rep(Z)\setminus\Gamma(Z)$ its hidden represented candidates. The set $\Rep(Z)$ is used only in the correctness proof. The algorithm does not store it, since it may contain many candidate identities.

A \emph{completed pair} consists of finite derivations rooted at $Z$ and $\bar Z$, whose guard sets and boundary coincide and whose designated vertex interiors are complementary. Define its complete candidate set by
\[
 K:=\Rep(Z)\cup\Rep(\bar Z).
\]
No multiplicities are used in this union. At this stage a completed pair is purely geometric; later rules require that $u\in K$ and that at least one opponent is present.

\begin{lemma}\label{lem:locations}
For every locally compatible finite derivation, and every occurrence $Z$,
\[
 q\in\Rep(Z)\setminus\Gamma(Z)\ \Longrightarrow\ q\in I(Z).
\]
For every occurrence in a completed pair, also
\[
 q\in K\setminus\Rep(Z)\ \Longrightarrow\ q\notin R(Z).
\]
\end{lemma}
\begin{proof}
The first assertion is proved bottom-up. At a leaf, $\Rep(Z)=\Gamma(Z)$, so it is vacuous. At a parent $P$, a represented identity not in $\Gamma(P)$ is either already hidden in a child, or is a child guard disappearing at $P$. In the first case induction puts it in $I(X)\subseteq I(P)$. In the second, \eqref{eq:disappear} puts it in $I(P)$. These cases exhaust the union of the child represented sets and the parent guard set.

For the second assertion, start at the roots. An identity in $K$ absent from one root derivation occurs on the other side and cannot be a common root guard. The first assertion places it in that other side's strict interior, which is outside the first root's closed region.

Proceed top-down from a parent $P$ to child $X$. An identity absent from $\Rep(X)$ is either absent from $\Rep(P)$, a parent or sibling guard, or hidden in a sibling. In the first case it lies outside $R(P)$ by the parent induction hypothesis and hence outside $R(X)\subseteq R(P)$. In the second case it is not a child guard and \eqref{eq:siteoutside} applies. In the third case its site lies in the sibling's strict interior by the first assertion, which is disjoint from $R(X)$ by \eqref{eq:geomsib}. A unary parent simply has no sibling case. This completes the induction.
\end{proof}

\begin{lemma}\label{lem:internalguard}
In every locally compatible finite derivation, if $q\in\Rep(Z)\setminus\Gamma(Z)$ and $z\in\Lambda(Z)$, then some $b\in\Gamma(Z)$ satisfies $\kappa_z(b)<_{\lex}\kappa_z(q)$.
\end{lemma}
\begin{proof}
Induct bottom-up. A leaf has no hidden represented candidates. At parent $P$, a newly disappearing child guard $q$ belongs to $E_P\setminus\Gamma(P)$. The boundary condition \eqref{eq:guardparent} gives a parent guard that strictly beats $q$ at every parent-boundary vertex.

Otherwise $q$ was already hidden in some child $X$. Its site lies in $I(X)$ by Lemma~\ref{lem:locations}. A vertex $z\in\Lambda(P)$ does not lie in $I(X)$, since $I(X)\subseteq I(P)$ and $\Lambda(P)$ is disjoint from $I(P)$. If $z\in\Lambda(X)$, the child induction hypothesis supplies a child guard beating $q$ there. If $z$ is strictly outside the child's closed region, a shortest $q$--$z$ path meets $\Lambda(X)$ at some vertex $y$. The child induction hypothesis gives a guard beating $q$ at $y$. Reverse geodesic persistence, Lemma~\ref{lem:geodesic}, propagates that strict preference to $z$. In either case a member of $\Gamma(X)\subseteq E_P$ beats $q$ at $z$. The best member of $E_P$ is no worse than that member, and \eqref{eq:guardparent} says it is a parent guard. It therefore strictly beats $q$, as required.
\end{proof}

\begin{lemma}\label{thm:globalguard}
In a completed pair with candidate set $K$, every separator occurrence satisfies
\[
 \best_z(K)\in\Gamma(Z)\qquad(z\in\Lambda(Z)).
\]
\end{lemma}
\begin{proof}
At the common root boundary, any candidate outside the common guard set is hidden in at least one root derivation. Lemma~\ref{lem:internalguard} supplies a root guard strictly beating it at that boundary vertex. Thus no nonguard can be the unique favorite there, proving the root assertion.

Assume the assertion at a parent $P$ and consider a child $X$. Fix $z\in\Lambda(X)$ and $q\in K\setminus\Gamma(X)$. The following cases exhaust all possibilities.

First, if $q$ is hidden in $X$, use Lemma~\ref{lem:internalguard} for $X$. Second, if $q\in E_P\setminus\Gamma(X)$, condition \eqref{eq:guardchild} supplies a child guard strictly beating it.

Third, suppose $q$ is hidden in a sibling $Y$. Its site lies in $I(Y)$, while $z\notin I(Y)$ by the strong sibling condition. A shortest path from $q$ to $z$ meets $\Lambda(Y)$, possibly at $z$ itself. Internal guarding there and reverse geodesic persistence give a sibling guard that beats $q$ at $z$. Since that guard belongs to $E_P$, condition \eqref{eq:guardchild} replaces it by a no-worse child guard.

Finally, if $q$ is not represented below $P$, Lemma~\ref{lem:locations} puts it outside $R(P)$. The point $z$ belongs to $R(P)$ by \eqref{eq:geomparent}. A shortest $q$--$z$ path therefore meets $\Lambda(P)$, again possibly at $z$. Global guarding of the parent supplies a parent guard strictly beating $q$ at that crossing, because $q$ is not a parent guard. Reverse geodesic persistence carries the preference to $z$, and \eqref{eq:guardchild} replaces the parent guard by a no-worse child guard.

Every nonguard is thus strictly beaten at $z$ by some child guard. The unique favorite in $K$ must be a guard, proving the induction step.
\end{proof}

\begin{lemma}\label{lem:shield}
In a completed pair, let $q\in K\setminus\Gamma(Z)$.
\begin{enumerate}[label=(\roman*)]
\item If $q\in I(Z)$, no voter outside $I(Z)$ supports $q$.
\item 
If $q\notin R(Z)$, no voter in $R(Z)$ supports $q$.
\end{enumerate}
\end{lemma}
\begin{proof}
Boundary voters support guards by Lemma~\ref{thm:globalguard}, so they never support $q$. For a voter on the opposite strict side from $q$, take a shortest path from the candidate site $q$ to the voter. Graph separation gives a boundary vertex $z$ on that path. The favorite at $z$ is a guard and is different from $q$, so it strictly beats $q$. Reverse geodesic persistence then shows that the opposite-side voter also prefers that guard to $q$. This proves both directions. The premise that $q$ is not a guard is essential; without it the boundary favorite could be $q$ itself.
\end{proof}

\begin{lemma}\label{lem:hiddenidentity}
A candidate cannot be hidden in both sibling derivations, or in both opposite root derivations. Once hidden in a derivation, it cannot later reappear as a parent guard or an unexposed sibling guard in a compatible extension.
\end{lemma}
\begin{proof}
Two hidden occurrences in siblings would place the same physical site in both disjoint strict interiors, by Lemma~\ref{lem:locations}. The two root interiors are disjoint for the same reason. For reappearance, consider the first ancestor step at which the hidden identity is a parent guard or a sibling guard while absent from the corresponding child guard set. It remains in the child's strict interior by Lemma~\ref{lem:locations}, but \eqref{eq:siteoutside} requires its site outside that child's closed region. This is a contradiction.

A candidate still exposed in two children is a different case: its identity is explicitly stored in both guard sets and its votes are added by identity. It is finalized only once, when it ceases to be a parent guard. The lemma excludes precisely the cases in which forgotten identities could have needed later deduplication.
\end{proof}

\begin{lemma}\label{m:shield}
Take two finite derivations satisfying \eqref{eq:geomparent}--\eqref{eq:guardchild}, with the same root guards and boundary and opposite choices of interior. Let $K$ be the set of all candidate identities appearing as guards in either derivation. For every separator occurrence $Z$ in either derivation, every boundary voter chooses a guard: $\best_z(K)\in\Gamma(Z)$ for $z\in\Lambda(Z)$. For $q\in K\setminus\Gamma(Z)$, if $q\in I(Z)$, then no voter outside $I(Z)$ votes for $q$; if $q\notin R(Z)$, then no voter in $R(Z)$ votes for $q$. Finally, a candidate that has appeared as a guard below $Z$ but is not in $\Gamma(Z)$ cannot reappear as a guard at an ancestor of $Z$ or in a sibling subtree outside the subtree rooted at $Z$.
\end{lemma}
\begin{proof}
Lemma~\ref{thm:globalguard} gives global boundary guarding for the candidate set reconstructed from a completed pair. Lemma~\ref{lem:shield} gives both score-shielding assertions. Lemmas~\ref{lem:locations} and~\ref{lem:hiddenidentity} exclude reappearance of hidden identities in compatible ancestors or sibling subtrees.
\end{proof}

These results depend only on the finite derivation and the local tests. They do not use an asserted correctness of the numerical states. This order of proof is important: global guarding is available before any claim that stored scores are exact.

\subsection{Numerical States and Their Intended Meaning}
\label{subsec:planarstates}

For every oriented accepted separator $Z$, the dynamic program stores tuples
\begin{equation}\label{eq:planarstate}
 (\sigma,\mu,\theta,\rho).
\end{equation}
Their fields are as follows.

The vector $\sigma$ is indexed by the at most $b_0$ guards. Its entry $\sigma(b)\in\{0,\ldots,n\}$ counts voters in the strict interior $I(Z)$ assigned to guard $b$.

The field $\mu$ is either $\infty$ or an elimination key $(t,-\id(q))$ with $0\le t\le n$. It records the lexicographically smallest \emph{complete} score key of hidden opponents represented below $Z$. It does not include the target. These scores are complete, not partial: a hidden candidate lies inside and cannot receive any vote from outside the strict interior by Lemma~\ref{lem:shield}.

The target field has one of three forms:
\[
 \theta=\Out,\qquad \theta=\Exposed,\qquad \theta=\Hidden(t_u),\quad 0\le t_u\le n.
\]
An exposed target is exactly $u\in\Gamma(Z)$. A hidden target has appeared as a guard below, lies in $I(Z)\setminus\Gamma(Z)$ and has already acquired its complete score $t_u$. An outside target has not been represented on this side and must satisfy $u\notin R(Z)$.

Finally, $\rho\in\{0,1\}$ records whether any nontarget identity has appeared as a guard anywhere in the derivation. Every nontarget guard must be allowed: $\Gamma(Z)\setminus\{u\}\subseteq A$. This is checked at every inserted separator occurrence.

The numerical interpretation will be proved for every state used in a legal completed pair. An isolated table entry need not have a realizable external completion, and the algorithm does not assume one. This avoids an unjustified existential boundary context. The tables are defined operationally by the following finite rules, and soundness constructs their actual completed candidate set afterwards.

\subsection{Leaf Initialization}

Every accepted separator is eligible to be a derivation leaf. Set
\begin{equation}\label{eq:planarleaf}
 \sigma(b)=|\{v\in I(Z):\best_v(\Gamma(Z))=b\}|,\qquad
 \mu=\infty,\qquad
 \rho=\ind{\Gamma(Z)\setminus\{u\}\ne\varnothing}.
\end{equation}
If $u\in\Gamma(Z)$, set $\theta=\Exposed$. Otherwise permit $\theta=\Out$ only if $u\notin R(Z)$. There is no hidden-target leaf state, because no target has yet appeared below it. Reject any leaf whose guard eligibility check fails.

For a true retained leaf, Lemma~\ref{lem:coverage} ensures that every selected nonguard candidate lies strictly outside. Therefore all selected candidates capable of receiving a strict-interior vote are guards. The initialization gives its exact true score vector. For an arbitrary accepted leaf, exactness will follow only after legal completion, using Lemma~\ref{lem:shield}.

\subsection{Residual Voters and the Score Transition}

Let $P$ have one or two children $\mathcal X$ with stored states, and suppose all local compatibility tests pass. Write
\begin{equation}\label{eq:residual}
 Y_P:=I(P)\setminus\bigcup_{X\in\mathcal X} I(X).
\end{equation}
The child strict interiors are disjoint. The residual set includes child-boundary vertices that lie strictly inside the parent. A vertex appearing on both child boundaries is still one element of $Y_P$.

For each $q\in E_P$, compute directly
\begin{equation}\label{eq:residualcounts}
 r_P(q):=|\{v\in Y_P:\best_v(E_P)=q\}|.
\end{equation}

\begin{lemma}\label{lem:residualowner}
In any completed pair, every residual voter in $Y_P$ has its true owner in $E_P$. Consequently the counts in \eqref{eq:residualcounts} are the true contributions of the residual voters.
\end{lemma}
\begin{proof}
A candidate in $K\setminus E_P$ is either hidden in a child or absent from the parent derivation. In the first case its site lies in that child's strict interior. A residual voter is not in that strict interior, so Lemma~\ref{lem:shield}(i) excludes that candidate. This also treats a voter on a shared boundary. In the second case the candidate lies outside $R(P)$ by Lemma~\ref{lem:locations}, and Lemma~\ref{lem:shield}(ii) excludes it from all of $I(P)$. No other case remains: the parent represented set is its guard set together with the child represented sets. Thus a true favorite is retained in $E_P$, and restricting the active set to $E_P$ cannot change it.
\end{proof}

For a parent guard $b$, set
\begin{equation}\label{eq:parentcounts}
 \sigma_P(b)=r_P(b)+\sum_{\substack{X\in\mathcal X\\b\in\Gamma(X)}}\sigma_X(b).
\end{equation}
For a disappearing child guard, meaning
\[
 q\in\left(\bigcup_{X\in\mathcal X}\Gamma(X)\right)\setminus\Gamma(P),
\]
compute
\begin{equation}\label{eq:finalcounts}
 t(q)=r_P(q)+\sum_{\substack{X\in\mathcal X\\q\in\Gamma(X)}}\sigma_X(q).
\end{equation}
Each identity is considered once in the union, even if it is exposed by both children. Update
\begin{equation}\label{eq:planarmu}
 \mu_P=\minlex\left(\{\mu_X:X\in\mathcal X\}\cup
 \{(t(q),-\id(q)):q\ne u\text{ disappears at }P\}\right),
\end{equation}
with the usual convention that all-infinite entries have minimum $\infty$.

\begin{lemma}\label{lem:finalization}
In a completed pair, if the child score vectors and hidden minima are exact, then \eqref{eq:parentcounts} gives exact parent strict-interior guard scores, \eqref{eq:finalcounts} gives the complete global score of every disappearing candidate, and \eqref{eq:planarmu} gives the exact parent hidden-opponent minimum.
\end{lemma}
\begin{proof}
Partition the parent's strict-interior voters into $Y_P$ and its child strict interiors. Residual contributions are exact by Lemma~\ref{lem:residualowner}. When a parent or another child guard $q$ is not exposed by child $X$, condition \eqref{eq:siteoutside} puts its site outside $R(X)$; it is a nonguard of $X$ and receives no vote in $I(X)$ by Lemma~\ref{lem:shield}. Therefore a term omitted from either sum is genuinely zero. Terms that are present are exactly the child vector entries, so the sums account for all strict-parent-interior contributions.

A disappearing candidate has site in $I(P)$ by \eqref{eq:disappear} and is not a parent guard. By Lemma~\ref{lem:shield} it receives no vote outside $I(P)$, including no vote on $\Lambda(P)$. Its strict-interior total is therefore its full global score. Old hidden child candidates already have complete scores and cannot receive new residual or sibling votes. Parent-hidden candidates consist precisely of these old hidden candidates and the disappearing child guards. Lemma~\ref{lem:hiddenidentity} excludes forgotten-identity duplication and reappearance. Taking the minimum once per new identity, together with the already exact child minima, is thus exact.
\end{proof}

\subsection{The Four Target Cases and Opponent Presence}

The target must be represented exactly as the same mandatory vertex, not treated as an anonymous hidden opponent. Only the following cases are permitted.

\paragraph*{Target exposed at the parent.}
If $u\in\Gamma(P)$, require that no child has a hidden target, and set $\theta_P=\Exposed$. Children may be exposed or outside, subject to their own status conditions and the site tests. A hidden child would contradict the no-reappearance lemma.

\paragraph*{Target disappears at this transition.}
If $u\notin\Gamma(P)$, no child is hidden, and at least one child is exposed, compute $t(u)$ by \eqref{eq:finalcounts} and set $\theta_P=\Hidden(t(u))$. If both children are exposed, their vote contributions are added but the target identity is finalized only once. The site test places it in $I(P)$.

\paragraph*{Target was already finalized in one child.}
If exactly one child has $\Hidden(t_u)$, all other children are outside, and $u$ is neither a parent guard nor exposed by another child, set $\theta_P=\Hidden(t_u)$. Do not add any vote to $t_u$: by Lemma~\ref{lem:shield} it was already the complete score. Strict-interior containment propagates its valid location.

\paragraph*{Target has not yet appeared on this side.}
If all children are outside and $u\notin\Gamma(P)$, permit $\theta_P=\Out$ only if $u\notin R(P)$.

Reject every other combination, including two hidden children, a hidden child together with an exposed sibling, and an unrepresented target located in the parent closed region. Also reject any computed count outside $\{0,\ldots,n\}$; genuine election scores always lie in this range.

In all four cases update
\begin{equation}\label{eq:planarrho}
 \rho_P=\left(\bigvee_{X\in\mathcal X}\rho_X\right)
       \vee\ind{\Gamma(P)\setminus\{u\}\ne\varnothing}.
\end{equation}
These cases are exhaustive for any actual candidate set: the target is either a parent guard, newly hidden, previously hidden in one child, or unrepresented. It cannot be hidden in two children or hidden and exposed elsewhere by Lemma~\ref{lem:hiddenidentity}.

\subsection{Finite Construction of the State Tables}

Initialize all legal leaf states. Repeatedly apply every permitted one-child and two-child transition to states already present in the tables, inserting any new state that is produced. Stop when a full pass produces no new state.

A separator need not have a strictly smaller interior than its parent, so termination cannot be justified by decreasing region size. Instead, it follows from the fact that only finitely many
separator-state pairs can be inserted.

\begin{lemma}\label{lem:finitecertificate}
Every stored state can be obtained by a finite sequence of transition applications beginning with leaf states. Predecessor pointers can always be chosen to refer to states inserted earlier.
\end{lemma}

\begin{proof}
A leaf state already has the required finite construction. Now consider a state when it is inserted for the first time. Every child state used by that transition was present earlier. By induction, each of those child states has a finite construction from leaf states. Adding the current transition gives one for the new state.

Store pointers to those earlier child states. Their insertion times strictly decrease when the pointers are followed backwards, so the pointer structure cannot contain a cycle. If the same numerical
state is generated again later, keep its original predecessor pointers.
\end{proof}

\subsection{Complementary Root Completion and Voter Accounting}

For every separator $Z$, pair its stored states with those of its inverse $\bar Z$. Their common guard set is $\Gamma$, their common boundary is $\Lambda$, and
\begin{equation}\label{eq:rootpartition}
 V(G)=I(Z)\ \dot\cup\ \Lambda\ \dot\cup\ I(\bar Z).
\end{equation}
For each root guard $b$, compute
\begin{align}
 \beta(b)&=|\{v\in\Lambda:\best_v(\Gamma)=b\}|,\label{eq:boundarycounts}\\
 t(b)&=\sigma_Z(b)+\sigma_{\bar Z}(b)+\beta(b).\label{eq:rootguardscores}
\end{align}
If $u\in\Gamma$, require both states exposed and use \eqref{eq:rootguardscores} for $t_u$. If $u\notin\Gamma$, require exactly one hidden target and one outside target, and use the hidden score as $t_u$. No other root-target pattern is legal.

Set
\begin{equation}\label{eq:rootoppmin}
 \mu_{\rm opp}=\minlex\left(\{\mu_Z,\mu_{\bar Z}\}\cup
       \{(t(q),-\id(q)):q\in\Gamma\setminus\{u\}\}\right).
\end{equation}
Require $\rho_Z\vee\rho_{\bar Z}=1$, and accept this pair precisely when
\begin{equation}\label{eq:planaraccept}
 (t_u,-\id(u))<_{\lex}\mu_{\rm opp}.
\end{equation}
The presence requirement prevents an election containing only the target from being accepted on comparison with $\infty$.

\begin{lemma}\label{lem:voteraccounting}
In any completed pair, the leaf, residual and root-boundary voter sets partition the original voters, with no multiplicities even when lifted curves share vertices or edges. Equations \eqref{eq:boundarycounts}--\eqref{eq:rootguardscores} are exact once the two strict-interior score vectors are exact.
\end{lemma}
\begin{proof}
At each parent, its strict interior is the disjoint union of its child strict interiors and the set difference $Y_P$. This is a set partition, not a sum of boundary incidences. Recursively substituting these partitions gives a partition into leaf strict interiors and residual voter sets. A vertex on a child boundary is omitted from that child score vector. If it becomes strictly interior at the parent, it is in $Y_P$; if it remains on the parent boundary, it is still deferred. A vertex on both sibling boundaries is counted only once in $Y_P$, and it cannot be in a sibling strict interior by \eqref{eq:geomsib}.

At the roots the two strict interiors are disjoint and leave exactly their common boundary, as in \eqref{eq:rootpartition}. Global guarding assigns every remaining boundary voter to a root guard. Its favorite among those guards is its actual favorite in $K$. Thus $\beta$ gives the exact remaining contribution, and \eqref{eq:rootguardscores} adds the three disjoint voter parts once each.
\end{proof}

\subsection{Algorithmic Overview}

The width bound concerns the quotient, not $G$ itself. For the completeness proof, Section~\ref{sec:planar} constructs an auxiliary embedded graph $H_K$, called the \emph{prediagram}, with one face for each candidate cell. The algorithm does not construct $H_K$: this graph depends on the chosen witness and its cells, which are not known to the algorithm.
A \emph{branch decomposition} is a tree whose leaves correspond to graph edges and whose internal nodes have degree three. Cutting a decomposition-tree edge partitions the graph edges into two groups. The width is the largest number of graph vertices incident with edges in both groups, over all such cuts; the minimum possible width is the graph's \emph{branchwidth}. Corollary~\ref{lem:prediagramwidth} bounds the branchwidth of $H_K$ by five. We retain the conservative uniform parameter $b_0=6$ throughout the descriptor enumeration and complexity analysis. Using five throughout would require rechecking the enumeration, simultaneous lifting, transition and recovery arguments with that tighter capacity. Our purpose here is to establish polynomial-time tractability rather than optimize the exponent, so all stated bounds use six.

The cuts are represented by \emph{nooses}, simple closed curves meeting $H_K$ only at vertices and separating the two edge groups. A \emph{bridge} is an edge whose deletion disconnects the graph. Splitting the decomposition at a cut isolating a non-bridge edge gives two rooted families of cuts, one on each side. Removing the exceptional leaf cuts that isolate a single bridge, while retaining their parents, gives the two \emph{retained hierarchies}.
We \emph{lift} these curves by replacing their portions inside the corresponding faces of $H_K$ with paths in the corresponding candidates' fixed shortest-path trees in $G$ and with fixed connections specified in this section. The candidates whose trees supply these paths are the \emph{guards}; they are candidate identities, not necessarily vertices on the resulting curve. The connections are fixed for the whole hierarchy, and the replacements are made simultaneously to preserve the relations between parents, children and siblings.
The algorithm enumerates descriptions of possible lifted curves directly from the input. Each description lists at most $b_0$ guards from $D=A\cup\{u\}$ and the information needed to reconstruct the curve. Retain it only if the curve is simple and closed, and every original vertex on a tree path supplied by guard $c$ prefers $c$ among all listed guards. The boundary may contain many vertices. An orientation chooses which side is the interior. For a separator $Z$, let $\Gamma(Z)$ be its guards, $\Lambda(Z)$ its original boundary vertices, $I(Z)$ the original vertices strictly inside, and $R(Z)=I(Z)\cup\Lambda(Z)$.

The proof lifts the hierarchy \emph{simultaneously}, preserving the parent--child and sibling relations in part (3) of Lemma~\ref{m:lifting}; independently replacing each noose would not suffice. The connection to the algorithm is that every lifted separator obtained from the chosen witness has a constant-size descriptor using the fixed shortest-path trees and at most $b_0$ guards. The algorithm enumerates such descriptors from $G,u,A$, reconstructs their boundaries and the vertex sets on their two sides, and first checks each descriptor individually. The six tests in \eqref{eq:geomparent}--\eqref{eq:guardchild} check proposed parent--child combinations. Lemma~\ref{lem:enumeration} bounds the family by $O(n^{2b_0})$ and guarantees that it contains both orientations of every separator in the retained hierarchies obtained from that witness. Thus the DP need not guess $F$ or its cells; repeated application of the transition rules can reconstruct the required hierarchy. The enumeration can include separators that need not arise from any first-round witness, which is why the local tests are essential.
For parent $P$ and one or two children $\mathcal X$, put $E_P=\Gamma(P)\cup\bigcup_{X\in\mathcal X}\Gamma(X)$. All six conditions are required:
The compatibility conditions are \eqref{eq:geomparent}--\eqref{eq:guardchild}.
The first two compatibility conditions place children inside their parent and keep sibling interiors disjoint; the next two constrain candidate locations, and the final two ensure that each boundary voter's preferred candidate in $E_P$ is a guard of that boundary. A parent and its one or two children are \emph{locally compatible} when all six conditions hold. A \emph{finite derivation} is a finite rooted tree whose nodes are labelled by enumerated separators; labels may repeat. Each parent and its children must satisfy all six compatibility conditions. A \emph{completed pair} consists of two such derivations with the same root guards and boundary, but opposite interiors. The full conditions appear in \eqref{eq:geomparent}--\eqref{eq:guardchild}.
For a separator occurrence $Z$, ``below $Z$'' means at a strict descendant in its derivation.

A \emph{parent guard} is in $\Gamma(P)$. A child guard \emph{disappears at the parent} when it belongs to a child guard set but not to $\Gamma(P)$. It remains a participating candidate; only its identity is no longer explicitly stored at the parent.
Lemma~\ref{m:shield} applies globally to the candidate set reconstructed from a completed pair, not merely to a locally compatible parent--child combination. It justifies finalizing disappearing candidates' scores and prevents forgotten identities from reappearing. The preceding lemmas establish these geometric properties before proving score correctness (Lemmas~\ref{lem:locations}--\ref{lem:hiddenidentity}).
For each enumerated separator $Z$, the algorithm stores a set of records of the form $(\sigma,\mu,\theta,\rho)$. The entry $\sigma(q)$ counts votes from $I(Z)$ for a guard $q$. The value $\mu$ is the minimum complete elimination key among opponents that appeared as guards below $Z$ but are no longer in $\Gamma(Z)$, or $\infty$ if none exists. The target status $\theta$ is $\Exposed$ when $u\in\Gamma(Z)$, $\Hidden(s_u)$ when it appeared below but is no longer a guard and has complete score $s_u$, and $\Out$ when it has not appeared on this side and lies outside $R(Z)$. The flag $\rho$ records whether an opponent has appeared as a guard at or below $Z$. The target is never included in $\mu$. These score meanings are proved for records used together in a legal completed pair; an isolated record need not admit such a completion.

For each separator, form an initial leaf record by counting votes from its strict interior using only its guards, and set $\mu=\infty$. The target must be a guard or lie outside the closed region; a hidden-target leaf is not permitted. 
At a compatible parent $P$ with children $\mathcal X$, let $Y=I(P)\setminus\bigcup_{X\in\mathcal X}I(X)$ and $r(q)=|\{v\in Y:\best_v(E_P)=q\}|$. For each $q\in E_P$, update
$t(q) := r(q)+\sum_{X\in\mathcal X:\ q\in\Gamma(X)}\sigma_X(q)$.
For each parent guard $q$, set $\sigma_P(q)=t(q)$. The new minimum $\mu_P$ is the minimum of the child minima and the elimination keys $(t(q),-\id(q))$ of all disappearing opponents $q\ne u$. The target remains exposed, becomes newly hidden with this complete score, remains hidden in exactly one child, or remains outside the parent closed region; inconsistent combinations are rejected. At a merge, only voters in the residual set $Y$ are newly counted, even if a voter lies on both child boundaries. Common root-boundary voters are counted when the two root records are combined.
Starting from leaf states, repeatedly apply the unary and binary rules until a full pass inserts no new separator--state pair. The invariant is a \emph{finite derivation} from earlier states (Lemma~\ref{lem:finitecertificate}), not decreasing region size. The two root interiors and their common boundary partition the voter set. If $u$ is a common root guard, both records must mark it $\Exposed$. Otherwise, exactly one must contain $\Hidden(s_u)$ and the other $\Out$. Pair the two root records, add the two interior contributions and the common-boundary contribution to each guard's score, and accept only if an opponent is present and the target's elimination key is smaller than the minimum opponent key. There are polynomially many separator descriptors and polynomially many records per descriptor; see Section~\ref{sec:planar} for the proofs.

\begin{algobox}{{IV. PlanarCandidateLoss}}
\algio{A connected planar graph $G$, target $u$, and allowed opponents $A$.}{Whether $\Kill(G,u,A)=\yes$, together with a first-round witness when one exists.}
\begin{enumerate}[leftmargin=5mm,itemsep=.35ex,topsep=.45ex]
\item Test every one- and two-opponent subset of $A$. If one eliminates $u$ first, return true together with that subset as a witness.
\item Enumerate all oriented separator descriptors passing the individual checks described above with at most $b_0$ guards; compute their boundaries, interiors, and opposite orientations.
\item Starting from leaf states, repeatedly apply unary and binary transitions to locally compatible descriptor combinations until no new separator--state pair is inserted.
\item Examine pairs of stored root records with opposite orientations and the same guards and boundary. Require one of the legal target-status pairings described above. Complete the guard scores. If an opponent is present and the target's elimination key is strictly smaller than every opponent key, recover a witness through the stored predecessor pointers and return \mbox{{\sf true}} together with that witness.
\item Return \mbox{{\sf false}} if no direct witness or accepted completed pair is found.
\end{enumerate}
\end{algobox}

\subsection{Soundness, Completeness and Polynomial Termination}

\begin{lemma}\label{thm:planarsound}
If the algorithm accepts, then $\Kill(G,u,A)=\yes$.
\end{lemma}
\begin{proof}
Direct acceptance of a one- or two-opponent witness is checked by literal first-round scores and is sound. Otherwise take an accepting opposite-root state pair and finite derivations guaranteed by Lemma~\ref{lem:finitecertificate}. Let $K$ be the union of all guard identities in those derivations.

Every nontarget guard belongs to $A$ by the eligibility checks. The legal target statuses imply that the target belongs to $K$: an exposed target is a root guard, and a hidden target must trace back through its status transitions to an earlier exposed occurrence, since a hidden state cannot originate at a leaf. The presence bit is an OR over all represented opponents, so some opponent exists. Therefore $K=\{u\}\cup F$ for a nonempty $F\subseteq A$.

The local geometric and boundary tests apply to these concrete derivations. Lemmas~\ref{lem:locations}--\ref{lem:hiddenidentity} and Lemma~\ref{thm:globalguard} thus hold for this actual $K$, independently of its scores. At a leaf, every candidate in $K$ not in its guard set is outside the leaf closed region by Lemma~\ref{lem:locations}. Lemma~\ref{lem:shield} excludes all such candidates from its strict-interior voters, so the guard-only leaf counts are exact. There are no hidden leaf candidates.

Induct upward. Exact child vectors and hidden minima imply exact parent vectors, exact newly finalized candidate scores and an exact parent hidden minimum by Lemma~\ref{lem:finalization}. The four target cases either store its vector entry while exposed, compute its exact global score at first disappearance, propagate an already complete score, or retain valid absence. The OR recurrence exactly records opponent presence. This proves the numerical interpretation throughout both derivations.

At the root, Lemma~\ref{lem:voteraccounting} gives exact scores for all root guards. Every other opponent is hidden on exactly one side, so \eqref{eq:rootoppmin} is the true smallest opponent elimination key. The target score is exact in either legal target pattern. Consequently \eqref{eq:planaraccept} states that the target's key is strictly smaller than every opponent key. It is the first eliminated candidate of $K$. The common restart equivalence, Lemma~\ref{lem:restart}, proves that $\Kill(G,u,A)=\yes$.
\end{proof}

\begin{lemma}[Completeness of the planar separator algorithm]\label{thm:planarcomplete}
If $\Kill(G,u,A)=\yes$, then the algorithm accepts.
\end{lemma}
\begin{proof}
Choose an inclusion-minimal first-round witness $F$ by Lemma~\ref{lem:restart} and finiteness. If $|F|\le2$, the initial direct enumeration accepts. Otherwise $K=\{u\}\cup F$ has at least four candidates, so all structural hypotheses in this section hold.

Lemma~\ref{thm:adjacency} gives target adjacency for every opponent cell, and Lemma~\ref{lem:quotientwidth} bounds the quotient width. The direct prediagram construction is valid by Lemma~\ref{lem:directprediagram}. Its branchwidth is at most $b_0$ by Corollary~\ref{lem:prediagramwidth}. External result~\ref{ext:spherecuts} and the stated suppression rule give the two ordinary retained hierarchies. Lemma~\ref{thm:lift} gives simultaneous lifted separators and all required containment relations.

Every one of these separators occurs in the enumeration by Lemma~\ref{lem:enumeration}. Every selected candidate appears among their guards by Lemma~\ref{lem:coverage}. The guard eligibility checks hold because $K\setminus\{u\}=F\subseteq A$. Lemma~\ref{lem:truecompatible} proves all local tests for each retained transition. Nonguard selected candidates lie outside true leaves by Lemma~\ref{lem:coverage}, so their exact leaf states are inserted. If the target is not a leaf guard, that lemma puts it outside the closed leaf, making the outside status legal.

Induct upward along the finite retained hierarchies. The residual and score lemmas make the displayed transition formulas coincide with the true election counts. The appropriate target case is legal and exhaustive, and all true counts lie in the specified integer ranges. Therefore the exact state at each true parent is eventually generated by saturation. Both exact root states are generated. They have a legal target pairing because the target is either a common guard or represented strictly on one side. Opponent presence is true. Root completion reconstructs the actual first-round scores of $K$, and $F$ was chosen to eliminate $u$ first. Hence \eqref{eq:planaraccept} holds and the algorithm accepts.
\end{proof}

\begin{lemma}\label{lem:planarcomplexity}
Let $L$ be the number of accepted oriented descriptors and put
\[
 M=6(n+1)^{b_0+3}.
\]
There are at most $M$ numerical states per descriptor and at most $LM$ successful insertions. After preprocessing, a direct full-scan implementation takes
\[
 O\bigl((LM+1)L^3M^2(n+1)\bigr)
\]
word operations and $O(n^2+Ln+LM)$ words of space. In particular, for fixed $b_0=6$, its time is $O(n^{76})$ and its space is $O(n^{21})$ words.
\end{lemma}
\begin{proof}
A score vector has at most $(n+1)^{b_0}$ possibilities. The hidden minimum has at most $n(n+1)+1\le(n+1)^2$ possibilities. The target has $n+3\le3(n+1)$ possibilities, and the presence bit has two. Their product is at most $M$.

A complete scan may inspect every parent and every ordered child pair, at most $L^3$ choices, together with at most $M^2$ child-state pairs. Unary choices and leaves fit in this conservative bound. For fixed $b_0$, guard-set operations have constant size. Each geometry, site, boundary and residual check scans at most $O(n)$ vertices with constant-size candidate comparisons, using precomputed distances and side-membership arrays. A direct scan therefore costs $O(L^3M^2(n+1))$. Each unsuccessful final scan terminates the procedure; each other scan makes at least one new insertion. There are at most $LM$ insertions, proving the displayed bound.

For fixed $b_0$, use word-sized vertex and descriptor indices, distance and score entries, and predecessor pointers. The direct implementation keeps at most the following data simultaneously:
\begin{center}
\small
\begin{tabular}{@{}p{.31\linewidth}p{.45\linewidth}p{.17\linewidth}@{}}
\toprule
Stored data & Justification & Words\\
\midrule
Distances and fixed shortest-path trees & The distance matrix and at most $n$ parent arrays of length $n$; the plane embedding, $G^+$ and $J$ take $O(n)$ more. & $O(n^2)$\\
Descriptor descriptions and sides & At most $L$ accepted oriented descriptors; each has $O(1)$ guard/contact fields and at most $O(n)$ boundary, side and optional expanded-path entries. Inverse indices are constant size. & $O(Ln)$\\
Numerical states and predecessors & At most $M$ states per descriptor, each with at most $b_0$ scores, constant-size status/key fields and at most two predecessor pointers. & $O(LM)$\\
Temporary scan workspace & Inspect one parent/child combination at a time; use $O(n)$ temporary vertex masks, score counters and graph-search arrays, without storing all combinations. & $O(n)$\\
\bottomrule
\end{tabular}
\end{center}
The fixed shortest-path trees store $O(n^2)$ parent pointers in total; descriptor expansion and side computation reuse temporary arrays, rather than retaining an expansion for every trial descriptor. Storing all accepted descriptor sides explicitly costs $O(Ln)$ words. All inserted numerical states and their constant-size predecessor records cost $O(LM)$ words. Neither the number of full scans nor the number of examined child-state combinations multiplies simultaneous storage. Therefore one $\Kill$ query, including preprocessing and witness recovery from the stored predecessor DAG, uses
\[
 O(n^2+Ln+LM)=O(n^{21})\ \text{words}
 \qquad(b_0=6,\ L=O(n^{12}),\ M=O((n+1)^9)).
\]
The witness itself uses at most $n$ entries. This is an upper bound on the space of one loss query, not a bound for a transfer computation that retains the outputs of several calls.

Lemma~\ref{lem:enumeration} gives $L=O(n^{2b_0})$. Thus the dominant displayed bound is $O(L^4M^3(n+1))=O((n+1)^{11b_0+10})$, a deliberately coarse $O((n+1)^{76})$ bound at $b_0=6$. It is a tractability bound, not a claim of practical efficiency. Computing distances and the descriptor family is also polynomial and is dominated by this bound. All counted integers and indices have $O(\log n)$ bits because the exponent is fixed; bit-operation costs remain polynomial.
\end{proof}

\paragraph{Conditional parameter-five calculation.}
We retain $b_0=6$ in all stated bounds. Although Corollary~\ref{lem:prediagramwidth} gives branchwidth at most five, using $b_0=5$ throughout would require rechecking the descriptor enumeration, simultaneous lifting, retained hierarchies, transitions and recovery with the tighter capacity. If all these arguments are checked to remain valid with $b_0=5$, the same full-scan analysis gives $L=O(n^{10})$ and $M=O((n+1)^8)$. Its dominant time term $O(L^4M^3(n+1))$ becomes $O(n^{65})$, and $O(n^2+Ln+LM)$ becomes $O(n^{18})$ words for one loss query. Theorem~\ref{thm:transfer} would then give $O(n^{66})$ verification time and $O(n^{67})$ time for minimum-zone computation or enumeration. These are conditional substitutions in the existing formulas, not bounds asserted for a newly verified parameter-five implementation; the theorem statements and summary table continue to use six.

\begin{theorem}[Planar tractability]\label{m:planartheorem}\label{thm:planar}
For every connected planar graph under the voting convention of Section~\ref{sec:model}, every target $u$ and every allowed set $A$, a direct full-scan implementation decides $\Kill(G,u,A)$ and recovers a first-round witness when one exists in $O(n^{76})$ word operations and $O(n^{21})$ words of space. 
Zone verification takes $O(n^{77})$ time; minimum-zone computation and enumeration of all nonempty zones take $O(n^{78})$ time.
\end{theorem}
\begin{proof}
Lemmas~\ref{thm:planarsound} and~\ref{thm:planarcomplete} 
establish correctness;
termination and the polynomial bound follow from Lemma~\ref{lem:planarcomplexity}. A direct small witness is already explicit. Otherwise traverse the predecessor records reachable from the accepted root pair and collect their guard identities, deleting $u$ from the union. At most $LM$ different stored entries are visited, even if their fully unrolled derivation trees would repeat subderivations. The soundness proof shows that this union is a valid nonempty $F$. It can also be checked directly by Lemma~\ref{lem:evaluate}.
\end{proof}

\begin{corollary}\label{cor:planarzones}
Under the same assumptions, exclusion-zone verification is in $\Pclass$, the unique minimum exclusion zone can be computed in $\FP$, and all nonempty exclusion zones can be enumerated in polynomial total time. These statements use unrestricted elections in the definition of a zone, not elections restricted to one fixed allowed set.
\end{corollary}
\begin{proof}
Apply Theorem~\ref{thm:transfer} with the planar algorithm for $\Kill$. Its verification calls use the varying allowed set $V(G)\setminus S$, and its minimum-zone procedure tests the at most $n$ distinct singleton closures. Each call is covered by Theorem~\ref{m:planartheorem}, which permits an arbitrary input set $A$. The family is nested and contains a unique smallest nonempty member by Lemma~\ref{lem:nesting}.
\end{proof}

\section{IRV on Exact One-Apex Graphs}\label{m:apex}\label{sec:apex}

A graph is exact one-apex if it is nonplanar but becomes planar after deleting some vertex. Candidate loss is $\NP$-complete on this class, already when the specified target is a planarizing vertex. This restricted hardness also establishes hardness without a restriction on the target's planarizing role; it does not classify the case in which the target is required to be non-planarizing. Section~\ref{subsec:full-apex-classification} explains the relation between these input families.
Here $\ecc_G(u)=\max_{v\in V(G)}d_G(u,v)$ and $\diam(G)=\max_{v,w\in V(G)}d_G(v,w)$.
The reduction is illustrated in the following figure.
\begin{figure}[htbp]
\centering
\begin{tikzpicture}[x=1cm,y=1cm,every node/.style={font=\scriptsize}]
\node[draw,circle,inner sep=2pt] (u) at (0,0) {$u$};
\node[draw,circle,inner sep=1.5pt] (p) at (1.05,0) {$\alpha_x$};
\node[draw,circle,inner sep=1.5pt] (r) at (2.1,0) {$\beta_x$};
\node[draw,circle,inner sep=2pt] (e) at (3.15,0) {$e_x$};
\node[draw,circle,inner sep=1.5pt] (z) at (4.2,0) {$\gamma_{x,i}$};
\node[draw,circle,inner sep=2pt] (c) at (5.25,0) {$c_i$};
\draw (u)--(p)--(r)--(e)--(z)--(c);
\node at (1.56,-.38) {length $3$};
\node at (4.18,-.38) {length $2$};
\node[draw,circle,inner sep=1pt] (bl) at (-.3,1.04) {$\ell$};
\node at (.1,1.58) {$B$ baseline leaves};
\draw (u)--(bl);
\node[draw,circle,inner sep=1pt] (el) at (3.15,1.04) {$\ell_x$};
\node at (3.15,1.58) {$W$ element leaves};
\draw (e)--(el);
\node at (4.05,-.95) {$x\in C_i$};
\end{tikzpicture}
\caption{One element--triple incidence in the one-apex reduction. An element $x$ and a triple $C_i$ containing it are represented by vertices $e_x$ and $c_i$, respectively. A pictured leaf represents one of $B$ distinct leaves at $u$ or $W$ distinct leaves at $e_x$; there is one length-three $u$--$e_x$ path per element and one length-two $e_x$--$c_i$ path per incidence. An element-bundle leaf is distance four from $u$ and three from an incident $c_i$. (The planar incidence graph and the fixed padding component are not shown.) 
Displayed vertices are unit voters and eligible candidate sites in the unrestricted instance.}
\label{fig:apex-gadget}
\end{figure}

This section gives the complete proof of Theorem~\ref{m:apextheorem}. The election model is exactly the one
defined in Section~\ref{sec:model}; in particular, every constructed vertex hosts one voter and is also a possible candidate location.

We first prove hardness for the restricted case in which only the vertices corresponding to triples of the matching instance may participate alongside the target. We then remove this restriction and allow every other vertex of the constructed graph to participate as well. The main additional task is to prove that none of these
auxiliary candidates can be present in a round that eliminates the target.

Finally, we use Lemma~\ref{lem:zonekill} and Lemma~\ref{lem:nesting} to transfer the unrestricted hardness result to exclusion-zone verification and minimum-zone computation. The hardness argument is independent of the planar algorithm in Section~\ref{sec:planar}.

\subsection{Exact One-Apex Inputs and Complexity Conventions}
\begin{definition}\label{def:apex}
Define
\[
 \apx(G)=\min\{|X|:X\subseteq V(G),\ G-X\text{ is planar}\}.
\]
An exact one-apex graph satisfies $\apx(G)=1$: the graph itself is nonplanar and at least one single-vertex deletion is planar. For a vertex $v$, define
\[
 \ecc_G(v)=\max_{x\in V(G)}d_G(v,x),\qquad
 \rad(G)=\min_v\ecc_G(v),\qquad
 \diam(G)=\max_v\ecc_G(v).
\]
Let $\mathcal A_{5,10}$ consist of triples $(G,u,\id)$ with these properties:
\begin{enumerate}[label=(\roman*)]
\item $G$ is finite, simple, connected, unweighted and bipartite;
\item $G-u$ is planar, whereas $G$ is nonplanar;
\item $u$ has the smallest identifier among all vertices;
\item $\ecc_G(u)\le5$ and $\diam(G)\le10$.
\end{enumerate}
The supplied vertex $u$ is both the target and a planarizing vertex. Property (iv) also implies $\rad(G)\le5$.
\end{definition}

\begin{lemma}\label{lem:apexrecognize}
Membership in $\mathcal A_{5,10}$ can be decided in polynomial time in the explicit graph encoding.
\end{lemma}
\begin{proof}
Check the vertex and edge lists for loops, repeated edges and valid identifiers. Breadth-first search checks connectedness; assigning alternating colors in each search checks bipartiteness. Apply External result~\ref{ext:planarity} separately to $G-u$ and $G$, requiring respectively a planar and a nonplanar answer. Breadth-first search from $u$ gives its eccentricity, and searches from all vertices give the diameter. Finally compare identifiers to confirm that $u$ is smallest. Every step is polynomial. Thus no hardness statement below relies on an unrecognizable promise.
\end{proof}

Whenever we state a complexity result on $\mathcal A_{5,10}$, the input is first checked to satisfy Definition~\ref{def:apex}, using Lemma~\ref{lem:apexrecognize}. Inputs failing these conditions are rejected. For the minimum-zone function, such an invalid input may instead return a fixed error symbol. All instances produced by the reduction itself satisfy Definition~\ref{def:apex}.

\subsection{The Source Problem}

\begin{definition}\label{def:planar3dm}
An instance consists of three pairwise disjoint sets $X,Y,Z$ with
$|X|=|Y|=|Z|=r$ and a family $\mathcal C\subseteq X\times Y\times Z$. The \emph{incidence graph} has a vertex for each element and for each triple, and an edge from a triple to each of its three elements. This incidence graph must be planar. An \emph{exact matching} is a subfamily that contains every element of $X\cup Y\cup Z$ exactly once. It necessarily has $r$ triples.
\end{definition}

\begin{external}[Planar three-dimensional matching]\label{ext:planar3dm}
Planar three-dimensional matching with the incidence-graph planarity condition of Definition~\ref{def:planar3dm} is $\NP$-complete, by Dyer and Frieze~\cite{DyerFrieze}.
\end{external}

No connectedness condition or bound on element occurrences is used by the reduction. The three parts specify the matching instance; they do not describe the bipartition of the eventual voting graph. The planarity promise concerns the incidence graph, not a graph obtained by joining elements within each triple.

\subsection{A Fixed Component that Preserves the Matching Answer and
Forces Nonplanarity}

First rename source elements so that they are disjoint from all following labels. Add three new four-element parts
\[
 X^\dagger=\{a,p_1,p_2,p_3\},\quad
 Y^\dagger=\{b,q_1,q_2,q_3\},\quad
 Z^\dagger=\{z,d_1,d_2,d_3\},
\]
and the seven triples
\begin{equation}\label{eq:forcingtriples}
 H^\dagger=(a,b,z),\qquad
 U_i^\dagger=(p_i,q_i,d_i),\qquad
 T_i^\dagger=(a,b,d_i)\quad(1\le i\le3).
\end{equation}
Each triple has one element from each new part. The component is disjoint from the input instance. The triples $T_i^\dagger$ will certify nonplanarity of the voting graph even though none belongs to the component's exact matching.

\begin{lemma}[Planarity and unique matching of the fixed component]\label{lem:paddingcomponent}
The incidence graph of \eqref{eq:forcingtriples} is planar. Its unique exact matching is
\[
 \mathcal M^\dagger=\{H^\dagger,U_1^\dagger,U_2^\dagger,U_3^\dagger\}.
\]
\end{lemma}
\begin{proof}
The vertices $a,b,H^\dagger,T_1^\dagger,T_2^\dagger,T_3^\dagger$ form $K_{2,4}$. It has a planar drawing as four internally disjoint two-edge paths from $a$ to $b$, arranged in order around those endpoints. The element $z$ is a leaf attached to $H^\dagger$. For each $i$, all remaining vertices associated with $T_i^\dagger$ form the attached tree consisting of the path $T_i^\dagger-d_i-U_i^\dagger$ and the two leaves $p_i,q_i$ at $U_i^\dagger$. Draw these trees in small disjoint neighborhoods of their attachment vertices. No existing edge is crossed, proving planarity.

For the matching, $z$ occurs only in $H^\dagger$, forcing that triple. Each $p_i$ occurs only in $U_i^\dagger$, forcing all three such triples; $q_i$ gives the same forcing condition. The four forced triples cover all twelve new elements exactly once. Each $T_i^\dagger$ intersects them in $a,b,d_i$, so it cannot be added to an exact matching. There is no other triple in the component. This proves existence and uniqueness.
\end{proof}

\begin{lemma}\label{lem:padding}
Let $I^\oplus$ be the disjoint union of a source instance $I$ and the fixed component. Then $I^\oplus$ has planar incidence graph and equal part size $h=r+4$, and
\[
 I\text{ has an exact matching}\quad\Longleftrightarrow\quad
 I^\oplus\text{ has an exact matching}.
\]
\end{lemma}
\begin{proof}
A disjoint union of planar graphs is planar, and exactly four elements are added to each part. An exact matching of $I$ combined with $\mathcal M^\dagger$ covers the padded parts exactly once. Conversely, any padded exact matching restricts to an exact matching on each disjoint component and, on the fixed component, to its unique one. Its remaining triples are an exact matching of $I$.
\end{proof}

For the rest of the construction write $\mathcal U$ for the padded universe of $3h$ elements and $C_1,\ldots,C_m$ for its triples. In particular $h\ge4$ and $m\ge7$. The symbol $n$ continues to denote the order of the final voting graph, not a matching-part size.

\subsection{Construction of the Base Graph}

A path of length $k$ has $k$ edges. All vertices described as fresh are new vertices, distinct from every previously introduced vertex unless an endpoint is explicitly shared.

Create the following base graph $G_0$.

\paragraph*{Target and element centers.}
Create the target vertex $u$. For each $x\in\mathcal U$, create one vertex $e_x$, the \emph{element center}. It is the common attachment point for the objects associated with element $x$.

\paragraph*{Triple candidates.}
For every triple $C_i$, create a vertex $c_i$. These are the intended choices of the reduction: selecting $c_i$ will represent selecting $C_i$. We call $c_i$ the \emph{triple candidate} corresponding to $C_i$.

\paragraph*{Private access paths.}
For every element $x$, create two vertices $\alpha_x,\beta_x$ and the path
\begin{equation}\label{eq:accesspath}
 P^u_x:\quad u-\alpha_x-\beta_x-e_x.
\end{equation}
Its length is three. Two different access paths share only $u$, and their internal vertices have no other base edges.

\paragraph*{Incidence paths.}
For every incidence $x\in C_i$, create a vertex $\gamma_{x,i}$ and the path
\begin{equation}\label{eq:incidencepath}
 P_{x,i}:\quad e_x-\gamma_{x,i}-c_i.
\end{equation}
Its length is two. Each internal vertex $\gamma_{x,i}$ is globally fresh. Paths may share their designated element-center or triple-candidate endpoints, but no internal vertices.

There are no further base edges. Formally,
\begin{align*}
 V(G_0)={}&\{u\}\ \dot\cup\ \{e_x:x\in\mathcal U\}
 \ \dot\cup\ \{c_i:1\le i\le m\}\\
 &\dot\cup\ \{\alpha_x,\beta_x:x\in\mathcal U\}
 \ \dot\cup\ \{\gamma_{x,i}:x\in C_i\},\\
 E(G_0)={}&\{u\alpha_x,\alpha_x\beta_x,\beta_xe_x:x\in\mathcal U\}
 \cup\{e_x\gamma_{x,i},\gamma_{x,i}c_i:x\in C_i\}.
\end{align*}
The disjoint unions specify every permitted identification. All edges have length one.

The choice of lengths three and two anticipates a voter attached by one edge to $e_x$: it will be four edges from $u$ and three from an incident $c_i$. This is a motivation, not yet a shortest-path assertion; the proof below excludes other shorter routes explicitly.

\subsection{Ordinary Voters, Leaf Bundles and the Scale Inequalities}

Before adding leaves, define the exact number
\begin{equation}\label{eq:N0}
 N_0:=|V(G_0)|=1+3h+m+6h+3m=1+9h+4m.
\end{equation}
There is one target, $3h$ centers, $m$ triple candidates, $6h$ access interiors and $3m$ incidence interiors. Every base vertex hosts a voter. These $N_0$ voters will be called \emph{ordinary voters}. Their detailed assignments need not be controlled individually: their combined contribution to any candidate is at most $N_0$.

For each element $x$, attach $W$ fresh degree-one vertices to $e_x$, forming its \emph{element bundle}
\[
 L_x=\{\ell_{x,1},\ldots,\ell_{x,W}\}.
\]
Also attach $B$ fresh degree-one vertices directly to $u$, forming the \emph{baseline}
\[
 \mathcal B_u=\{b_1,\ldots,b_B\}.
\]
Every one of these leaves is a separate unit voter and a possible candidate. There are no weighted voters.

We derive $W$ and $B$ before using them. A selected exact-matching triple must beat a target supported by its baseline and all ordinary voters, requiring $B+N_0<3W$. A candidate with at most two bundles must lose to the target even if all ordinary voters support that opponent, requiring $2W+N_0<B$. Together these require
\begin{equation}\label{eq:scalerequirements}
 2W+N_0<B<3W-N_0.
\end{equation}
The interval is nonempty as soon as $W>2N_0$. Choose
\begin{equation}\label{eq:scales}
 W:=4(N_0+1),\qquad B:=\frac{5W}{2}=10(N_0+1).
\end{equation}
These are positive integers, with $W$ divisible by four.

\begin{lemma}[All score inequalities used in the reduction]\label{lem:scalemargins}
The choices in \eqref{eq:scales} satisfy
\begin{gather}
 2W+N_0<B,\qquad B+N_0<3W,\label{eq:scalecore}\\
 3W+N_0<B+W,\qquad W+N_0<B,\qquad N_0<W/4.\label{eq:scaleextra}
\end{gather}
The first two differences are both $N_0+2$.
\end{lemma}
\begin{proof}
Since $W/2=2N_0+2$,
\[
 B-(2W+N_0)=W/2-N_0=N_0+2,
\]
and
\[
 3W-(B+N_0)=W/2-N_0=N_0+2.
\]
Both are positive. Adding $W$ to the first strict comparison gives $3W+N_0<B+W$. The inequality $W+N_0<B$ follows from $W+N_0<2W+N_0<B$. Finally $W/4=N_0+1>N_0$.
\end{proof}

The third inequality means that one uncovered bundle makes the target too strong for any triple candidate to eliminate. The fourth will exclude every nonbaseline auxiliary candidate. None of these comparisons assumes an exact assignment of an ordinary voter; the worst-case allowance $N_0$ covers all of them.

\subsection{The Full Graph, Identifiers, and Polynomial Size}

The final voting graph $G$, obtained from the base graph $G_0$ by adding the element and baseline leaves, has precisely the following vertex and edge sets:
\begin{align*}
 V(G)&=V(G_0)\ \dot\cup\
       \left(\mathop{\dot\bigcup}_{x\in\mathcal U}L_x\right)
       \ \dot\cup\ \mathcal B_u,\\
 E(G)&=E(G_0)\cup\{e_x\ell:\ell\in L_x,\ x\in\mathcal U\}
                    \cup\{ub:b\in\mathcal B_u\}.
\end{align*}
Set $\id(u)=0$ and give the remaining vertices ranks $1,\ldots,n-1$ in any fixed construction order. Thus the target wins every equal-distance comparison against another candidate, and loses an equal-score elimination tie to none of them.

Define the intended candidate set
\[
 \Areal:=\{c_1,\ldots,c_m\}.
\]
The restricted instance uses $A=\Areal$. The unrestricted instance on the identical graph uses $A=V(G)\setminus\{u\}$. A vertex of $V(G)\setminus(\{u\}\cup\Areal)$ is \emph{auxiliary}; the terminology does not make that vertex unavailable in the unrestricted instance.

\begin{lemma}\label{lem:apexsize}
The reduction writes a graph with
\[
 n=N_0+3hW+B,\qquad |E(G)|=9h+6m+3hW+B.
\]
Both are $O(h(h+m))$, and the complete instance is computable in polynomial time in the source encoding.
\end{lemma}
\begin{proof}
Equation~\eqref{eq:N0} counts all base vertices. There are $3h$ disjoint bundles of $W$ leaves and $B$ baseline leaves. There are $3h$ access paths of three edges each and $3m$ incidence paths of two edges each, giving $9h+6m$ base edges. Every new leaf adds one edge. These yield the displayed exact counts.

The quantities $N_0,W,B$ are linear in $h+m$. Therefore the leaf count and edge count are $O(h(h+m))$; here $h\ge4$, so the additive linear terms fit the same bound. Each vertex and edge can be emitted once by the explicit construction loops. Identifiers use $O(\log n)$ bits. This is a polynomial-size explicit graph, not a compressed instruction to a weighted-voter simulator.
\end{proof}

\subsection{Why the Constructed Graph is Exact One-Apex}

\begin{lemma}\label{lem:deleteapex}
The graph $G-u$ is planar.
\end{lemma}
\begin{proof}
Embed the padded planar incidence graph. Replacing every incidence edge by its two-edge path preserves this embedding, with the new subdivision vertex placed on the original edge. After deleting $u$, the remainder of access path \eqref{eq:accesspath} is the pendant path $e_x-\beta_x-\alpha_x$, attached to the incidence graph only at $e_x$. Draw it in a sufficiently small incident face neighborhood. Bundle leaves can be drawn in similarly small neighborhoods of their centers. Different attachments can be made disjoint. Baseline leaves become isolated vertices and can be placed in any face. Each operation preserves planarity, proving the lemma even when the incidence graph is disconnected.
\end{proof}

The implication from a $K_{3,3}$ subdivision to nonplanarity is standard~\cite[Section~10.5]{BondyMurty}; the work here is identifying disjoint branch paths in this specific voting graph.
\begin{lemma}\label{lem:k33}
The full graph $G$ contains a subdivision of $K_{3,3}$ with branch-vertex classes
\[
 \{u,e_a,e_b\}\quad\text{and}\quad
 \{c_{T_1^\dagger},c_{T_2^\dagger},c_{T_3^\dagger}\}.
\]
In particular $G$ is nonplanar.
\end{lemma}
\begin{proof}
For each $i=1,2,3$ select these three paths:
\begin{align*}
 e_a&-\gamma_{a,T_i^\dagger}-c_{T_i^\dagger},\\
 e_b&-\gamma_{b,T_i^\dagger}-c_{T_i^\dagger},\\
 u&-\alpha_{d_i}-\beta_{d_i}-e_{d_i}
                -\gamma_{d_i,T_i^\dagger}-c_{T_i^\dagger}.
\end{align*}
The first two exist because $a,b\in T_i^\dagger$. The third combines the access path for $d_i$ with its incidence in $T_i^\dagger$.

The elements $a,b,d_1,d_2,d_3$ are pairwise distinct. All access interiors are private and all incidence interiors are fresh. Thus the three paths starting at $u$ share only $u$; the paths starting at $e_a$ share only $e_a$, and those starting at $e_b$ share only $e_b$. A path from one of these three classes meets a path from another only if they end at the same $c_{T_i^\dagger}$. No branch vertex is internal to a selected path. Therefore the nine paths have exactly the endpoint intersections of the nine edges of $K_{3,3}$, and no internal intersections.

For completeness, $K_{3,3}$ is nonplanar by Euler's formula: a planar embedding of this simple bridgeless bipartite graph would have every facial boundary length at least four. With six vertices and nine edges, Euler gives five faces, whereas counting edge sides would give $4\cdot5\le2\cdot9$, a contradiction. A planar embedding of a subdivision would give a planar embedding of $K_{3,3}$ by suppressing the internal degree-two vertices along these selected paths. Hence the subdivision is nonplanar. A planar graph cannot contain a nonplanar subgraph, proving that $G$ is nonplanar. Extra edges in $G$ do not affect this $K_{3,3}$ subdivision, because the argument only requires it to occur as a subgraph, not as an induced subgraph.
\end{proof}

\begin{lemma}\label{lem:apexclass}
The construction satisfies $(G,u,\id)\in\mathcal A_{5,10}$, $\apx(G)=1$, and $\rad(G)\le5$.
\end{lemma}
\begin{proof}
Fresh vertices prevent loops or parallel edges, and all edges have length one. Every element center reaches $u$ by an access path. Every triple vertex reaches a center by one of its three incidences. Every other vertex is on such a path or is a leaf. Thus the graph is connected, including when the source has isolated elements or multiple incidence components.

An explicit bipartition gives color zero to $u$, every $\beta_x$, every $\gamma_{x,i}$ and every bundle leaf, and color one to every $\alpha_x$, center $e_x$, triple candidate $c_i$ and baseline leaf. Inspecting the complete edge list shows that every edge joins opposite colors.

The access vertices are at most two steps from $u$, centers at most three, bundle leaves and incidence interiors at most four, and triple candidates at most five. Baseline leaves are one step away. These explicit paths show $\ecc_G(u)\le5$; therefore $\rad(G)\le5$. For any vertices $v,w$, the triangle inequality through $u$ gives $d_G(v,w)\le d_G(v,u)+d_G(u,w)\le10$, proving the diameter bound.

By Lemma~\ref{lem:deleteapex}, $\apx(G)\le1$. By Lemma~\ref{lem:k33}, $G$ is not planar, so $\apx(G)\ge1$. The identifier construction supplies the remaining class condition.
\end{proof}

\subsection{Shortest-Path Distances and How the Bundle Voters Vote}

All statements in this subsection refer to shortest paths in the full graph, including all the added leaves and access paths. They are not distances in just the source incidence graph.

\begin{lemma}\label{lem:apexdistances}
For every element $x$ and triple candidate $c_i$,
\[
 d_G(u,e_x)=3,\qquad d_G(u,c_i)=5.
\]
\end{lemma}
\begin{proof}
A shortest path can be taken simple. It cannot leave $u$ through a baseline leaf and then proceed to a different vertex, because that leaf has degree one. Every other edge at $u$ enters one private access path. Its degree-two internal vertices force the path to reach the corresponding element center after three edges before reaching any other part of the graph. Thus no center is reached in fewer than three steps, and the direct access path to $e_x$ attains three.

After first reaching any center, a path to a triple candidate needs at least one entire incidence path of length two: a center and a triple candidate are not adjacent, and the only edges at a triple candidate are its incidence edges. Thus no triple candidate is reached from $u$ in fewer than five steps. Taking the access path to any of its three elements and then the corresponding incidence path attains five. Added leaves cannot shorten a simple path between other vertices.
\end{proof}

\begin{lemma}\label{lem:centerdistance}
For $x\ne y$, $d_G(e_x,e_y)\ge4$.
\end{lemma}
\begin{proof}
A path through $u$ has length at least $3+3=6$ by Lemma~\ref{lem:apexdistances}. A simple path avoiding $u$ cannot enter a bundle leaf or enter and later leave a pendant access fragment: it would have to repeat the same attachment to proceed. It therefore lies in the subdivided incidence graph. A path between distinct element vertices in the original bipartite incidence graph has positive even length, hence at least two incidence edges. Subdivision doubles this to at least four graph edges. Taking the minimum over both possibilities proves the bound.
\end{proof}

\begin{lemma}\label{lem:bundlerouting}
For each $\ell\in L_x$ and triple candidate $c_i$,
\[
 d_G(\ell,u)=4,\qquad
 d_G(\ell,c_i)=3\ \text{if }x\in C_i,\qquad
 d_G(\ell,c_i)\ge7\ \text{if }x\notin C_i.
\]
If the active set consists only of $u$ and triple candidates, the entire bundle $L_x$ chooses an incident selected candidate when one exists, and otherwise chooses $u$.
\end{lemma}
\begin{proof}
Every path from $\ell$ begins with its unique edge to $e_x$. The first equality follows from $d_G(e_x,u)=3$. For $x\in C_i$, its incidence path has length two, and no shorter center--candidate path exists because they are distinct and not adjacent. Thus the leaf distance is exactly three.

Suppose $x\notin C_i$. A shortest simple $e_x$--$c_i$ path either contains $u$ or avoids it. A path containing $u$ has length at least $d_G(e_x,u)+d_G(u,c_i)=8$. A path avoiding $u$ cannot use a leaf or make an excursion into a pendant access fragment, so it lies in the subdivided incidence graph. Before subdivision, an element-to-triple path has odd length; since this is a nonincidence, its length is not one and is at least three. Its length after subdivision is at least six. In both cases the leaf distance is at least seven.

In the restricted active set every candidate is outside $L_x$. For all such candidates $c$ and all $\ell\in L_x$,
\[
 d_G(\ell,c)=1+d_G(e_x,c).
\]
Thus every leaf in the bundle has the identical candidate order. An incident triple candidate is strictly closer than $u$, and a nonincident one is strictly farther. If several selected incident candidates are tied at distance three, their common identifier order selects the same one for every bundle voter. This proves indivisible whole-bundle routing.
\end{proof}

The last lemma deliberately states its active-set restriction. A bundle leaf that is itself active receives its own zero-distance vote and need not behave like another leaf. Such configurations are covered by the auxiliary-candidate proof, not by an incorrect blanket interchangeability assertion.

By Lemma~\ref{lem:accept}, because $u$ has the smallest identifier among all vertices, it is eliminated in the first round of any election containing at least one other candidate exactly when
\begin{equation}\label{eq:apexstrict}
 \score_K(u)<\score_K(q)\qquad\text{for every }q\in K\setminus\{u\}.
\end{equation}
Equality is insufficient: a tied opponent has a larger identifier and precedes $u$ in the elimination order.

\subsection{Exact Matching if and only if the Target can be Eliminated Using Only Triple Candidates}

\begin{lemma}\label{thm:apexrestricted}
The problem $\Kill(G,u,A)$ is $\NP$-complete on $\mathcal A_{5,10}$. The reduction uses $A=\Areal$, consisting precisely of the vertices $c_i$ corresponding to triples.
\end{lemma}
\begin{proof}
\emph{Membership.} Check validity and class membership by Lemma~\ref{lem:apexrecognize}. For membership in $\NP$, guess a nonempty set $F\subseteq A$. Whether $F$ eliminates $u$ in the first round can be checked in
polynomial time by Lemmas~\ref{lem:evaluate} and~\ref{lem:restart}.

\emph{Reduction and its valid output.} Start with an instance of Planar three-dimensional matching. Append the fixed component and build the graph and identifiers above, with $A=\Areal$. Lemma~\ref{lem:apexsize} proves polynomial construction time and Lemma~\ref{lem:apexclass} proves that every output belongs to the stated graph class.

\emph{Forward direction.} Suppose the padded instance has an exact matching $\mathcal M$. Activate $u$ and the $h$ triple candidates corresponding to $\mathcal M$. Every element is incident with exactly one selected triple, so its bundle supports that triple's candidate by Lemma~\ref{lem:bundlerouting}. Each selected triple therefore has at least $3W$ votes. No baseline leaf is a candidate. Every baseline voter is one edge from $u$ and at least two from any other candidate, so all $B$ support $u$. No element bundle supports $u$, and there are only $N_0$ ordinary voters. Hence
\[
 \score_K(u)\le B+N_0<3W\le\score_K(c_i)
 \qquad(C_i\in\mathcal M).
\]
The strict comparison \eqref{eq:apexstrict} holds and $u$ is eliminated first.

\emph{Reverse direction: every element is covered.} Suppose a nonempty $F\subseteq\Areal$ eliminates $u$ first. Every baseline voter still supports $u$. If element $x$ were uncovered by the selected triples, its entire bundle would also support $u$, giving $\score_K(u)\ge B+W$. A selected triple candidate can receive bundle votes only from the three elements in its own triple. Even crediting all ordinary voters to it gives
\[
 \score_K(c_i)\le3W+N_0<B+W\le\score_K(u).
\]
This contradicts \eqref{eq:apexstrict}. Thus all $3h$ elements are covered.

\emph{Reverse direction: the number of selected triples.} Put $k=|F|$. Since $k$ triples cover $3h$ elements and each has three elements, $k\ge h$. If $k>h$, by Lemma~\ref{lem:bundlerouting}, all voters in each of the
$3h$ bundles vote for the same selected triple candidate. If every selected candidate received at least three bundles, their total would be at least $3k>3h$, impossible. Some candidate therefore receives at most two bundles. Its score is at most $2W+N_0<B$, whereas $u$ has at least its baseline $B$. This again contradicts \eqref{eq:apexstrict}. Hence $k=h$.

\emph{Reverse direction: exactness.} The selected $h$ triples have exactly $3h$ incidences in total and cover a universe of $3h$ elements. Let $t_x$ be the number of selected occurrences of element $x$. Coverage says $t_x\ge1$ for every element and the incidence count says $\sum_x t_x=3h$. Therefore $t_x=1$ for every element: the triples form an exact matching.

This proves equivalence between a padded exact matching and a first-round restricted witness. By Lemma~\ref{lem:restart}, first-round witnesses are equivalent to eventual loss in the definition of $\Kill$. By Lemma~\ref{lem:padding}, the padded matching answer equals the original source answer. External result~\ref{ext:planar3dm} now gives $\NP$-hardness; together with membership it gives completeness.
\end{proof}

\subsection{Why Auxiliary Candidates Cannot Eliminate the Target}

To pass from restricted to unrestricted opponents, every auxiliary location must be controlled. There are exactly five types: baseline leaves $b\in\mathcal B_u$, centers $e_x$, access interiors $\alpha_x,\beta_x$, bundle leaves $\ell\in L_x$, and incidence interiors $\gamma_{x,i}$. This follows directly from the vertex sets listed in the construction of $G_0$ and from the two families of leaves added after Equation~\eqref{eq:N0}.

We handle baseline-leaf candidates separately. For every other auxiliary candidate associated with an element $x$, we show that voters in each other element bundle $L_y$, with $y\ne x$, prefer the target $u$ to that candidate. These comparisons remain valid regardless of which additional candidates participate. The candidate can therefore receive votes from at most one element bundle. Its possible votes from the ordinary voters
are bounded separately by $N_0$.

\begin{lemma}\label{lem:baseline}
Let $K$ be any active set containing $u$.
\begin{enumerate}[label=(\roman*)]
\item 
An active baseline leaf has score exactly one, and its presence prevents elimination of $u$ in that round.
\item If no baseline leaf is active, every baseline voter chooses $u$, giving $\score_K(u)\ge B$.
\end{enumerate}
\end{lemma}

\begin{proof}
For a baseline leaf $b$ and any $v\ne b$, its unique incident edge forces
\[
 d_G(v,b)=d_G(v,u)+1.
\]
Thus every such voter strictly prefers $u$ to $b$ while both are active. The voter at $b$ itself supports $b$, so its score is exactly one. The target has at least its own vote. If its score exceeds one, $b$ has a smaller score; if it equals one, $b$ has a larger identifier and is ahead of $u$ in the minimum-key order. Either way $u$ is not eliminated. This also covers the case in which many or all baseline leaves are candidates.

For (ii), an inactive baseline leaf is at distance one from $u$. For any other active $q$, its distance to $q$ is $1+d_G(u,q)\ge2$. Therefore each of the $B$ baseline voters supports $u$, proving the bound.
\end{proof}

\begin{lemma}\label{lem:auxaccess}
Fix an element $x$ and an active auxiliary candidate $q$ equal to $e_x$, $\alpha_x$, $\beta_x$, or a bundle leaf adjacent to $e_x$. In an active set containing $u$, no voter in $L_y$ for $y\ne x$ chooses $q$.
\end{lemma}
\begin{proof}
Fix $\ell_y\in L_y$ with $y\ne x$. Its distance to $u$ is four. If $q$ is on the access path $P^u_x\setminus\{u\}$, any path to it enters that path through $u$ or through $e_x$; the access interiors have no other neighbors. A route through $u$ has length at least $4+1=5$. A route through $e_x$ has length at least
\[
 d_G(\ell_y,e_x)=1+d_G(e_y,e_x)\ge5
\]
by Lemma~\ref{lem:centerdistance}; any remaining portion to $q$ is nonnegative. For $q=e_x$ this latter bound already applies directly. Thus all these candidates are strictly farther than $u$.

If $q$ is a bundle leaf at $e_x$, its unique final edge adds one to the distance to $e_x$, giving at least six. Again $u$ is strictly closer. These strict pairwise comparisons exclude $q$ regardless of all other active candidates.
\end{proof}

\begin{lemma}\label{lem:auxincidence}
If $q=\gamma_{x,i}$ is active and $u$ is active, no voter in $L_y$ for $y\ne x$ chooses $q$.
\end{lemma}
\begin{proof}
The only neighbors of $q$ are $e_x$ and $c_i$. A path from $\ell_y\in L_y$ through $e_x$ has length at least
\[
 1+d_G(e_y,e_x)+1\ge6.
\]
A path through $c_i$ has length at least
\[
 d_G(\ell_y,c_i)+1\ge3+1=4,
\]
since the shortest possible bundle-to-triple-candidate distance is three by Lemma~\ref{lem:bundlerouting}. Therefore $d_G(\ell_y,q)\ge4=d_G(\ell_y,u)$. When the inequality is strict, $u$ is closer. When equality holds, $u$ has smaller identifier than $q$, so the voter still strictly prefers $u$ in the lexicographic order. This is the auxiliary case that genuinely uses the target's smallest identifier for a distance tie.
\end{proof}

\begin{corollary}\label{cor:auxbundle}
For any active set containing $u$, every active auxiliary candidate outside $\mathcal B_u$ receives at most $W$ votes from all element bundles combined.
\end{corollary}
\begin{proof}
Each possible type is associated with one element $x$. Lemmas~\ref{lem:auxaccess} and \ref{lem:auxincidence} exclude all other bundles. The one remaining bundle contains only $W$ unit voters, so it contributes at most $W$ votes, even if some of its leaves also run and its votes are no longer an indivisible block.
\end{proof}

\begin{lemma}\label{lem:auxscore}
In any active set containing $u$ and no baseline-leaf candidate, every participating auxiliary candidate $q\in K$ satisfies
\[
 \score_K(q)\le W+N_0<B\le\score_K(u).
\]
\end{lemma}
\begin{proof}
Lemma~\ref{lem:baseline}(ii) gives all baseline votes to $u$ and none to $q$. Corollary~\ref{cor:auxbundle} bounds all bundle votes to $q$ by $W$. There are exactly $N_0$ other voters. Even assigning all of them to $q$ yields at most $W+N_0$, and Lemma~\ref{lem:scalemargins} makes this strictly less than $B$. The estimate includes the auxiliary candidate's own voter in whichever of the three disjoint voter classes contains it; that vote is not added a second time.
\end{proof}

\begin{lemma}\label{thm:noaux}
For every active set $K$ containing $u$, if $u$ is eliminated in the next round, then
\[
 K\setminus\{u\}\subseteq\Areal.
\]
\end{lemma}
\begin{proof}
If a baseline leaf participates, Lemma~\ref{lem:baseline}(i) shows that $u$ cannot be eliminated.
Suppose no baseline leaf participates. Then Lemma~\ref{lem:baseline}(ii) gives $\score_K(u)\ge B$.

If any other auxiliary candidate $q$ participates,
Lemma~\ref{lem:auxscore} gives $\score_K(q)<B\le\score_K(u)$.
Thus $q$ has a smaller elimination key than $u$, so $u$ cannot be
eliminated in that round.

The auxiliary classification above contains every vertex other
than $u$ and the triple candidates. Therefore, if $u$ is eliminated,
every other participating candidate must be a triple candidate.
\end{proof}

\subsection{Unrestricted Exact-One-Apex Completeness}

\begin{lemma}\label{thm:apexunrestricted}
The problem
\[
 \Kill(G,u,V(G)\setminus\{u\})
\]
is $\NP$-complete on $\mathcal A_{5,10}$. In particular, hardness holds on exact one-apex graphs when the supplied planarizing vertex is the target, the target has smallest identifier, the graph is bipartite, its target eccentricity and radius are at most five, and its diameter is at most ten.
\end{lemma}
\begin{proof}
Membership in $\NP$ follows as in Lemma~\ref{thm:apexrestricted}. Guess a nonempty set of participating candidates and check in polynomial time whether it causes $u$ to lose. For hardness use the identical graph construction, but now allow every nontarget vertex.

If the source has an exact matching, its triple-candidate witness from Lemma~\ref{thm:apexrestricted} remains available in the unrestricted instance. Conversely, suppose an arbitrary allowed opponent set eventually eliminates $u$. Restart at the decisive round using Lemma~\ref{lem:restart}. The restarted election has at least one opponent and eliminates $u$ first. By Lemma~\ref{thm:noaux}, every other participating candidate in the restarted election is a triple candidate. It is therefore a first-round witness containing only triple
candidates, so the reverse direction of Lemma~\ref{thm:apexrestricted} yields an exact matching of the padded source.

Auxiliary candidates may have participated earlier in the original election; no claim about their full elimination order is needed. The theorem proves exactly that they cannot survive into a round that eliminates the target. This is why the first-round restart, rather than a speculative forced sequence of eliminations, suffices for the unrestricted reduction.
\end{proof}

\subsection{Exclusion-Zone Verification and the Minimum-Zone Function}

\begin{lemma}\label{thm:apexzones}
On $\mathcal A_{5,10}$, deciding whether the supplied singleton $\{u\}$ is an exclusion zone is $\coNP$-complete. General exclusion-zone verification on the same class is also $\coNP$-complete.
\end{lemma}
\begin{proof}
Corollary~\ref{cor:singleton} gives
\[
 \{u\}\text{ is not an exclusion zone}
 \quad\Longleftrightarrow\quad
 \Kill(G,u,V(G)\setminus\{u\})=\yes.
\]
The unrestricted reduction proves that, on its valid outputs, the singleton is a zone exactly for source no-instances. Thus the complement of Planar three-dimensional matching reduces to singleton verification. Since the source problem is $\NP$-complete, its complement is $\coNP$-complete, giving $\coNP$-hardness of singleton verification. Allowing arbitrary proposed $S$ includes the singleton case and preserves this lower bound.

For membership, the complement of the total verification language has a polynomial certificate in every case. Invalid graph or order encodings, failure of a class condition, or an empty proposed set can be detected deterministically in polynomial time. For a valid nonempty set that is not a zone, supply a candidate set $K$ with $K\cap S\ne\varnothing$ and $\IRV(G,K)\notin S$. The set uses $n$ bits and the winner is polynomially computable by Lemma~\ref{lem:evaluate}. Therefore the complement belongs to $\NP$ and verification belongs to $\coNP$. This proves both completeness statements with the invalid-input convention treated explicitly.
\end{proof}

\begin{lemma}\label{thm:apexminzone}
Computing the unique minimum exclusion zone on $\mathcal A_{5,10}$ is not in $\FP$ unless $\Pclass=\coNP$.
\end{lemma}
\begin{proof}
By Lemma~\ref{lem:nesting}, there is a unique exclusion zone $S^\star$ contained in every other exclusion zone. If $\{u\}$ is a zone, $S^\star\subseteq\{u\}$ and nonemptiness forces $S^\star=\{u\}$. Conversely, if $S^\star=\{u\}$, that singleton is a zone because $S^\star$ itself is one. Thus
\[
 \{u\}\text{ is a zone}\quad\Longleftrightarrow\quad S^\star=\{u\}.
\]
A deterministic polynomial-time algorithm outputting $S^\star$ would decide the $\coNP$-complete singleton-verification language after the polynomial validity check, simply by comparing its output with $\{u\}$. This would give $\coNP\subseteq\Pclass$. Since $\Pclass\subseteq\coNP$, it implies equality. The argument uses only the unique minimum-zone theorem and the unrestricted hardness, not the planar upper bound.
\end{proof}

The output here is an actual vertex set. The correct function-class notation is $\FP$; it is not appropriate to label this function problem itself ``$\NP$-complete.'' The collapse $\Pclass=\coNP$ is equivalent to $\Pclass=\NP$ because deterministic polynomial time is closed under complement, but the displayed formulation follows directly from the verification reduction.

\subsection{The Combined Hardness Theorem}
\begin{theorem}[Exact one-apex hardness]\label{m:apextheorem}
Consider connected, simple, unweighted and bipartite exact one-apex graphs with a supplied target $u$ having the smallest identifier, such that $G-u$ is planar, $\ecc_G(u)\le5$, and $\diam(G)\le10$. Even when every other vertex is an allowed candidate, so that $A=V\setminus\{u\}$, $\Kill(G,u,A)$ is $\NP$-complete. On the same graph class, deciding whether the singleton $\{u\}$ is an exclusion zone is $\coNP$-complete, and computing the minimum exclusion zone $S^\star$ is not in $\FP$ unless $\Pclass=\coNP$.
\end{theorem}
\begin{proof}
Lemma~\ref{thm:apexunrestricted} establishes unrestricted candidate-loss completeness under the conditions of Definition~\ref{def:apex}. Lemma~\ref{thm:apexzones} establishes singleton-zone verification completeness, and Lemma~\ref{thm:apexminzone} establishes the conditional minimum-zone function lower bound. These are exactly the three assertions of the theorem.
\end{proof}

\subsection{The Full Exact-One-Apex Classification}
\label{subsec:full-apex-classification}

The graph class of exact one-apex graphs is defined without reference to an election target: a graph belongs to it when it is nonplanar and some vertex deletion makes it planar. The hard instances constructed above satisfy the stronger marked-input condition that the supplied target $u$ itself is a planarizing vertex, so $G-u$ is planar.

Consequently, hardness on these instances implies hardness on the full exact-one-apex graph class: any algorithm for all exact-one-apex graphs and arbitrary supplied targets would also solve this hard subfamily. This implication does not classify the complementary marked subclass in which the supplied target is non-planarizing, that is, $G-u$ is nonplanar.

Theorems~\ref{m:treetheorem}, \ref{m:outertheorem}, \ref{m:planartheorem}, and~\ref{m:apextheorem} state the four graph-class conclusions of the paper. The exact-one-apex lower bounds already hold on $\mathcal A_{5,10}$; the loss lower bound has $A=V(G)\setminus\{u\}$ and verification hardness holds for $S=\{u\}$. To see that the decision classifications extend to the full exact-one-apex class, test nonplanarity of $G$ and planarity of $G-v$ for each $v$. Thus membership in this graph class is polynomially recognizable. Corollary~\ref{cor:membership} puts $\Kill$ in $\NP$ and zone verification in $\coNP$ on the full class. The hard subclass is contained in it. A polynomial-time minimum-zone algorithm for the full class would restrict to one on $\mathcal A_{5,10}$, giving the stated conditional function lower bound.

\section{Overview and Open Problems}\label{m:conclusion}

Under deterministic whole-vote graph-IRV, candidate loss, zone verification, and minimum-zone computation are polynomial-time solvable on planar graphs, but become hard on the stated bipartite exact one-apex class, even with every nontarget vertex eligible. The planar algorithm exploits the bounded-width cell quotient of a minimal witness, not bounded treewidth of the input graph.
We conclude with some open problems.

Can the outerplanar and planar exponents be reduced by an incremental approach analogous to the linear-time tree algorithm, or by more efficient separator--state saturation? Which larger graph classes admit polynomial-time algorithms, including bounded-genus graphs and proper minor-closed classes containing all planar graphs? Crossing-restricted beyond-planar classes offer further test cases~\cite{HongTokuyama}. Hardness transfers only to classes containing our actual hard instances under the same conventions; the presence of some apex graphs is insufficient. 
Finally,
which results persist under alternative tie conventions
or split votes?

\clearpage
\section{Notation and Theorem Index}\label{app:notation}
This is an index to definitions already given, not an additional set of conventions. Symbols with a local subscript retain the explicitly stated meaning in their algorithmic section.

\small
\begin{longtable}{@{}>{\raggedright\arraybackslash}p{.19\linewidth}>{\raggedright\arraybackslash}p{.555\linewidth}>{\raggedright\arraybackslash}p{.19\linewidth}@{}}
\toprule
Symbol & Meaning & Location\\
\midrule
\endfirsthead
\toprule
Symbol & Meaning & Location\\
\midrule
\endhead
\bottomrule
\endfoot
$G=(V,E)$; $n$ & The original voter graph and its number of vertices. Every vertex is one unit voter and a possible candidate. & Section~\ref{sec:model}\\
$d_G(v,c)$ & Unweighted shortest-path distance in the original graph. It never changes when a candidate is eliminated. & Section~\ref{sec:model}\\
$\id(v)$ & Common identifier rank, a bijection to $\{0,\ldots,n-1\}$. & Section~\ref{sec:model}\\
$\kappa_v(c)$ & Preference key $(d_G(v,c),\id(c))$, minimized by each voter. & Definition~\ref{def:scores}\\
$K$ & A nonempty active candidate set. & Definition~\ref{def:irv}\\
$\best_v(K)$ & The unique favorite of voter $v$ in $K$. & Definition~\ref{def:scores}\\
$\cell_K(c)$ & Deterministic nearest-candidate cell of $c$. & Definition~\ref{def:scores}\\
$\score_K(c)$ & Number of unit voters in that cell. & Definition~\ref{def:scores}\\
$\lambda_K(c)$ & Elimination key $(\score_K(c),-\id(c))$, minimized by the rule. & Definition~\ref{def:scores}\\
$\IRV(G,K)$ & The last remaining candidate after deterministic elimination. & Definition~\ref{def:irv}\\
$u$ & Target candidate in every graph class. & Definition~\ref{def:kill}\\
$A$; $D$ & Allowed opponents; eligible identities $D=A\cup\{u\}$. & Definition~\ref{def:kill}\\
$F$ & A chosen nonempty opponent set contained in $A$. & Definition~\ref{def:kill}\\
$\Kill(G,u,A)$ & Whether some such $F$ makes $u$ lose. & Definition~\ref{def:kill}\\
$S$; $S^\star$ & A proposed nonempty exclusion zone; the unique smallest zone. & Section~\ref{sec:model}\\
$L(G)$; $\cl(X)$ & Pairwise-loss tournament; outgoing reachability closure, including length-zero paths. & Definition~\ref{def:lossgraph}\\
$\bot$; $\infty$ & No-candidate sentinel; key larger than every finite elimination key. & Section~\ref{sec:model}\\
$\mu$ & Minimum finalized opponent key; $\infty$ for no such opponent. & Section~\ref{sec:model}\\
$T_x$ & Subtree containing $x$ and all its descendants. & Section~\ref{m:trees}\\
$R_q$; $b_q$ & Voters preferring $q$ to target $u$ in their pairwise election; root of this region. & Lemma~\ref{m:tree-regions}\\
$Q_t$; $U_t$ & Qualifying opponents at threshold $t$; union of their preference regions. & Lemma~\ref{m:tree-threshold}\\
$r_q$; $h_q$ & Region size $|R_q|$; largest tested threshold at which $q$ qualifies, with $h_q=0$ indicating that it never qualifies in the scan. & Section~\ref{m:trees}\\
$\mathcal T$; $X_t$ & Tree-decomposition tree; bag at its node $t$. & Definition~\ref{def:treedecomp}\\
$V_t$; $I_t$ & Union of bags at $t$ and its descendants; private vertices $V_t\setminus X_t$. & Definition~\ref{def:nice}\\
$\alpha_t$ & Best exterior candidate at each bag vertex. & Definition~\ref{def:context}\\
$\pi_t$; $\Pi_t$ & Best internal candidate at each bag vertex; set of exposed internal identities. & Definition~\ref{def:exposed}\\
$R_t$; $\nu_t$ & Union of visible exterior/internal representatives; their private-voter score map. & Equation~\eqref{eq:bagstate}\\
$\rho$ & Whether an opponent has been represented. & Equations~\eqref{eq:bagstate}, \eqref{eq:planarstate}\\
$Q_K$; $Q_K^+$ & Cell quotients in $G$ and its triangulated supergraph. & Section~\ref{sec:planar}\\
$G^+$ & Same-vertex triangulated plane supergraph used only for topology. & Lemma~\ref{lem:triangulate}\\
$T_c$; $T_K(c)$ & Fixed original shortest-path tree from $c$; its restriction to $c$'s true cell. & Lemma~\ref{lem:celltrees}\\
$H_K$ & Prediagram: dual after deleting cell-tree dual edges and pruning leaves. & Definition~\ref{def:prediagram}\\
$\phi_c$ & Prediagram face corresponding to candidate $c$. & Lemma~\ref{lem:directprediagram}\\
$\tw$; $\bw$ & Treewidth; branchwidth. & Definitions~\ref{def:treedecomp}, \ref{def:branch}\\
$b_0$ & Fixed separator-description bound, equal to six. & Corollary~\ref{lem:prediagramwidth}\\
$D_c$; $\widehat T_c$ & Closed face disk; restricted tree extended to its contact terminals. & Section~\ref{subsec:radial}\\
$J$; $h_f$ & Auxiliary plane graph with face-center spokes; center of triangle $f$. Neither changes the election. & Section~\ref{subsec:radial}\\
$Z$; $\bar Z$ & Oriented separator; the same boundary with the opposite side. & Definition~\ref{def:descriptor}\\
$\Gamma(Z)$ & Candidate guards listed by a separator. & Lemma~\ref{thm:lift}\\
$\Lambda(Z)$ & Original graph vertices on the separator curve. & Lemma~\ref{thm:lift}\\
$I(Z)$; $R(Z)$ & Strict-interior original vertices; $I(Z)\cup\Lambda(Z)$. Not the bag-visible set $R_t$. & Lemma~\ref{thm:lift}\\
$\Rep(Z)$ & Guard identities at $Z$ or at any descendant in its derivation; used in proofs, not stored in a state. & Lemma~\ref{lem:locations}\\
$E_P$; $Y_P$ & Immediate parent/child guard union; residual strict-parent-interior voters. & Section~\ref{sec:planar}\\
$\sigma$; $\theta$ & Strict-interior guard scores; target status outside/exposed/hidden with complete score. & Equation~\eqref{eq:planarstate}\\
$\apx(G)$ & Smallest number of vertex deletions making $G$ planar. & Definition~\ref{def:apex}\\
$\mathcal A_{5,10}$ & Hard exact-one-apex subclass with supplied smallest-ID target, bipartiteness and distance bounds. & Definition~\ref{def:apex}\\
$h$; $m$ & Size of each padded matching part; number of padded source triples. & Lemma~\ref{lem:padding}\\
$\mathcal U$; $C_i$ & Padded matching universe; its $i$th triple. & Section~\ref{sec:apex}\\
$e_x$; $c_i$ & Element-centre vertex; candidate vertex corresponding to triple $C_i$. & Section~\ref{sec:apex}\\
$\alpha_x,\beta_x$ & Private access-path vertices in $u-\alpha_x-\beta_x-e_x$. & Equation~\eqref{eq:accesspath}\\
$\gamma_{x,i}$ & Fresh incidence interior in $e_x-\gamma_{x,i}-c_i$. & Equation~\eqref{eq:incidencepath}\\
$G_0$; $N_0$ & Base graph before leaves; exact ordinary-voter count $1+9h+4m$. & Equation~\eqref{eq:N0}\\
$L_x$; $\mathcal B_u$ & Element bundle of leaves at $e_x$; baseline leaves at $u$. & Section~\ref{sec:apex}\\
$W$; $B$ & Number of leaves per element bundle, $4(N_0+1)$; baseline size, $5W/2$. & Equation~\eqref{eq:scales}\\
$\Areal$ & Set $\{c_1,\ldots,c_m\}$ of triple candidates;
different from unrestricted $A=V\setminus\{u\}$. &
Section~\ref{sec:apex}\\
\end{longtable}

\clearpage

\end{document}